\documentclass{article}
\usepackage{amsmath, amsthm, amssymb,color,fullpage,url,booktabs,enumitem}  
\usepackage{graphicx}
\usepackage{prettyref}
\newrefformat{cor}{Corollary~\ref{#1}}
\newrefformat{prop}{Proposition~\ref{#1}}

\newcommand{\MYstore}[2]{%
  \global\expandafter \def \csname MYMEMORY #1 \endcsname{#2}%
}

\newcommand{\MYload}[1]{%
  \csname MYMEMORY #1 \endcsname%
}

\newcommand{\MYnewlabel}[1]{%
  \newcommand\MYcurrentlabel{#1}%
  \MYoldlabel{#1}%
}

\newcommand{\MYdummylabel}[1]{}

\newcommand{\torestate}[1]{%
  \let\MYoldlabel\label%
  \let\label\MYnewlabel%
  #1%
  \MYstore{\MYcurrentlabel}{#1}%
  \let\label\MYoldlabel%
}

\newcommand{\restatetheorem}[1]{%
  \let\MYoldlabel\label
  \let\label\MYdummylabel
  \begin{theorem*}[Restatement of \prettyref{#1}]
    \MYload{#1}
  \end{theorem*}
  \let\label\MYoldlabel
}

\newcommand{\restatelemma}[1]{%
  \let\MYoldlabel\label
  \let\label\MYdummylabel
  \begin{lemma*}[Restatement of \prettyref{#1}]
    \MYload{#1}
  \end{lemma*}
  \let\label\MYoldlabel
}

\newcommand{\restateprop}[1]{%
  \let\MYoldlabel\label
  \let\label\MYdummylabel
  \begin{proposition*}[Restatement of \prettyref{#1}]
    \MYload{#1}
  \end{proposition*}
  \let\label\MYoldlabel
}

\newcommand{\restatefact}[1]{%
  \let\MYoldlabel\label
  \let\label\MYdummylabel
  \begin{fact*}[Restatement of \prettyref{#1}]
    \MYload{#1}
  \end{fact*}
  \let\label\MYoldlabel
}

\newcommand{\restatecorollary}[1]{%
  \let\MYoldlabel\label
  \let\label\MYdummylabel
  \begin{corollary*}[Restatement of \prettyref{#1}]
    \MYload{#1}
  \end{corollary*}
  \let\label\MYoldlabel
}

\newcommand{\restate}[1]{%
  \let\MYoldlabel\label
  \let\label\MYdummylabel
  \MYload{#1}
  \let\label\MYoldlabel
}

\ifdefined\blackandwhite
\usepackage[pdftex]{hyperref} 
\else
\usepackage[colorlinks,pdftex]{hyperref} 
\hypersetup{citecolor = blue, linkcolor = red!70!black}
\fi

\usepackage{braket}
\usepackage{mathtools}
\usepackage{mdframed}

\usepackage{tikz}
\usepackage{authblk}

\usepackage[ruled,vlined,linesnumbered]{algorithm2e}

\newcommand{\nc}{\newcommand}
\nc{\rnc}{\renewcommand}

\newcommand{\proj}[1]{\left|#1\right\rangle\left\langle #1\right|}

\newcommand{\Channel}{\text{Ch}}
\newcommand{\Coll}{\text{Coll}}

\newcommand{\Haar}{\text{Haar}}
\newcommand{\lattice}{\text{lattice}}

\newcommand{\tensor}{\otimes}
\newcommand{\ot}{\otimes}

\newcommand{\Id}{\mathrm{id}}
\newcommand{\C}{\mathbb{C}}
\newcommand{\R}{\mathbb{R}}

\newcommand{\dist}{\mathsf{dist}}
\newcommand*{\vv}[1]{\vec{\mkern0mu#1}}
\newcommand{\multiprod}{\mathsf{multiprod}}
\newcommand{\vvec}{\mathsf{vec}}

\DeclareMathOperator{\Bin}{Bin}
\DeclareMathOperator{\conv}{conv}
\DeclareMathOperator{\id}{id}

\DeclareMathOperator{\Pois}{\mathsf{Pois}}
\DeclareMathOperator{\poly}{poly}

\DeclareMathOperator{\rank}{rank}

\DeclareMathOperator{\Wg}{Wg}

\newcommand{\Tr}{\operatorname*{Tr}}

\DeclareMathOperator{\BPP}{\mathsf{BPP}}
\DeclareMathOperator{\FBPP}{\mathsf{FBPP}}
\renewcommand{\P}{\mathsf{P}}

\DeclareMathOperator{\NP}{\mathsf{NP}}

\DeclareMathOperator{\PH}{\mathsf{PH}}

\def\be#1\ee{\begin{equation}#1\end{equation}}
\def\bea#1\eea{\begin{eqnarray}#1\end{eqnarray}}
\def\beas#1\eeas{\begin{eqnarray*}#1\end{eqnarray*}}
\def\ba#1\ea{\begin{align}#1\end{align}}
\def\bas#1\eas{\begin{align*}#1\end{align*}}
\def\bpm#1\epm{\begin{pmatrix}#1\end{pmatrix}}
\nc{\non}{\nonumber}
\nc{\nn}{\nonumber}
\nc{\eq}[1]{(\ref{eq:#1})}
\nc{\eqs}[2]{(\ref{eq:#1}) and (\ref{eq:#2})}
\rnc{\L}{\left} 
\nc{\ra}{\rightarrow}
\nc{\grad}{{\vec{\nabla}}}
\newtheorem{theorem}{Theorem}
\newtheorem*{theorem*}{Theorem}
\newtheorem{remark}{Remark}
\newtheorem{conjecture}{Conjecture}
\newtheorem{lemma}[theorem]{Lemma}
\newtheorem*{lemma*}{Lemma}

\newtheorem{proposition}[theorem]{Proposition}
\newtheorem*{proposition*}{Proposition}
\newtheorem{definition}[theorem]{Definition}
\newtheorem{corollary}[theorem]{Corollary}
\newtheorem*{corollary*}{Corollary}

\newenvironment{claim}[1]{\par\noindent \textbf{Claim:}\space#1}{}

\newtheorem{fact}[theorem]{Fact}
\newtheorem{coupling}[theorem]{Coupling}
\newtheorem*{fact*}{Fact}

\nc\eps{\epsilon}

\nc\cA{\mathcal{A}}
\nc\cB{\mathcal{B}}
\nc\cC{\mathcal{C}}
\nc\cD{\mathcal{D}}
\nc\cE{\mathcal{E}}
\nc\cF{\mathcal{F}}
\nc\cG{\mathcal{G}}
\nc\cH{\mathcal{H}}
\nc\cI{\mathcal{I}}
\nc\cJ{\mathcal{J}}
\nc\cK{\mathcal{K}}
\nc\cL{\mathcal{L}}
\nc\cM{\mathcal{M}}
\nc\cN{\mathcal{N}}
\nc\cO{\mathcal{O}}
\nc\cP{\mathcal{P}}
\nc\cQ{\mathcal{Q}}
\nc\cR{\mathcal{R}}
\nc\cS{\mathcal{S}}
\nc\cT{\mathcal{T}}
\nc\cU{\mathcal{U}}
\nc\cV{\mathcal{V}}
\nc\cW{\mathcal{W}}
\nc\cX{\mathcal{X}}
\nc\cY{\mathcal{Y}}
\nc\cZ{\mathcal{Z}}

\DeclareMathOperator*{\E}{\mathbb{E}}
\DeclareMathOperator*{\bbE}{\mathbb{E}}
\nc\bbF{\mathbb{F}}
\nc\bbM{\mathbb{M}}
\nc\bbN{\mathbb{N}}
\nc\bbR{\mathbb{R}}
\nc\bbZ{\mathbb{Z}}

\nc\benum{\begin{enumerate}}
\nc\eenum{\end{enumerate}}
\nc\bit{\begin{itemize}}
\nc\eit{\end{itemize}}

\newcommand{\secref}[1]{Section~\ref{sec:#1}}

\newcommand{\lemref}[1]{Lemma~\ref{lem:#1}}
\newcommand{\thmref}[1]{Theorem~\ref{thm:#1}}
\newcommand{\propref}[1]{Proposition~\ref{prop:#1}}

\newcommand{\defref}[1]{Definition~\ref{def:#1}}

\nc{\todo}[1]{\textcolor{red}{todo: #1}}
\nc{\Anote}[1]{\textcolor{red}{Aram note: #1}}

\def\begsub#1#2\endsub{\begin{subequations}\label{eq:#1}\begin{align}#2\end{align}\end{subequations}}
\nc\qand{\qquad\text{and}\qquad}
\nc\mnb[1]{\medskip\noindent{\bf #1}}

\nc{\pder}[2]{\frac{\partial {#1}}{\partial {#2}}}
\nc{\p}{\partial}
\usepackage{wasysym}
\usepackage[pagewise]{lineno}

\definecolor{orange}{rgb}{1,0.5,0}
\newcommand{\newsentence}[1]{ { \color{blue}  (New sentence: #1)}}

\newcommand{\OldNormalFont}{}

\usepackage{hyperref,tabu, verbatim}
\usepackage{mdframed}
\hypersetup{%
  colorlinks = true,
  linkcolor  = black
}
\newcommand{\Hsquare}{%
  \text{\fboxsep=-.2pt\fbox{\rule{0pt}{1ex}\rule{1ex}{0pt}}}%
}

\title{
		\LARGE Approximate unitary $t$-designs by short random quantum circuits using nearest-neighbor and long-range gates. \\
}
\author{Aram W.~Harrow \thanks{MIT Center for Theoretical
    Physics. aram@mit.edu}\hspace{2cm} Saeed Mehraban \thanks{Tufts CS. Saeed.Mehraban@tufts.edu}}

\date{\today}

\begin{document}
\maketitle
\thispagestyle{empty}

We prove that $\poly(t) \cdot n^{1/D}$-depth local random quantum
circuits with two qudit nearest-neighbor gates on a $D$-dimensional
lattice with $n$ qudits are approximate $t$-designs in various
measures.  These include the ``monomial'' measure, meaning that the
monomials of a random circuit from this family have expectation close
to the value that would result from the Haar measure. Previously, the
best bound was $\poly(t)\cdot n$ due to Brand\~ao-Harrow-Horodecki
\cite{BHH-designs} for $D=1$.  We also improve the ``scrambling'' and
``decoupling'' bounds for spatially local random circuits due to Brown
and Fawzi~\cite{BF13-2}.

One consequence of our result is that assuming the polynomial
hierarchy ($\PH$) is infinite and that certain counting problems are
$\#\P$-hard ``on average'', sampling within total variation distance
from these circuits is hard for classical computers.  Previously,
exact sampling from the outputs of even constant-depth quantum
circuits was known to be hard for classical computers under these
assumptions. However the standard strategy for extending this hardness result to approximate
sampling requires the quantum circuits to have a property called
``anti-concentration'', meaning roughly that the output has
near-maximal entropy.  Unitary 2-designs have the desired
anti-concentration property.  Our result improves the required depth
for this level of anti-concentration from linear depth to a sub-linear
value, depending on the geometry of the interactions.  This is
relevant to a recent experiment by the Google Quantum AI group to
perform such a sampling task with 53 qubits on a two-dimensional
lattice \cite{Google19, BISBDJMN16} (and related experiments by USTC),
and confirms their conjecture that $O(\sqrt{n})$ depth suffices for
anti-concentration.

The proof is based on a previous construction of $t$-designs by \cite{BHH-designs}, an
analysis of how approximate designs behave under composition, and an extension of the
quasi-orthogonality of permutation operators developed by \cite{BHH-designs}.  Different
versions of the approximate design condition correspond to different norms, and part of
our contribution is to introduce the norm corresponding to anti-concentration and to
establish equivalence between these various norms for low-depth circuits.

For random circuits with long-range gates, we use different methods to show that
anti-concentration happens at circuit size $O(n\ln^2 n)$ corresponding to depth $O(\ln^3 n)$.  We also show a
lower bound of $\Omega(n \ln n)$ for the size of such circuit in this case. We also prove that
anti-concentration is possible in depth $O(\ln n \ln \ln n)$ (size
$O(n \ln n \ln \ln n)$) using a different model.
\newpage

\thispagestyle{empty}
\setcounter{tocdepth}{2}
\tableofcontents

\hypersetup{%
  colorlinks = true,
  linkcolor  = red!70!black
}
\newpage
\setcounter{page}{1}
\section{Introduction}
Random unitaries are central resources in quantum information science. They appear in many applications including algorithms, cryptography, and communication. Moreover, they are important toy models for random chaotic systems, capturing phenomena like thermalization or scrambling of quantum information. 

An idealized model of a random unitary is the uniform distribution
over the unitary group, also known as the Haar measure. However, the
Haar measure is an unrealistic model for large systems because the number of random
coins and gates needed to generate an element of the Haar distribution
scale exponentially with the size of the system (i.e.~polynomially
with dimension, meaning exponentially in the number of qubits or
independent degrees of freedom). To resolve this dilemma, approximate $t$-designs have been proposed as physically and computationally realistic alternatives to the Haar measure. They approximate the behavior of the Haar measure if one only cares about up to the first $t$ moments.

Several constructions of $t$-designs have been proposed based on either random or structured circuits.  While structured circuits can in some cases be more efficient~\cite{DCEL09, CLLW16, NHMW17}, random quantum circuits have other advantages.  They are plausible models for chaotic random processes in nature, such as scrambling in black holes~\cite{BF13-2, S14}, or the spread of entanglement in condensed matter systems \cite{NRVH16,NRVH17}, {growth of quantum complexity \cite{BCHJKP19}}, and decoupling \cite{BF13}. Moreover, they are practical candidates to benchmark computational advantage for quantum computation over classical models, since they seem to capture the power of a generic polynomial-size unitary circuit.  {Indeed, the Google quantum AI group has recently run a random unitary circuit on a 53-qubit superconducting device and has argued that this should be hard to simulate classically~\cite{Google19, BISBDJMN16} (see Figure \ref{fig:google} for a demonstration of their proposal)
.} Here the random gates are useful not only for the 2-design property, specifically ``anti-concentration'', but also for evading the sort of structure which would lend itself to easy simulation, such as being made of Clifford gates.

All previous random circuit based constructions for $t$-designs required the circuits to have linear depth. In this paper, we show that certain random circuit models with small depth are approximate $t$-designs. We consider two models of random circuits. The first one is nearest-neighbor local interactions on a $D$-dimensional lattice. 
 In this model, we apply random $\text{U}(d^2)$ gates on neighboring qudits of a
 $D$-dimensional lattice in a certain order.

Depending on the application we want, we can define convergence to the Haar measure in
different ways. For example, for scrambling \cite{BF13-2} we measure convergence w.r.t. the
norm $\E_C \|\rho_S(s) - \frac{1}{2^{|S|}} \|^2_1$, where $\rho_S(s)$ is the density
matrix $\rho(s)$ reduced to a subset $S$ of qudits and $\rho(s)$ is the quantum state
that results from $s$ steps of the random process. But for anti-concentration, which
corresponds loosely to a claim that typical circuit outputs have nearly maximal entropy,
we use a norm related to $\E_C\sum_x |\braket {x|C|0}|^4$. For other measures of convergence to the
Haar measure see \cite{L10} or \secref{norms}. In general, these measures are equivalent but moving between
them involves factors that are exponentially large in the number of qudits, i.e., if one
norm converges to $\eps$ the translation implies that another norm converges to $2^{O(n)}
\eps$.  Some of the known size/depth bounds for designs are of the form $O(f(n,t)(n + \ln
1/\eps))$ (e.g. \cite{BHH-designs}) and in 1-D simple arguments yield an $\Omega(n +
\ln(1/\eps))$ lower bound~\cite{BF13-2}. 
In this case, replacing $\eps$ with $2^{-O(n)}\eps$
will not change the asymptotic scaling. \cite{BHH-designs} defined a strong notion of convergence which implies all the mentioned definitions.

However, in $D$ dimensional lattices the natural lower bound is
$\Omega(n^{1/D}+\ln(1/\eps))$.  Our main challenge in this work is to show that this
depth bound is asymptotically achievable, and along the way, we need to deal with the fact
that we can no longer freely pay norm-conversion costs of $2^{O(n)}$.  We are able to
achieve the desired $\poly(t)(n^{1/D} + \ln(1/\eps))$ in many operationally relevant
norms, but due in part to the difficulty of converting between norms, we do not establish
it in all cases. The asymptotic dependency on $t$ in our result for $D =2$ is $O(t \ln t)$ times the best asymptotic dependency on 
$t$ for the $D=1$ architecture, according to the strong measure defined in \cite{BHH-designs}. 
\cite{BHH-designs} gave a bound of $t^{10.5}$. Recently this bound was improved to $t^{5 + o_t(1)}$ by Haferkamp \cite{haferkamp2022random}. 
The dependency on $t$ in our result is hence $t^{6 + o_t(1)} \ln t$.

\paragraph{Approximate unitary designs.} 
We will consider several notions of approximate designs in this paper.  First, we will
introduce some notation.  A degree-$(t,t)$ monomial in $C\in \text U((\C^d)^{\tensor n})$ is
degree $t$ in the entries of $C$ and degree $t$ in the entries of $C^*$.  We can collect
all these monomials into a single matrix of dimension $d^{2nt}$ by defining
$C^{\tensor t,t} := C^{\ot t} \ot C^{\ast \tensor t}$.  We say that $\mu$ is an exact
[unitary] $t$-design if expectations of all $t,t$ moments of $\mu$ match those of the Haar
measure.  We can express this succinctly in terms of the operator \be G_\mu^{(t)} = \E_{C
  \sim \mu}  \left [C^{\tensor t} \tensor C^{\ast \tensor t} \right].\ee Then $\mu$ is an exact
$t$-design iff $G_\mu^{(t)} = G_{\Haar}^{(t)}$.
Since $G_{\Haar}^{(t)}$ is a projector, we sometimes call $G_\mu^{(t)}$ a quasi-projector
operator and we will later use the fact that it can sometimes be shown to be very close
to a projector.

Most definitions of approximate designs demand that some norm of $G_\mu^{(t)} -
G_{\Haar}^{(t)}$ be small.  Three norms that we will consider are based on viewing
$G_{\mu}^{(t)}$ as either a vector of length $d^{4nt}$, a matrix of dimension $d^{2nt}$
or a quantum operation acting on a space of dimension $d^{nt}$.  In each case, one can show that the $t$-design property implies the $t'$-design property for
$1\leq t'\leq t$.

\begin{definition}[Monomial definition of $t$-designs]\OldNormalFont{} $\mu$ is a monomial-based
  $\eps$-approximate $t$-design if any monomial has expectation within $\eps d^{-nt}$ of
that resulting from the Haar measure.  In other words, 
\be
\left \|\vvec\left[G^{(t)}_\mu\right] - \vvec\left [G^{(t)}_\mu\right] \right\|_\infty \leq \frac{\eps}{d^{nt}}.
\ee
$\vvec(A)$ is a vector consisting of the elements of matrix $A$ (in the computational
basis) and $\|\cdot\|_\infty$ refers to the vector $\ell_\infty$ norm.
\label{def:monomialdesigns}
\end{definition}
The monomial measure is natural when studying anti-concentration, since a sufficient
condition for anti-concentration is that  $\E_C |\braket{0|C|0}|^4$ is close to the
quantity that arises from the Haar measure, namely $ \frac{2}{2^n(2^n+1)}$. This is
achieved by [monomial measure] $2$-designs.

If the operator-norm distance between $G_\mu^{(t)}$ and $G_\Haar^{(t)}$ is small then
instead of calling $\mu$ an approximate design we call it a $t$-tensor product
expander~\cite{HH09}.   This controls the rate at which certain nonlinear (i.e. degree-$t$
polynomial) functions of the state converge to the average value they would have under the
Haar measure. We can also measure the distance between $G_\mu^{(t)}$ and $G_\Haar^{(t)}$
in the 1-norm (i.e. trace norm) and this notion of approximate designs has been considered
before \cite{AS04, S05}, although it does not have direct operational meaning.  We
will show $\poly(t)(n^{1/D}+\ln(1/\eps))$-depth convergence in each of these measures.

Finally, we can consider $G_\mu^{(t)}$ to be a superoperator using the
following canonical map. Define $\Channel \left [\sum_i X_i \ot Y_i^T\right ]$ by
$ \Channel\left [\sum_i X_i \ot Y_i^T\right](Z) := \sum_i X_i Z Y_i$.
Thus
\be
\Channel\left [G^{(t)}_{\mu}\right] (Z) = \E_{C \sim \mu}  \left [ C^{\tensor t} Z  C^{\dagger \tensor t} \right ].
\ee
Note that $\Channel\left [G_{\mu}\right]$ is completely positive and trace preserving, i.e., a quantum
channel.   For superoperators $\cM,\cN$ we say that $\cM\preceq \cN$ if $\cN-\cM$ is a
completely positive (cp) map.  Based on this ordering, a strong notion of being an
approximate design was proposed by Andreas Winter and first appeared in ~\cite{BHH-designs}.  

\begin{definition}[Strong definition of $t$-designs]
\OldNormalFont{} A distribution $\mu$ is a strong $\eps$-approximate $t$-design if
\be
(1-\eps) \Channel\left [G^{(t)}_\Haar\right] \preceq \Channel\left[G^{(t)}_\mu\right] \preceq (1+\eps) \Channel \left[G^{(t)}_\Haar\right].
\ee
\label{def:strongdesigns}
\end{definition}

\paragraph{Circuit models.}
The result of \cite{BHH-designs} constructs $t$-designs in the strong measure (Definition
\ref{def:strongdesigns}) for $D=1$ and linear depth, and we generalize this result to construct weak monomial designs for arbitrary $D$ and $O(n^{1/D})$ depth. We also show that the
same construction converges to the Haar measure in other norms: diamond, infinity and
trace norm. Our proof techniques do not seem to yield $t$-designs in the strong
measure. We do not even know whether the construction of ``strong'' $t$-designs in
sub-linear depth is possible.  

The second model we consider is circuits with long-range two-qubit interactions. In this
model, at each step, we pick a pair of qubits uniformly at random and apply a random
$\text{U}(4)$ gate on them.  This model is the standard one when considering bounded-depth
circuit classes, such as $\mathsf{QNC}$.  Physically, it could model chaotic systems with
long-range interactions.  Following Oliveira, Dahlsten and Plenio~\cite{ODP06} (see also \cite{BF13-2, BF13, HL09}), we can map the $t=2$ moments
of this process onto a simple random walk on the points $\{1,2,\ldots,n\}$.  We map this
random walk to the classical (and exactly solvable) Ehrenfest model, meaning a 
random walk with a linear bias towards the origin.  Further challenges are that this
mapping introduces random and heavy-tailed delays and that the norm used for
anti-concentration is exponentially sensitive to some of the probabilities.  However, we
are able to show (in Section \ref{sec:all-to-all-anti-conc}) that after $O(n \ln^2 n)$ rounds of this process the resulting distribution
over the unitary group converges to the Haar measure in the mentioned norm.

For a distribution $p$ its collision probability is defined as
$\Coll(p) = \sum_x p_x^2$. If $\Coll (p)$ is large ($\Omega(1)$)
then the support of $p$ is concentrated around a constant number of outcomes, and if it is small
($\approx 1/2^n$) then it is anti-concentrated. The norm that we consider for
anti-concentration is basically the expected collision probability of the output
distribution of a random circuit. The expected 
collision probability for the Haar measure is $\frac{2}{2^n+1}$ and our result shows that a typical circuit of size $O(n \ln^2 n)$ outputs a distribution 
with expected collision probability $\frac{2}{2^n} \left(1+\frac 1 {\poly(n)}\right)$.
Along with the
Paley-Zygmund anti-concentration inequality this result proves that these circuits have
the following anti-concentration property: 
\be
\min_x \Pr_{C\sim \mu} \left [|\braket{x|C|0}|^2 \geq \frac{1}{2^{n+1}} \right] \geq \text{constant}.
\label{eq:zyg0}
\ee
Here $\mu$ is the distribution of random circuits we consider, and $x$ is any $n$-bit string. This bound is related to the hardness of classical simulation for random circuits. We furthermore show that sub-logarithmic depth
quantum circuits in this model have expected collision probability $\frac{2}{2^n+1}
\omega(1)$. The best anti-concentration depth bound we get from this model is $O(\ln^2 n)$. However, we are able to construct a natural family of random circuits with depth $O(\ln n \ln \ln n)$ that are anti-concentrated. 

\paragraph{The organization of this paper.}  The rest of this introductory section states the basic results, ideas and implications related to the main results. In particular, in Section \ref{Intro:supremacy} we explain the connections between our result and the result experiments performed by groups such as the Google AI group aiming towards demonstrating the superiority of quantum computing compared with classical computers on specific tasks. In Section \ref{Intro:models}, we describe the models we consider in this paper. In Section \ref{Intro:results}, we express the main results of this paper including proof sketches and basic ideas. We then give a brief overview of the previous works related to this paper in Section \ref{Intro:prev} and several open questions in Section \ref{Intro:open}.

The organization of the rest of this paper is as follows. In Section \ref{sec:prelim} we introduce the preliminary concepts, definitions and tools needed for our proofs. In Section \ref{sec:lattices-designs} we give detailed proofs about how we get approximate $t$-designs on $D$-dimensional lattices. In Section \ref{sec:all-to-all-anti-conc} we give detailed proofs related to anti-concentration bounds from circuits with all-to-all interactions. In Section \ref{sec:2Dmain} we provide alternative proofs for anti-concentration via low-depth $D$-dimensional lattices and in Section \ref{sec:scr} we provide improvements on the existing scrambling and decoupling bounds. Appendix \ref{sec:supremacy} gives a proof of Theorem \ref{thm:supremacy} about the implications of anti-concentration bound we obtain on computational difficulty of simulating low-depth random quantum circuits. Finally Appendix \ref{section:krawtchuk} gives a background about the basic properties of Krawtchouk polynomials which we use in Section \ref{sec:all-to-all-anti-conc}.

\subsection{Connections with quantum computational supremacy experiments}
\label{Intro:supremacy}

Outperforming classical computers, even for an artificial problem such as sampling from
the output of an ideal quantum circuit would be a significant milestone for early quantum
computers which has recently been called ``quantum computational supremacy'' \cite{HM17,P12}.  The reason to
study quantum computational supremacy in its own right (as opposed to general quantum algorithms) is
that it appears to be a distinctly easier task than full-scale quantum computing and even
various non-universal forms \cite{ABKM16, AC16, BISBDJMN16, BGK17, BJS10, FH16} of quantum
computing can be shown to be hard to simulate classically. For example, the outputs of
constant-depth quantum circuits cannot be simulated exactly by classical computers unless
the $\PH$ collapses \cite{TD02}.  In general, families of quantum circuits have this property
if they are universal under postselection, meaning that after measuring all the qubits at
the end of the circuit and producing a string of bits, we condition on the values of some
of these bits and use the other bits for the output.

However, these hardness results are not robust under noise and error
in measurements. A central open question in the theory of quantum
computational supremacy is whether simulating these distributions to
within constant or $1/\poly(n)$ variational distance would still be
hard. It is plausible to conjecture that if such a robust hardness of
sampling is true, it would also hold for generic circuits \cite{AC16,
  AA11} (although see \cite{2dsim20} for a counterexample). A standard
approach to proving such a robust hardness result for generic
circuits has been to prove that ``anti-concentration'' holds, and to use this to relate additive error approximation to average-case relative error approximation; see e.g. \cite{BMS16}.
Here ``anti-concentration'' means having near-maximal
entropy in the output of a quantum circuit, which implies that any fixed
amplitude of a quantum circuit is likely to be
$\geq \frac {\Omega(1)}{2^n}$.
 This property implies that the complexity of
estimating the amplitudes additively (within
$\pm \frac 1 {\poly(n) \cdot 2^n})$ is on-average as hard as computing
them within inverse polynomial relative error.  This lets us turn an
assumption about the average-case hardness of relative-error
approximation of the amplitudes into a hardness result for the
sampling problems.  Approximate $t$-designs (and even approximate
$2$-designs) have the desired ``anti-concentration'' property.

For experimental verification of quantum computational supremacy we can consider the following sampling task: let $\mu$ be a distribution over random circuits that satisfies
\be
\Pr_{C\sim \mu} \left[|\braket{0|C|0}|^2 \geq \frac{1}{2^{n+1}} \right] \geq 1/8 - 1/\poly(n).
\label{eq:zyg2}
\ee
(which we call the anti-concentration property). Let $\mathcal{C}_x$ be the family of circuits constructed by first applying a circuit $C \sim \mu$ and then an $X$ gate to each qubit with probability $1/2$ (and identity with probability $1/2$). A similar line of reasoning as in Bremner-Montanaro-Shepherd (see Theorem 6 and 7 of \cite{BMS16}) implies that
\begin{theorem} \OldNormalFont{}
Fix $\epsilon>0$ and $0<\delta<1/8$. Let $\mu$ be a $\frac 1{\poly(n)}$-approximate 2-design. If there exists a $\BPP$ machine which takes $C \sim \mu$ as input and for at least $1-\delta$ fraction of such inputs outputs a probability distribution that is within $\epsilon$ total variation distance from the probability distribution $p_x = |\braket {x | C| 0}|^2$, then 
 there exists an 
$\FBPP^{\NP}$
 algorithm that succeeds with probability $1-\delta$ and computes the value 
 $|\braket{0|C'|0}|^2$ 
  within multiplicative error $\frac{2(\epsilon + 1/{\poly(n)})}{\delta}$ for $1/8 -\frac1 {\poly(n)}$ fraction of circuits $C' \sim \mathcal{C}_x$.
\label{thm:supremacy}
\end{theorem}

This theorem is proved in Appendix \ref{sec:supremacy}.  If we further
conjecture the $\PH$ is infinite and that amplitudes of the random
circuits in \thmref{supremacy} are $\#\P$-hard to approximate on
average, then this implies that classical computers cannot efficiently
sample from any distribution close to the ones generated by these
circuits.  At the moment, it is only known that nearly exact
computation of these amplitudes is hard for $\# \P$ \cite{BFNV19,
  M19-1, M19-2}. It is an open question whether
average-case hardness for the approximation task remains
$\#\P$-hard.

The linear to sub-linear improvement of the depth required for anti-concentration provided
in this paper is likely to be significant for near-term quantum computers that will be
constrained both in terms of the number of qubits ($n$) and noise rate per gate ($\delta$).
Due to the constraints in the number of qubits (say 50-100), quantum computational supremacy will only be
possible without the overhead of error correction, since even the most efficient known
schemes for fault-tolerant quantum computation reduce the number of qubits by more than a
factor of two~\cite{ChaoR17}.  Thus a quantum circuit with $S$ gates will have an expected
$S\delta$ errors.   Recent work due to Yung and Gao~\cite{YG17} and the Google
group~\cite{BoixoSN17} states that noisy random quantum circuits with $O(\ln n)$ random
errors output distributions that are nearly uniform, and thus are trivially classically
simulable.  Thus $S$ can be at most $\ln(n)/\delta$.  In proposed near-term quantum
devices~\cite{Barends14,Debnath16,Ofek16,BISBDJMN16} we can expect $n\sim 10^2$ and
$\delta\sim 10^{-2}$.  Thus the $S=O(n \ln^2 n)$ for long-range interactions or
$ S = O(n \sqrt{n})$ bound for 2-D lattices from our work is much closer to being practical
than the previous $S=O(n^2)$.  (This assumes that the constants are reasonable.  We have
not made an effort to calculate them rigorously but for the case of long-range
interactions we do present a heuristic that suggests that in fact
$\approx \frac 56 n\ln n$ gates are necessary and sufficient.)

\subsection{Our models}
\label{Intro:models}
We consider two models of random quantum circuits. The first involves nearest-neighbor local interactions on a $D$-dimensional lattice and the second involves long-range random two-qubit gates. The order of gates in the first model has some structure but in the second model it is chosen at random. Hence, we can view the second model as the natural dynamics of an $n$-qubit system, connected as a complete graph.

We first define the following random circuit model for $D=1$ which was also considered in ~\cite{BHH-designs}: 

\begin{definition}[Random circuits on one-dimensional lattices]\OldNormalFont{}
$\mu^{\lattice, n}_{1, s}$ is the distribution over unitary circuits resulting from the
following random process.
\begin{mdframed}[nobreak=true]
 For $j = 1 : s$ \hfill \% for $t$-designs, view $s$ as $\poly(t) {n}$
\begin{itemize}
\item Apply independent random gates from $\text U(d^2)$ on qudits $(1,2), (3,4), \ldots , (n-1, n)$.
\item Apply independent random gates from $\text U(d^2)$ on qudits $(2,3), (4,5), \ldots , (n-2, n-1)$.
\end{itemize}

\end{mdframed}
\label{def:models2}
\end{definition}
This definition assumes that $n$ is even but we modify it in the obvious way when $n$ is
odd.  Another modification which would not change our results would be to put the qudits on a
ring so that sites $n$ and 1 are connected.

Building on this, we define the following distribution of random circuits on a two-dimensional lattice.
\begin{definition}[Random circuits on two-dimensional lattices] \OldNormalFont{}
  Consider a two-dimensional lattice with $n$ qudits. Let $r_{\alpha,i}$ be the
  i$^{\text{th}}$ row of the lattice in direction $\alpha \in \{1,2\}$, for
  $1 \leq i \leq \sqrt n$.  For each $\alpha\in\{1,2\}$ let $\text{SampleAllRows}(\alpha)$
  denote the following procedure (see Figure \ref{fig:ourmodels}):
\begin{mdframed}[nobreak=true]  
  For each $i\in [\sqrt n]$, sample a random circuit from
  $\mu^{\lattice, \sqrt{n}}_{1,s}$ and apply it to $r_{\alpha,i}$.
\end{mdframed}

Now define $\mu^{\lattice, n}_{2, c, s}$ to be the distribution over unitary
  circuits resulting from the following random process:
\begin{mdframed}[nobreak=true]  
\begin{itemize}
\item Repeat these steps $c$ times: apply SampleAllRows(1) and then SampleAllRows(2).
\item Apply SampleAllRows(1) a final time.
\end{itemize}
\end{mdframed}
\label{def:modelsD}
\end{definition}

This distribution has depth $(2c+1)2s$ and is related but not identical to the Google AI
group's {experiment~\cite{Google19, BISBDJMN16}}, see Figure \ref{fig:google}.  For our results on
$t$-designs, we will take $c$ to be $\poly(t)$ and $s$ to be $\poly(t) \cdot \sqrt{n}$.
We believe that our result can be extended to any natural family of circuits with
nearest-neighbor interactions.  We also assume for convenience that $\sqrt{n}$ is an
integer, but believe that this assumption is not fundamentally necessary.

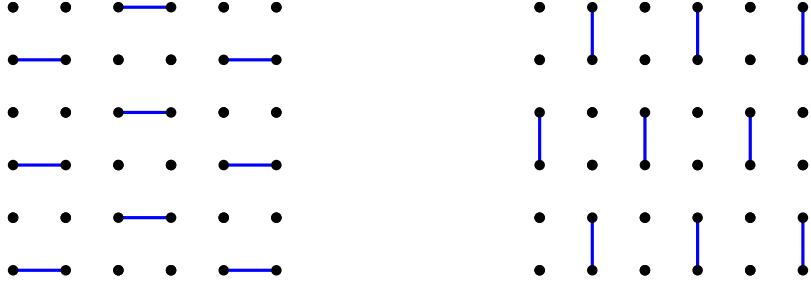
\begin{figure}[h]
\begin{center}
\begin{tikzpicture}[very thick,scale=0.7]
\draw[blue] (2,1) -- (3,1); \draw[blue] (2,3) -- (3,3); \draw[blue] (2,5) -- (3,5);
\draw[blue] (0,0) -- (1,0); \draw[blue] (0,2) -- (1,2); \draw[blue] (0,4) -- (1,4);
\draw[blue] (4,0) -- (5,0); \draw[blue] (4,2) -- (5,2); \draw[blue] (4,4) -- (5,4);
\foreach \x in {0,...,5}{
\foreach \y in {0,...,5}{
\fill (\x,\y) circle (0.1);
}}
\foreach \x in {0,1,4,5}{
\foreach \y in {1,3,5}{
\fill[black] (\x,\y) circle (0.1);
}}
\foreach \y in {0,2,4}{
\fill[black] (2,\y) circle (0.1);
\fill[black] (3,\y) circle (0.1);
}
\begin{scope}[xshift=15cm]
\begin{scope}[rotate=90]
\draw[blue] (2,1) -- (3,1); \draw[blue] (2,3) -- (3,3); \draw[blue] (2,5) -- (3,5);
\draw[blue] (0,0) -- (1,0); \draw[blue] (0,2) -- (1,2); \draw[blue] (0,4) -- (1,4);
\draw[blue] (4,0) -- (5,0); \draw[blue] (4,2) -- (5,2); \draw[blue] (4,4) -- (5,4);
\foreach \x in {0,...,5}{
\foreach \y in {0,...,5}{
\fill (\x,\y) circle (0.1);
}}
\foreach \x in {0,1,4,5}{
\foreach \y in {1,3,5}{
\fill[black] (\x,\y) circle (0.1);
}}
\foreach \y in {0,2,4}{
\fill[black] (2,\y) circle (0.1);
\fill[black] (3,\y) circle (0.1);
}
\end{scope}
\end{scope}
\end{tikzpicture}
\end{center}
\caption{ The architecture proposed by the quantum AI group at Google to
  demonstrate quantum supremacy consists of a 2D lattice of
  superconducting qubits.  This figure depicts two illustrative
  timesteps in this proposal.  At each timestep, 2-qubit gates (blue)
  are applied across some pairs of neighboring qubits.
\label{fig:google}}
\end{figure}

Next, we give a recursive definition for our random circuits model on arbitrary
$D$-dimensional lattices. We view a $D$-dimensional lattice as a collection of $n^{1/D}$
sub-lattices of size $n^{1-1/D}$, labeled as $\xi_{1}, \ldots, \xi_{n^{1-1/D}}$. We label
the rows of the lattice in the $D$-th direction by $r_1,\ldots, r_{n^{1/D}}$. 

\begin{definition}[Random circuits on $D$-dimensional lattices] \OldNormalFont{}$\mu^{\lattice,
    n}_{D,c,s}$ is the distribution resulting from the following random process. 
\begin{mdframed}[nobreak=true]
\begin{enumerate}
\item Repeat these steps $c$ times.
\begin{enumerate}
\item For each $i\in [n^{1/D}]$,
\begin{itemize}
\item Sample a random circuit from $\mu^{\lattice, n^{1-1/D}}_{D-1,c,s}$ and apply it to $\xi_i$.
\end{itemize}

\item For each $j \in [n^{1-1/D}]$
\begin{itemize}
\item Sample a random circuit from $\mu^{\lattice, n^{1/D}}_{1,s}$ and apply it to $r_j$.
\end{itemize}
\end{enumerate}
\item For each $i\in [n^{1/D}]$,
\begin{enumerate}
\item Sample a random circuit from $\mu^{\lattice, n^{1-1/D}}_{D-1,c,s}$ and apply it to $\xi_i$.
\end{enumerate}
\end{enumerate}
\end{mdframed}
\end{definition}

Next, we define the model with long-range interactions on a complete graph.
\begin{definition}[Random circuit models on complete graphs] \OldNormalFont{}
$\mu^{\text{CG}}_s$ is the distribution over unitary circuits resulting from the
following random process.
\begin{mdframed}[nobreak=true]
Repeat this step $s$ times \hfill \% view $s$ as $O(n \ln^2 n)$.
\begin{itemize}
\item Pick a random pair of qudits $(i,j)$ and apply a random $\text{U}(d^2)$ gate between them.
\end{itemize}

\end{mdframed}
\label{def:modelcg}
\end{definition}

The size of the circuits in this ensemble is $s$.

\subsection{Our results}
\label{Intro:results}

Our first result is the following.

\begin{theorem} \torestate{\OldNormalFont{}
Let $s, c,n > 0$ be positive integers
with  $\mu^{\lattice,n}_{2,c,s}$ defined as in \defref{modelsD}.
\begin{enumerate}
\item $s = \poly(t)\left (\sqrt n + \ln \frac{1}{\delta} \right), c = O\left(t \ln t + \frac{\ln (1/\delta)}{\sqrt  n}\right) \implies \left \| \vvec\left[G_{\mu^{\lattice,n}_{2,c,s}}^{(t)}-G_\Haar^{(t)}\right] \right\|_\infty \leq \frac{\delta}{d^{nt}}$.
\label{p1}
\item $s = \poly(t) \left(\sqrt n + \ln \frac{1}{\delta}\right ), c = O\left (t \ln t + \frac{\ln (1/\delta)}{\sqrt  n}\right) \implies \left \| \Channel\left[G_{\mu^{\lattice,n}_{2,c,s}}^{(t)}-G_\Haar^{(t)}\right]\right \|_\diamond \leq \delta$.
\item $s = \poly(t) \left(\sqrt n + \ln \frac{1}{\delta} \right), c = O\left(t \ln t + \frac{\ln (1/\delta)}{\sqrt n}\right) \implies \left \| G_{\mu^{\lattice,n}_{2,c,s}}^{(t)} - G_\Haar^{(t)} \right\|_1 \leq \delta$. 
\item $\left \| G_{\mu^{\lattice,n}_{2,c,s}}^{(t)} - G_\Haar^{(t)} \right \|_\infty \leq c \cdot \sqrt n \cdot e^{-s/\poly(t)} + \frac{1}{d^{O( c \sqrt n )}}$.
\end{enumerate}
\label{thm:grid}
}
\end{theorem}
\begin{figure}
\begin{center}
\begin{tikzpicture}[thick,scale=0.4]
\foreach \y in {0,..., 5}{
\draw[blue] (0,\y) -- (1,\y); \draw[blue] (2,\y) -- (3,\y); \draw[blue] (4,\y) -- (5,\y); 
}
\foreach \x in {0,...,5}{
\foreach \y in {0,...,5}{
\fill[black] (\x,\y) circle (0.1);
}}

\foreach \x in {0,1,4,5}{
\foreach \y in {1,3,5}{
\fill[black] (\x,\y) circle (0.1);
}}
\foreach \y in {0,2,4}{
\fill[black] (2,\y) circle (0.1);
\fill[black] (3,\y) circle (0.1);
}

\node[inner sep=0,anchor=east,text width=1cm] (note1) at (-3,2) {
   (1) };

\begin{scope}[xshift=20cm]

\foreach \y in {0,..., 5}{
\draw[blue] (1,\y) -- (2,\y); \draw[blue] (3,\y) -- (4,\y); 
}
\foreach \x in {0,...,5}{
\foreach \y in {0,...,5}{
\fill[black] (\x,\y) circle (0.1);
}}

\foreach \x in {0,1,4,5}{
\foreach \y in {1,3,5}{
\fill[black] (\x,\y) circle (0.1);
}}
\foreach \y in {0,2,4}{
\fill[black] (2,\y) circle (0.1);
\fill[black] (3,\y) circle (0.1);
}

\node[inner sep=0,anchor=east,text width=1cm] (note1) at (-3,2) {
   (2) };

\end{scope}
\end{tikzpicture}
\end{center}
\end{figure}

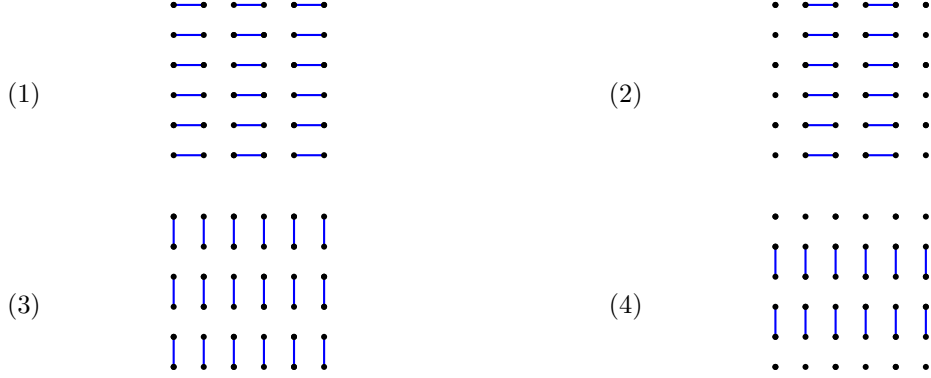
\begin{figure}
\begin{center}
\begin{tikzpicture}[thick,scale=0.4]
\begin{scope}[rotate=90]

\foreach \y in {0,..., 5}{
\draw[blue] (0,\y) -- (1,\y); \draw[blue] (2,\y) -- (3,\y); \draw[blue] (4,\y) -- (5,\y); 
}

\foreach \x in {0,...,5}{
\foreach \y in {0,...,5}{
\fill [black] (\x,\y) circle (0.1);
}}

\foreach \x in {0,1,4,5}{
\foreach \y in {1,3,5}{
\fill[black] (\x,\y) circle (0.1);
}}
\foreach \y in {0,2,4}{
\fill[black] (2,\y) circle (0.1);
\fill[black] (3,\y) circle (0.1);
}

\node[inner sep=0,anchor=east,text width=1cm] (note1) at (2,8) {
   (3) };

\end{scope}

\begin{scope}[xshift=20cm]
\begin{scope}[rotate=90]

\foreach \y in {0,..., 5}{
\draw[blue] (1,\y) -- (2,\y); \draw[blue] (3,\y) -- (4,\y); 
}
\foreach \x in {0,...,5}{
\foreach \y in {0,...,5}{
\fill[black] (\x,\y) circle (0.1);
}}

\foreach \x in {0,1,4,5}{
\foreach \y in {1,3,5}{
\fill[black] (\x,\y) circle (0.1);
}}
\foreach \y in {0,2,4}{
\fill[black] (2,\y) circle (0.1);
\fill[black] (3,\y) circle (0.1);
}

\node[inner sep=0,anchor=east,text width=1cm] (note1) at (2,8) {
   (4) };

\end{scope}
\end{scope}
\end{tikzpicture}
\end{center}
\caption{ The random circuit model in definition \ref{def:modelsD}. Each black circle is a
  qudit and each blue link is a random $\text{SU}(d^2)$ gate. The model does $O(\sqrt n\poly(t))$
  rounds alternating between applying (1) and (2). Then for $O(\sqrt n\poly(t))$ rounds it
  alternates between (3)
  and then (4).  This entire loop is then repeated $O(\poly(t))$ times.
\label{fig:ourmodels}}
\end{figure}
The three norms in the above theorem refer to the vector $\ell_\infty$ norm, the superoperator diamond norm
$\|\cdot\|_\diamond$ (see \secref{definitions}) and the operator
$S_\infty$ norm, also known simply as the operator norm.

\begin{proof} [\OldNormalFont{} Proof sketch for part \ref{p1}]
 We first give a brief overview of the proof in \cite {BHH-designs} and explain why
 their construction requires a circuit to have linear depth. Let $G_{i,i+1}$ be the
 projector operator for a random two-qudit gate applied to qudits $i$ and $i+1$, and let
 $G = \frac{1}{n-1} \sum_i G_{i,i+1}$. Therefore $G_s = G^s$ is the quasi-projector
 corresponding to a 1-D random circuit with size $s$.  \cite{BHH-designs} observed that
 $G -  G_{\Haar}$ corresponds to a certain local Hamiltonian and $\eps = 1- \|G -
 G_{\Haar}\|_\infty$ is its spectral gap. The central technical result of \cite{BHH-designs}
 is the bound $\eps \geq  \frac{1}{n\cdot \poly(t)}$. As a result, $\|G_s - G_{\Haar}\|_\infty = (1 -
 \frac 1{n\cdot\poly(t)})^s$. In general $G -  G_{\Haar}$ has rank $e^{O(n)}$, and in order to construct a
 strong approximate $t$-design (Definition \ref{def:strongdesigns}), one needs to apply a
 sequence of expensive changes of norm that lose factors polynomial in the overall
 dimension of $G$, i.e.,
 $e^{O(nt)}$. Thereby in order to compensate for such exponentially large factors one needs
 to choose $s = O(n^2\cdot \poly(t))$, meaning depth growing linearly with $n$. Brown and Fawzi \cite{BF13-2} furthermore
 observed that if $G$ is the projector corresponding to one step of a random circuit on a
 2-D lattice, the spectral gap still remains $1-\|G - G_{\Haar}\|_\infty = O(\frac 1{n\cdot \poly(t)})$, and
 using the same proof strategy one needs linear depth.

The new ingredient we contribute is to show that if
 $s = O(\sqrt n)$ one can replace $G^{(t)}_{\mu^{\lattice,n}_{2,1,s}}$ with a
 certain quasi-projector $G'$, such that 
\begin{itemize}
 \item[(1)] $G'-G_{\Haar}$ has rank $t!^{O(\sqrt{n})}$ and
\item [(2)] $\|G' -G_{\Haar}\|_\infty \approx 1/e^{\Omega(\sqrt{n})}$,
 \item[(3)] $G^{(t)}_{\mu^{\lattice,n}_{2,1,s}} \approx G'$ in various norms.
 \end{itemize}
 
 We first use (1) to relate the monomials definition of $t$-designs to the infinity norm and then use (2) to bound the infinity norm
 \be
 \left \| \vvec\left [G^{(t)}_{\mu^{\lattice,n}_{2,c,s}}\right] - \vvec\left[G^{(t)}_{\Haar}\right]\right\|_\infty \approx t!^{O(\sqrt{n})} \left \|G'^c -G^{(t)}_{\Haar}\right\|_\infty \cdot \frac{t!}{d^{nt}} \approx \frac {t!^{O(\sqrt n)}}{e^{\Omega(c \cdot \sqrt{n})}} \cdot \frac{1}{d^{nt}}.
 \ee 
For $c = t \ln t$ the error bound is $1/e^{\Omega(\sqrt n)} \frac 1 {d^{nt}}$. As a result using (3)
\be
\left\| \vvec\left[G^{(t)}_{\mu^{\lattice,n}_{2,c,s}}\right] - \vvec\left [G^{(t)}_{\Haar}\right]\right \|_\infty \approx \left \| \vvec\left[G'^c\Big] - \vvec\Big[G^{(t)}_{\Haar}\right]\right\|_\infty \approx 1/e^{\Omega(\sqrt{n})} \cdot \frac{1}{d^{nt}}.
\ee
 This step requires a certain change of norm for which we only have to pay a factor like
 $e^{O(\sqrt{n})}$, which we justify by bounding the ranks of the right intermediate
 operators.  The factor of $1/d^{nt}$ comes from the fact that the Haar measure itself has
 monomial expectation values on this order (in fact as large as $t!/d^{nt}$ but we
 suppressing the $t$-dependence in this proof sketch.)

We now briefly describe the construction of $G'$.
Let $G_R$ (and $G_C$) be the projector operators corresponding to applying a Haar unitary
to each row (and column) independently. Then $G' = G_R G_C$. $G'$ has rank $t!^{O(\sqrt n)}$ because $G_R$ and $G_C$ are each tensor products of $\sqrt n$ Haar projectors each with rank $t!$. 
Let $V_R$, $V_C$, and
$V_\Haar$ be respectively the subspaces that $G_R$, $G_C$ and $G_{\Haar}$ project onto. In
order to prove (1) in Section \ref{overlap}
 we first use the fact that our circuits are computationally universal to argue that $V_C \cap V_R = V_\Haar$. We then prove that the angle between 
$V_R \cap V_\Haar^\perp$ and $V_C \cap V_\Haar^\perp$ 
is very close to $\pi/2$, i.e., $\approx \pi/2 \pm \frac{1}{d^{\sqrt{n}}}$.
 This implies that $G_C G_R = G_{\Haar} + P$,
 where $P$ is a small matrix in the sense that $\|P\|_{\infty} \approx
 1/d^{\sqrt{n}}$. Choosing $c = \poly(t)$ we obtain (1). To show (2) it is not hard to see
 that the rank of $G' - G_{\Haar}$ is indeed $e^{O(\sqrt n)}$. For (3) we use the
 construction of $t$-designs from \cite{BHH-designs}. In particular, our random
 circuits model first applies an $O(\sqrt{n})$ depth circuit to each row and then
 an $O(\sqrt{n})$ depth circuit to each column and repeats this for $\poly(t)$ rounds. The result
 \cite{BHH-designs} implies that each of these rounds is effectively the same as
 applying a strong approximate $t$-design to the rows or columns of the lattice. We then
 analyze how these designs behave under composition in various norms and prove (3). 
\end{proof}


Our second result generalizes Theorem \ref{thm:grid} to arbitrary dimensions.

\begin{theorem} \torestate{\OldNormalFont{}
There exists a value $\delta = 1/d^{\Omega(n^{1/D})}$ such that for some large enough $c$ depending on $D$ and $t$:
\begin{enumerate}
\item $s > c \cdot n^{1/D} \implies \left \| \vvec\left[G_{\mu^{\lattice,n}_{D,c,s}}^{(t)}-G_\Haar^{(t)}\right] \right\|_\infty \leq \frac{\delta}{d^{nt}}$.
\item $s > c \cdot n^{1/D} \implies \left\| \Channel\left[G_{\mu^{\lattice,n}_{D,c,s}}^{(t)}-G_\Haar^{(t)}\right] \right\|_\diamond \leq \delta$.
\item $s > c \cdot n^{1/D} \implies \left\| G_{\mu^{\lattice,n}_{D,c,s}}^{(t)} - G_\Haar^{(t)} \right\|_\infty \leq \delta$.
\item $s > c \cdot n^{1/D} \implies \left\| G_{\mu^{\lattice,n}_{D,c,s}}^{(t)} - G_\Haar^{(t)} \right\|_1 \leq \delta$. 
\end{enumerate}
\label{thm:lattice}}
\end{theorem}

In order to understand the implication of this result for anti-concentration, let's first define 
\begin{definition}[Anti-concentration] We say a family of circuits $\mu$ satisfy the $(\alpha,\beta)$ anti-concentration property if for any $x \in \{0,1\}^n$
\be
\Pr_{U \sim \mu} (|\braket{x|U|0}|^2 \geq \frac \alpha{2^n}) \geq \beta
\ee
\label{def:anti-concentration}
\end{definition}
As mentioned before, unitary $2$-designs imply strong anti-concentration bound. In particular
\begin{theorem}
Let $\mu$ be a an $\epsilon$-approximate $2$-design in the monomial measure. Then $\mu$ satisfies the $(\alpha, \beta)$ anti-concentration property for $\alpha = \delta (1-\eps)$, $\beta = \frac{(1-\delta)^2 (1-\eps)^2}{2 (1+\eps)}$ and $0 \leq \delta \leq 1$.
\label{thm:anti-concentration}
\end{theorem}

\begin{proof}  (See Appendix \ref{sec:supremacy} and also Theorem 5 of \cite{HBSE18}). The proof is based on the Paley-Zigmond anti-concentration inequality: for a non-negative random variable $X$ and $\delta >0$ we have
\be
\Pr (X \geq \delta \cdot \E X) \geq (1-\delta)^2 \frac {\E(X)^2}{\E(X^2)}.
\label{eq:PZ}
\ee

\end{proof}
We remark that based on the result of \cite{DHJB20}, while sufficient, the $2$-design property is not necessary for anti-concentration. In Section \ref{sec:2Dmain} we give an alternative proof for anti-concentration in $O(D \cdot n^{\frac 1D})$ depth based on different ideas. The method directly implies anti-concentration and not the approximate $2$-design property.

For these spatially local circuits we also improve on some bounds in \cite {BF13-2} and
\cite{BF13} about scrambling and decoupling, removing polylogarithmic factors. 
Here we give an informal statement of the result with full details and definitions found
in Section \ref{sec:scr}. 
\begin{theorem}[Informal]
Random quantum circuits acting on $D$-dimensional lattices composed of $n$ qubits are scramblers and decoupler in the sense of \cite {BF13-2} and \cite{BF13} after $O(D \cdot n^{1/D})$ number of steps.
\end{theorem}

Our last result concerns the fully connected model.  If $s = O(n \ln^2 n)$ and $d=2$ then $\mu^{\text{CG}}_{s}$ satisfies the
anti-concentration criterion according to Definition \ref{def:anti-concentration} for constant $\alpha$ and $\beta$, i.e., \eqref{eq:zyg0}. We phrase our result in terms
of the expected ``collision probability'' of the output distribution of $C \sim
\mu^{\text{CG}}_s$ from which a bound similar to the one in theorem \ref{thm:anti-concentration} will follow using the Paley-Zygmond
inequality \eqref{eq:PZ}. 
In particular, if $C$ is a quantum circuit on $n$ qubits, starting from $\ket{0^n}$  the collision probability is  
\be
\Coll(C) := \sum_{x \in \{0,1\}^n}  |\braket{x|C|0} |^4.
\ee
For the Haar measure $\E_{C \sim \mathrm{ Haar}}\Coll(C) = \frac{2}{2^n+1}$,
and for the uniform distribution this value is $1/2^n$. In contrast, a depth-$1$ random
circuit has expected collision probability $ (\sqrt{\frac{2}{5}} )^n$, which is
exponentially larger than what we expect from the Haar measure.

\begin{theorem}\torestate{\OldNormalFont{}
There exists a  $c$ such that when $s > c n \ln^2 n$,
\be
\E_{C \sim \mu^{\text{CG}}_s} \Coll(C) \leq \frac{29}{2^n}.
\ee
Moreover if $t \leq \frac 1 {3 c'}  n \ln n$ for some large enough  $c'$, then 
\bea
\E_{C \sim \mu^{\text{CG}}_s} \Coll(C) \geq \frac {1.6 ^{n^{1-1/c'}}} {2^n}.
\eea
\label{thm:cg}}
\end{theorem}

\begin{proof}[\OldNormalFont{} Proof Sketch]
  For the upper bound, we translate the convergence time of the
  expected collision probability to the mixing time of a certain
  classical Markov chain (which we call $X_0,X_1,\ldots$). This Markov chain 
  has also been considered in previous work \cite{ODP06,HL09,BF13}. Part of our contribution is to analyze this Markov chain in a new norm.
  The Markov chain has $n$
  sites labeled as $1,\ldots,n$, and at each site $x$ it will move
  only to $x-1$, $x$ or $x+1$.  Such chains are known as ``birth and
  death'' chains, and in our case it results from representing the
  state of the system by a Pauli operator and then taking $x$ to be
  the Hamming weight of that Pauli operator.  It is known~\cite{ODP06}
  that the probability of moving to site $x+1$ is
  $\approx \frac 65 \frac{x(n-x)}{n^2}$ and the probability of moving
  to site $x-1$ is $\approx\frac 25 \frac{x(x-1)}{n^2}$.  The major
  difficulty in proving mixing for this Markov chain is that the norm
  which we have to prove mixing in is exponentially sensitive to small
  fluctuations (measured in either the 1-norm or the 2-norm).  Indeed, given starting condition 
\be
\Pr[X_0 = k] = \frac{{n \choose k}}{2^n-1}.
\ee
we would like to show that
 \be
\E_C[\Coll(C)] \approx  \sum_{k=1}^n\frac{\Pr\left[X_t = k\right]}{3^k},
\label{inio}
\ee
is  $\leq O(2^{-n})$.  We can think of \eqref{inio} as a
weighted 1-norm on probability distributions.

Our proof will compute the 
distribution of $X_t$ for
$t = O(n \ln^2 n)$ nearly exactly.  One distinctive feature of this
chain is that when $k/n \ll 1$, the probability of moving is $O(k/n)$
and the chain is strongly biased to move towards the right. When $k/n$
reaches $O(n)$, the chain becomes more like the standard discrete
Ehrenfest chain, which is a random walk with a linear bias towards (in
this case) $k=\frac 34 n$.  Thus the small-$k$ region needs to be
handled separately.  This is especially true for anti-concentration
thanks to the $1/3^k$ term in \eqref{inio}, so that even a small probability of waiting
for a long time in this region can have a large effect on the collision probability.

The approach of \cite{HL09,DJ10,BF13} has been to relate the original Markov chain to an
``accelerated'' chain which is conditioned on moving at each step.  The status of the
original chain can be recovered from the accelerated chain by adding a geometrically
distributed ``wait time'' at each step.  Then standard tools from the analysis of Markov
chains, such as comparison theorems and log-Sobolev inequalities, can be used to bound the
convergence rate of the accelerated chain.  Finally, it can be related back to the original
chain by arguing that the accelerated chain is unlikely to spend too long on small values
of $k$, allowing us to bound the wait time.  For our purposes, this process does not
produce sharp enough bounds, due to the heavy-tailed wait times combined with fairly weak
bounds on how quickly the accelerated chain converges and leaves the small-$k$ region.

We will sharpen this approach by incompletely accelerating; i.e.,
we will couple the original chain to a chain that moves with a
carefully chosen (but always $\Omega(1)$) probability.  In particular,
we will introduce a chain where the probabilities of moving from $x$
to $x-1$, $x$ or $x+1$ are each affine functions of $x$.  In fact our
new ``accelerated'' chain is only accelerated for $x< \frac 56n$
and is actually more likely to stand still for $x\geq \frac 56n$.
This will
allow us to exactly solve for the probability distribution of the
accelerated chain after any number of steps, using a method of Kac to relate this
distribution to the solution of a differential equation.   Our solution can be expressed
simply in terms of Krawtchouk polynomials, which have appeared in other exact solutions to
random processes on the hypercube.
  We relate
this back to the original chain with careful estimates of the mean and
large-deviation properties of the wait time.  This ends up showing
only that the collision probability is small for $t$ in some interval
$[t_1,t_2]$, and to show that it is small for a specific time, we need
to prove that the collision probability decreases monotonically when
we start in the state $\ket{0^n}$.  A further subtlety is that
\eqref{inio} technically only applies when all qubits have been hit by
gates and we need to extend this analysis to include the
non-negligible probability that some qubits have never been acted on
by a gate.

Because previous work achieved quantitatively less sharp bounds, they could omit some
of these steps.  For example, \cite{DJ10, HL09} used $O(n^2)$ gates, which meant
that the probability of most bad events was exponentially small.  By
contrast, in depth $O(n\ln^2(n))$, there is probability $n^{-O(\ln
  n)}$ of missing at least one qubit and so we cannot afford to let
this be an additive correction to our target collision probability of
$\text{constant} \cdot 2^{-n}$.  Likewise, \cite{BF13} used only $O(n
\ln^2(n))$ gates but achieved a collision probability of
$2^{\eps n-n}$ for small constant $\eps$, which allowed them to use a simpler version of
the accelerated chain whose convergence they bounded using generic
tools from the theory of Markov chains.

For the lower bound we just consider the event that the initial Hamming
 weight does not change throughout the process. 
The initial state with Hamming
weight $k$ has probability mass $\Pr[X_0=k]=\frac{{n\choose k}}{2^n-1}$.  Starting with Hamming weight $k$, the probability
of not moving in each step is $e^{-O(k/n)}$, so if  $t= c n \ln n$ for $c \ll 1$ then
we have $\Pr[X_t=k | X_0=k] \geq e^{-O(k t/n)}$. Hence 
\be
\E_{C \sim \mu_t} \Coll(C) \geq \sum_{k = 1}^n \frac{{n \choose k}}{2^n-1}\frac{\Pr[X_t=k | X_0=k]}{3^k} \geq\sum_{k = 1}^n \frac{{n \choose k}}{2^n-1}\frac{e^{-O(kt/n)}}{3^k} \approx \frac 1 {2^n}(1 + e^{- 3 t/n})^n \geq \frac {2^{n^{1-O(1)}}}{2^n}
\ee

\end{proof}

A natural question is whether there is a common generalization of our Theorems
\ref{thm:lattice} and \ref{thm:cg}.
In physics, the $D\ra\infty$ limit is often considered a good proxy
for the fully connected model. This raises the question of whether we
needed \thmref{cg} to handle the fully connected case, or whether it
would be enough to use \thmref{lattice} in the large $D$ limit.
However, Theorem \ref{thm:lattice} works only for $D = O(\ln n/ \ln \ln
n)$, and the best depth bound we can get from this theorem is
$e^{O(\ln n/\ln \ln n)}$, which is far above the $O(\ln^2(n))$
achievable by \thmref{lattice}.  
However, in Section \ref{sec:2Dmain} we give an alternative proof for anti-concentration
of outputs via circuits on $D$-dimensional circuits with $t=2$ and $D = O(\ln n)$. Using that approach we can make the depth as small as $O(\ln n \ln \ln n)$. We conjecture that $O(\ln n)$ depth should also possible.

In order to establish rigorous bounds, our results involve some inequalities that are not
always tight. As a result, the upper bound on collision probability in \thmref{cg} has a
factor of 29 rather than the $2+o(1)$ that we would expect and the bound on the number of
gates required may be too high by a factor of $\ln(n)$.  Since determining the precise
number of gates needed for anti-concentration may have utility in near-term quantum
hardware, we also undertake a heuristic analysis of what depth seems to be required to
achieve anti-concentration.  Here we ignore the possibility of large fluctuations in the
wait time, for example, and simply set it equal to its expected value.  We also freely
make the continuum approximation for the biased random walk that ignores wait time,
obtaining the Ornstein-Uhlenbeck process.  The resulting analysis (found in
\secref{constants}) suggests that $\frac 56 n \ln n + o(n \ln n)$ gates are needed to
achieve anti-concentration comparable to the Haar measure.

This result can also be useful for understanding the near-term power
  of certain variational quantum algorithms, such as VQE and QAOA.
\cite{MBSBN18, CSVCC20} show that when a gate sequence is
 drawn from a 2-design, the gradients used for optimizing VQE and
 other algorithms become exponentially small.    This is called the
 ``barren plateau'' phenomenon.  Our result would suggest that this
 occurs in 2-D circuits once the depth is $\gtrsim \sqrt n$.

\subsection{Previous work}
\label{Intro:prev}

The time evolution of the 2nd moments of random quantum circuits was
first studied by Oliveira, Dahlsten and Plenio~\cite{ODP06}, who
investigated their entanglement properties.  This was extended by
\cite{HL09,DJ10} to show that after linear depth, arndom circuits on
the complete graph yield approximate 2-designs.  In \cite{BHH-designs}
Brand\~ao-Harrow-Horodecki (BHH) extended this result and showed that
for a $1D$-lattice after depth
$t^{10.5} \cdot O(n + \ln \frac{1}{\epsilon})$ these random quantum
circuits become $\eps$-approximate $t$-designs.  This result was
subsequently improved to
$t^{5 + o_t(1)} O(n + \ln \frac{1}{\epsilon})$ by Haferkamp
\cite{haferkamp2022random}. All of these results (except \cite{ODP06}) directly imply
anti-concentration after the mentioned depths.  The construction of
$t$-designs in \cite{BHH-designs} is in a stronger measure than the
one in HL \cite{HL09}. The gap of the second-moment operator was
calculated exactly for $D=1$ and fully connected circuits by
{\v{Z}}nidari{\v{c}} \cite{Z08} and a heuristic estimate for the
$t^{\text{th}}$ moment operator was given by Brown and Viola for fully
connected circuits~\cite{BV10}.

In \cite{BF13,BF13-2} Brown and Fawzi considered ``scrambling'' and ``decoupling'' with random quantum circuits. In particular, they showed for a $D$-dimensional lattice scrambling occurs in depth $O(n^{1/D} \operatorname{polylog}(n))$, and for complete graphs, they showed that after polylogarithmic depth these circuits demonstrate both decoupling and scrambling. For the case of $D$-dimensional lattices they showed that for the Markov chain $K$, after depth $n^{1/D} \operatorname*{polylog}(n)$, a string of Hamming weight $1$ gets mapped to a string with linear Hamming weight with probability $1-1/\operatorname*{poly}(n)$. 
While this result is related to ours, it does not seem to yield the results we need
e.g.~for anti-concentration, due to the powers of Hilbert space dimension that are lost when
changing norms.

In \cite{NRVH17,NRVH16} Nahum, Ruhman, Vijay and Haah considered operator spreading for random quantum
circuits on $D$-dimensional lattices. They considered the case when a single Pauli operator
starts from a certain point on the lattice and they analyze the probability that after a certain time
a non-identity Pauli operator appears at an arbitrary point on the lattice. For $D=1$ they
showed that this probability function satisfies a biased diffusion equation. Their result in this case is exact. For $D=2$ they explained, both numerically and theoretically, that this
probability function spreads as an almost circular wave whose front satisfies the one
dimensional Kardar-Parisi-Zhang equation. They moreover explained: 1) the bulk of the occupied region is in
equilibrium, 2) fluctuations appear at the boundary of this region with  $\sim
t^{1/3}$, and 3) the area of the occupied region grows like $t^2$, where $t$ is the depth
of the circuit. As far as we understand this result does not directly lead to the
construction of $t$-designs and rigorous bounds on the quality of the approximations made
in that paper are not known.

If we assume that qudits have infinite local dimension ($d  \rightarrow \infty$) then the
evolution of  Pauli strings
on a 2-D lattice is closely related to Eden's model \cite{E61}. Here, Eden has
found certain explicit solutions. However, apart from the $d\ra\infty$ limit, his model
differs from ours also in that his considers only starting with a single occupied site and
running for a time much less than the graph diameter (or equivalently, considering an
infinitely large 2-D lattice), while we consider the initial distribution obtained by
starting in the $\ket{0^n}$ state.

After the first preprint version of this paper was posted online, \cite{DHJB20} improved on our results in several ways. Unlike what we expected, they proved that random quantum circuits acting on linear chains or complete graphs anti-concentrated after depth $\Theta (\ln n)$. It is left as an important open question whether the same bound holds for $D = 2,3, \ldots$. They also proved one of the conjectures of this paper that the constant factor for the depth bound for the complete graph model is $5/6$. The initial presentation of this result had a mistake in the heuristic reasoning and predicted the constant factor to be $5/3$. This was pointed out and corrected in \cite{DHJB20}.


\subsection{Open questions}
\label{Intro:open}
\begin{enumerate}
\item Is it possible to construct ``strong'' $t$-designs (Definition
  \ref{def:strongdesigns}) using sub-linear depth random circuits? If we can show that the
  off-diagonal moments (see Definition \ref{def:offdiagonal}) of the distribution, which
  have expectation zero according to the Haar measure, become smaller than $1/d^{3nt}$ in
  sub-linear depth, then our construction of monomial designs implies the construction of
  strong designs. On the other hand, we cannot rule out the possibility that strong designs require linear depth.

\item How large are the constant factors in bounds reported in this
  paper? Based on a heuristic argument in Section \ref{sec:constants}
  for the complete graph architecture we conjecture that such random
  circuits of size $s = \frac{5}{6} n (\ln n + \epsilon)$ are
  $O(\eps)$-approximate $2$-designs. See Conjecture
  \ref{conj:exact-constant} for a precise conjectured bound for obtaining $2$-designs.
 In work appearing after the first version of our paper, Ref.~\cite{DHJB20} proved this
  conjecture for anti-concentration. Our result had achieved an
  upper-bound of $O(n \ln^2 n)$.

\item We believe our dependence on $n$ is essentially optimal. But our
  depth scales with $t$ as $t^\alpha$ for some $\alpha\gtrsim 5$ that is almost
  certainly not optimal.  At the moment the
  best lower bound is $\Omega(t\ln n)$ depth for any circuit, or
  $\Omega(n^{1/D})$ in $D$ dimensions.  Indeed, very recently
  \cite{HJ19} provided strong analytical evidence that for the
  one-dimensional architecture, $\alpha=1$ for $D=1$.  The argument, however,
  contains uncontrolled approximations and is not known to extend even
  to $D=2$, although such an extension seems plausible. Intriguingly, also for constant $n$ and with a
  different gate model, some results are known that are completely independent of $t$~\cite{BG11}.

\item If we pick an arbitrary graph and apply random gates on the edges of this graph,
  after what depth do these circuits become $t$-designs? We conjecture that if the graph
  has large expansion and diameter $l$, then the answer is $O(l)$. However, if the graph
  has a tight bottleneck (like a binary tree), then even though the graph has small
  diameter, we suspect that certain measures of $t$-designs (including the monomial
  measure) require linear depth.  Ideally, the $t$-design time for any graph could be
  related to other properties of the graph such as mixing time, cover time, etc.

\item Can we prove a comparison lemma for random circuits, i.e., can we show that if two
  random circuits are close to each other, then they become $t$-designs after roughly the
  same amount of time? Such comparison lemma may imply that other natural families of
  low-depth circuits are approximate $t$-designs.  A related question is whether deleting
  random gates from a circuit family can ever speed up convergence to being a $t$-design.
  Such a bound has been called a   ``censoring'' inequality in the Markov-chain literature.

\item Our results do not say much about the actual constants that appear in the asymptotic
  bounds for the required size for anti-concentration. We conjecture the leading term in
  the anti-concentration time for random circuits on complete graphs is $\frac 56 n \ln
  n$.   

For the $D$-dimensional case our bounds inherit constant factors from \cite{BHH-designs}. Simple numerical simulation and also the analysis of \cite{NRVH17,NRVH16, BISBDJMN16} suggest that the constant should be $\approx 1$.

\item For the case of $D$-dimensional circuits, our result does not say much about the dynamics of the distribution when depth is $\ll n^{1/D}$. Such a result may explain the dynamics of entanglement in random circuits. \cite{NRVH17,NRVH16} consider this problem for the case when a single Pauli operator starts at the middle of the lattice; however, their result does not apply to arbitrary initial operators.

\item The best anti-concentration lower bound we are able to prove is $\Omega (\ln n)$ For $D$-dimensional lattices one would expect a lower-bound of $\Omega(n^{1/D})$ based on the following intuition for circuits of depth $s <
n^{1/D}$:  For $s \ll n^{1/D}$, we expect any two non-overlapping clusters of $s^D$ qubits will be close to Haar random. Hence, a crude model for such circuits would be $n/s^D$
copies of Haar-random unitaries each on $s^D$ qubits.  In this case we would expect the collision
probability to be $ \approx \frac{2^{n/s^D}}{2^n}$. Very
interestingly, the recent result \cite{DHJB20} refutes this intuition
for $D=1$ and showed an upper bound of $O(\ln n)$ for the depth at
which anti-concentration is achieved. It seems plausible that at $D =
2,3, \ldots$ we would also have anti-concentration in depth $O(\ln n)$
since it holds both for $D=1$ and for fully connected circuits.

\end{enumerate}


\section{Preliminaries}
\label{sec:prelim}
\subsection {Basic definitions}
\label{sec:definitions}

We need the following norms:
\begin{definition} \OldNormalFont{}
For a superoperator
$\mathcal{E}$  the diamond norm~\cite{KSV02} is defined as
$\|\mathcal{E}\|_\diamond := \sup_d \|\mathcal{E} \tensor \Id_d\|_{1  \rightarrow 1}$,
where for a superoperator $A$ the $1  \rightarrow 1$ norm is defined as$
\|A\|_{1  \rightarrow 1} := \sup_{X \neq 0} \frac{\|A (X)\|_1}{\|X\|_1}$.
\label{def:norms}
\end{definition}

A matrix is called positive semi-definite (psd) if it is Hermitian and has all non-negative eigenvalues. A superoperator $\mathcal A$ is called completely positive (cp) if for any $d \geq 0$, $\mathcal A \tensor \text{id}_d$ maps psd matrices to psd matrices. A superoperator is called trace-preserving completely positive (tpcp) if it maps if it preserves the trace and is furthermore cp.

Let $S$ be a set of qudits, then
\begin{definition} \OldNormalFont{} $\Haar (S)$ is the Haar measure on $U ((\C^d)^{\tensor |S|} )$. We refer to $\Haar(i,j)$ as the two qudit Haar measure on qudits indexed by $i$ and $j$ and also if $m$ is an integer, the notation $\Haar(m)$ means Haar measure on $m$ qudits.
\label{def:haar}
\end{definition}

We now define expected monomials, moment superoperators and quasi-projectors for a distribution $\mu$ over the unitary group:
\begin{definition} \OldNormalFont{}
Let $n,t >0$ be positive integers and $\mu$ be any distribution over $n$-qudit unitary group $\text U((\C^d)^{\tensor n})$. Then $G_\mu^{(t)} := \E_{C\sim \mu}  \left[C^{\tensor t,t} \right ]$ is the quasi-projector of $\mu$. Here $C^{\tensor t,t} = C^{\tensor t} \otimes C^{\ast \tensor t}$. Also $G^{(t)}_{(i,j)} = G^{(t)}_{\Haar(i,j)}$.
Using this Definition we will also use the following quantities:
\begin{enumerate}
\item Let $i_1, j_1, \ldots, i_t, j_t, k_1, l_1, \ldots, k_t, l_t \in [d]^n$ be any $2t$-tuple of words $\in [d]^n$. Then the $i_1,\ldots,l_t$ monomial 
is the expected value of a balanced monomial of $\mu$ defined as
\be
 \E_{C \sim \mu}  \left [ C_{i_1, j_1} \ldots C_{i_t, j_t} C^\ast_{k_1, l_1} \ldots C^\ast_{k_t l_t} \right ] = \bra{i_1, \ldots,  j_t}G_\mu^{(t)}\ket{k_1, \ldots,  l_t}
\ee
$C_{a,b}$ is the $a, b$ entry of the unitary matrix $C$.


\item Let $\mathrm{ad}_X (\cdot) := X (\cdot) X^\dagger$. Then $\Channel\left[G^{(t)}_\mu\right ] := \E_{C \sim \mu} \left [\mathrm{ ad}_{C^{\tensor t}} \right ]$ is the $t^{\text{th}}$ moment superoperator of $\mu$.
\end{enumerate}
\label{def:moments}
\end{definition}


Next, we define the building blocks of our $t$-design constructions. 

\begin{definition}[Rows of a lattice] \OldNormalFont{}
For $1 \leq i \leq
n^{1-1/D}$, $r_{\alpha,i}$ is the $i$-th row of a $D$-dimensional lattice in the
$\alpha$-th direction.  We will label the qubits in row $i$ by
$(\alpha,i,1),\ldots,(\alpha,i,n^{1/D})$.  
Assume for convenience that $n^{1/D}$
is an  even integer and define the sets of pairs
$E_{\alpha,i} := \{((\alpha,i,1),(\alpha,i,2)),\ldots,
((\alpha,i,n^{1/D}-1), (\alpha,i,n^{1/D}))\}$ and
$O_{\alpha,i} := \{((\alpha,i,2),(\alpha,i,3)),\ldots,
((\alpha,i,n^{1/D}-2), (\alpha,i,n^{1/D}-1))\}$.
\label{def:rows}
\end{definition}

\begin{definition} [Elementary random circuits]\OldNormalFont{}
The elementary quasi-projector in direction $\alpha$ is
\be
g_{\text{Rows} (\alpha, n)} :=\prod_{1\leq l \leq n^{1-1/D}} \bigotimes_{(i,j) \in E_{\alpha,l}} G^{(t)}_{(i,j)} \cdot 
 \bigotimes_{(i,j) \in O_{\alpha,l}} G^{(t)}_{(i,j)}=: \prod_{1\leq l \leq n^{1-1/D}} g_{r_{\alpha,l}} .
\ee
For the 2-D lattice
$g_R$ and $g_C$ for $g_1$ and $g_2$, respectively.
\label{def:elementarycircuit}
\end{definition}

The following defines the moment superoperator and quasi-projector of the Haar measure on the rows of a $D$-dimensional lattice in a specific direction. 

\begin{definition} [Idealized model with Haar projectors on rows]
Let $1\leq \alpha \leq D$ be one of the directions of a $D$-dimensional lattice then
\be
G_{\text{Rows}(\alpha,n)} := \prod_{ 1 \leq i \leq n^{1-1/D}} G^{(t)}_{\Haar(r_{\alpha,i})}=:\prod_{ 1 \leq i \leq n^{1-1/D}} G_{r_{\alpha,i}}.
\ee
For a 2-D lattice we use
$G_R$ and $G_C$ for $G_1$ and $G_2$, respectively.
\label{def:mm}
\end{definition}

 Next, we define moment operators and projectors corresponding to the Haar measure on the sub-lattices of a $D$-dimensional lattice. We view a $D$-dimensional lattice as a collection of $n^{1/D}$ smaller lattices each with dimension $D-1$, composed of $n^{1-1/D}$ qudits. We label these sub-lattices with $\text {Planes}(D) := \{p_1, \ldots, p_{n^{1/D}}\}$.

\begin{definition}[Haar measure on sub-lattices] \OldNormalFont{}
$G_{\text{Planes}(D)} = \bigotimes_{p \in \text{Planes}(D)} G^{(t)}_{\Haar(p)}\equiv G^{(t)\tensor n^{1/D}}_{\Haar(n^{1-1/D})}$,.
\label{def:planes}
\end{definition}

\begin{definition}\OldNormalFont{}
For $d=2$, $t=2$ and a superoperator $\cA$ define 
\be
\Coll (\cA) := \Tr   \left( \sum_{x \in \{0,1\}^n} \ket{x}\bra{x} \tensor \ket{x}\bra{x} \cA  (\ket{0^n}\bra{0^n} \tensor \ket{0^n}\bra{0^n} )  \right).
\label{rho}
\ee
\label{def:coll}
\end{definition}

In particular, for a distribution $\mu$ over circuits of size $s$ the expected collision probability is defined as
\be
\Coll_s := \Coll\left (\Channel \left [G^{(2)}_{\mu}\right ]\right ).
\ee

\begin{remark}\OldNormalFont{}
For $d=2$, $t=2$ and when $\nu$ is the Haar measure on $\text{U}(4)$, $\Channel\left [G^{(2)}_{(i,j)}\right]$ is the following map in the Pauli basis:
\be
\Channel \left [G^{(2)}_{(i,j)}\right] (\sigma_p \tensor \sigma_q) = 
\begin{cases}
\sigma_0 \tensor \sigma_0 & p q = 00\\
\frac{1}{15}\sum_{s\in \{0,1,2,3\}^2 \backslash 0} \sigma_s \tensor \sigma_s& p = q \neq 00\\
0 & \text{otherwise}
\end{cases}
\ee
More generally, if $S$ is a collection of qubits, and $p, q \in \{0,1,2,3\}^S$, then
\be
\Channel\left [G^{(2)}_S\right] (\sigma_p \tensor \sigma_q) = 
\begin{cases}
\sigma_0 \tensor \sigma_0 & p q = 00\\
\frac{1}{4^{|S|}-1}\sum_{s\in \{0,1,2,3\}^{|S|} \backslash 0} \sigma_s \tensor \sigma_s& p = q \neq 00\\
0 & \text{otherwise}
\end{cases}
\ee
when $p,q \in  \{0,1,2,3 \}^S$.
\end{remark}
See \cite{ODP06,HL09} for the proof of these remarks.

\subsection{Operator definitions of the models}

\begin{definition}[Random circuits on a two-dimensional lattice]
\OldNormalFont{} The 
quasi-projector of $\mu^{\lattice,n}_{2,c,s}$ 
is $G_{{\mu^{\lattice,n}_{2,c,s}}}^{(t)} = g^s_R  (g^s_C g^s_{R} )^c$.
\label{definition:model}
\end{definition}

The generalization of this definition to arbitrary $D$ dimensions is according to:
\noindent
\begin{definition} \OldNormalFont{}
[Recursive definition for random circuits on $D$-dimensional lattices] $ $
The
quasi-projector of $\mu^{\lattice,n}_{D,c,s}$ is specified by the recursive formula: 
\be
G^{(t)}_{ \mu^{\lattice,n}_{D,c,s}} = G^{(t) \tensor n^{1/D}}_{\mu^{\lattice,n^{1-1/D}}_{D-1,c,s}}  \left(g_{\text{Rows}(D,n)}^s G^{(t)\tensor n^{1/D}}_{\mu^{\lattice,n^{1-1/D}}_{D-1,c,s}} \right)^c.
\ee 
It will be useful to our proofs to also define:
\begin{enumerate}
\item $\tilde{G}_{n,D,c} =  \left (\tilde{G}^{\tensor n^{1/D}}_{n^{1-1/D}, D-1 , c} G_{\text{Rows}(D,n)} \tilde{G}^{\tensor n^{1/D}}_{n^{1-1/D}, D-1 , c} \right)^c$
\item ${\hat {G}}_{n,D,c,s} = G_{\text{Rows}(D,n)} \tilde{G}^{\tensor n^{1/D}}_{n^{1-1/D}, D-1, c,s} G_{\text{Rows}(D,n)}$
\end{enumerate}  
\label{def:recursive}
\end{definition}
In particular, $\tilde{G}_{n,D,c,s}$ is the same as $G_{\mu^{\lattice,n}_{D,c,s}}$ except that we have replaced $g_{\text{Rows}(D,n)}^s$ with $G_{\text{Rows}(D,n)}$.
Definition \ref{definition:model} is a special case of Definition \ref{def:recursive}, but we included both of them for convenience.

\begin{definition}\OldNormalFont{}
$G^{(t)}_{\mu^{\text{CG}}_{s}} =  \left(\frac{1}{{n \choose 2}} \sum_{i \neq j} G^{(t)}_{(i,j)}
 \right)^s$.
 \label{def:cg}
\end{definition}

\subsubsection{Summary of the definitions}

 See below for a summary of the definitions:
\begin{center}
\begin{tabular} 
{  l l l}
\toprule
\textbf{Notation} & \textbf{Definition} & \textbf{Reference}\\
\midrule
$\|\cdot \|_\diamond$ & superoperator diamond norm& Definition \ref{def:norms}\\
\hline
$\|\cdot \|_p$ & matrix $p$-norm for $p \in [0, \infty]$&Definition \ref{def:norms}\\
\hline
$\Haar$ & the Haar measure& Definition \ref{def:haar}\\ 
\hline
$\Haar(S)$ & Haar measure on subset $S$ of qudits & Definition \ref{def:haar}\\
\hline
$\Haar(i,j)$ & Haar measure on qudits $i$ and $j$ & Definition \ref{def:haar}\\
\hline
$U^{\ot t,t}$ & $C^{\otimes t} \ot C^{\ast, \otimes t}$ & Definition \ref{def:moments}\\
\hline
  $G^{(t)}_\mu$ & average of $C^{\ot t,t}$ over $C\sim \mu$ &Definition \ref{def:moments}\\
\hline
$G^{(t)}_\Haar$ & Projects onto vectors invariant under $C^{\ot t,t}$&Definition \ref{def:moments}\\
\hline
$G^{(t)}_{i,j}$ & Haar projector of order $t$ on qudit $i$ and $j$&Definition \ref{def:moments}\\
\hline
$\braket{i,j| G^{(t)}_\mu|k,l}$ & moment of order $t$: $\E_{C\sim \mu} [C_{i_1, j_1} \ldots C_{i_t, j_t} C^\ast_{i_1, j_1} \ldots C^\ast_{i_t, j_t}]$&Definition \ref{def:moments}\\
\hline
$\Channel[G^{(t)}_\mu]$ & moment superoperator,  equal to $\E_{C \sim \mu}  [\mathrm{ ad}_{C^{\tensor t}} ]$&Definition \ref{def:moments}\\
\hline
$r_{\alpha,i}$ & $i$-th row in the $\alpha$ direction with
                 $i\in[n^{1/D}], \alpha\in[D]$& Definition \ref{def:rows}\\
\hline 
$\text{Rows}(\alpha,n)$ & the collection of rows of a lattice (with $n$ points) in the $\alpha$ direction & Definition \ref{def:rows}\\
\hline
$g_{\text{Rows}(\alpha,n)}$ &  \parbox[t]{0.6\textwidth}{two-qudit gates applied to even then
                              odd neighbors in each row in the $\alpha$ direction\vspace{1.5mm}} &Definition \ref{def:elementarycircuit}\\
\hline
$g_{r(\alpha,i)}$ &  \parbox[t]{0.6\textwidth}{two-qudit gates applied to even then
                              odd neighbors in the $i$-th row in the $\alpha$ direction\vspace{1.5mm}} &Definition \ref{def:elementarycircuit}\\
\hline
$g_{R}$ and $g_C$ &  $g_{\text{Rows}(1,n)}$ and $g_{\text{Rows}(2,n)}$
                    when $D=2$.&Definition \ref{def:elementarycircuit}\\
\hline
$G_{\text{Rows}(\alpha,n)}$ &  Haar projector applied  to  each row in
                              the$\alpha^{\text{th}}$ direction &Definition \ref{def:mm}\\
\hline
$G_R(G_C)$ &  Haar projector applied  to each row (column) of a $2$D lattice &Definition \ref{def:mm}\\
\hline
$G_{\text{Planes}(\alpha)}$ & Haar projector applied  to each  plane perpendicular to the direction $\alpha$ & Definition \ref{def:planes}\\
\hline
$\Coll (\cA)$ & collision probability from superoperator $\cA$&Definition \ref{def:coll}\\
\hline
$\Coll_s$ & the expected collision probability of a random circuit after $s$ steps&Definition \ref{def:coll} \\
\hline
$\mu^{\lattice,n}_{D,c,s}$ & the distribution over $D$-dimensional
                             circuits with $n$ qudits&Definition \ref{def:recursive} \\
\hline
${\tilde G}_{n,D,c}$ & same as $G^{(t)}_{\mu^{\lattice,n}_{D,c,s}}$ except that we replace $g^s_{\text{Rows} (\alpha,n)}$ with $G_{\text{Rows} (\alpha,n)}$ &Definition \ref{def:recursive}\\
\hline
${\hat G}_{n,D,c,s}$ & one block of ${\tilde G}_{n,D,c}$ defined as $G_{\text{Rows}(D,n)} \tilde{G}^{\tensor n^{1/D}}_{n^{1-1/D}, D-1, c,s} G_{\text{Rows}(D,n)}$ &Definition \ref{def:recursive}\\
\hline
$\mu^{\text{CG}}_{s}$ & the distribution over circuits with $s$ random
                        two-qubit gates &Definition \ref{def:cg}\\
\hline
$\measuredangle (A,B)$ & $\cos^{-1}\max_{x\in A, y \in B} \braket{x,y}$ is the angle between two vector spaces $A$ and $B$&Section \ref{overlap}\\
\hline

\bottomrule\end{tabular}
\end{center}


\subsection{Elementary tools}
\label{sec:elementarytools}

If $A$ is a matrix and $\sigma_i$ are the singular values of $A$, then for $p \in [1,\infty)$ the Schatten $p$-norm of $A$ is defined as $\|A\|_p := (\sum_i \sigma^p_i)^{1/p}$. The $\infty$-norm of $A$ is $\|A\|_\infty := \max(i) \sigma_i$. The $1$-norm is related to the $\infty$-norm by $\|A\|_1 \leq \rank(A) \cdot \|A\|_\infty$. Moreover, for $p \in [1, \infty]$ and any two matrices $A$ and $B$, $\|A \tensor B\|_p = \|A\|_p\cdot \|B\|_p$.

If $\cA$ and $\cB$ are superoperators, then $\|\cA \tensor \cB\|_\diamond = \|\cA\|_\diamond \cdot \|\cB\|_\diamond$.

$\Channel\left[\cdot\right]$ is the linear map from matrices to superoperators such that for any two equally sized matrices $A$ and $B$, $\Channel \left[ A\tensor B^\ast\right] = A [\cdot ] B^\dagger $. Note that $\Channel \left[\cdot\right]$ is associative in the sense that $\Channel \left[A \tensor B^\ast\right] \circ \Channel\left [C \tensor D^\ast\right] = \Channel \left[A C  \tensor B^\ast C^\ast\right]$, for any equally sized matrices $A,B,C,D$.

Consider the Haar measure over $\text U(d)$. $\Channel\left[G^{(t)}_\Haar\right]$ (defined in the previous section) is the projector onto the matrix vector space of permutation operators (permuting length $t$ words over the alphabet $[d]$). In particular, for any matrix $X \in \C^{d^t \times d^t}$ we can write 
\be
\Channel\left[G^{(t)}_\Haar\right] [X] = \sum_{\pi \in S_t} \Tr (V(\pi) X) Wg(\pi),
\ee
where $V(\pi)$ is the permutation matrix $\sum_{(i_1, \ldots, i_t) \in [d]^t} \ket{i_1,
  \ldots, i_t} \bra{i_{\pi(1)}, \ldots, i_{\pi(t)}}$, and $Wg(\pi)$ is a linear combination
of permutations.  Specifically
\be
Wg(\pi) = \sum_{\sigma \in S_t} \alpha(\pi^{-1} \sigma) V(\sigma).
\ee
Here the coefficients $\alpha(\cdot)$ are known as Weingarten functions (see \cite{C06}).  
If $\mu,\nu\in S_t$ then let $\dist(\mu,\nu)$ denote the number of transpositions needed
to generate $\mu^{-1}\nu$ from the identity permutation.  Then we can define
$\alpha(\cdot)$ by the following relation.
\be \sum_{\mu,\nu\in S_t} \alpha(\mu^{-1}\nu)\ket{\mu}\bra{\nu}
= \left(\sum_{\mu,\nu\in S_t} \dist(\mu,\nu)\ket{\mu}\bra{\nu}\right)^{-1}.\ee
Note that $\alpha(\pi)$ is always real and $|\alpha (\lambda)| = O(1/d^{t+\dist
  (\lambda)})$. Thus for large $d$, $Wg(\pi) \approx V(\pi)/d^t$.

Furthermore,
\be
\Channel\left[G^{(t)}_\Haar\right] [X] = \sum_{\pi \in S_t} \Tr_M \left ((V(\pi)_M \tensor I_N) X_{MN}\right) \tensor \Wg(\pi)_M.
\ee

Let $A, B$ be matrices. For the superoperator $\cD \equiv B \Tr[A \cdot]$ we use the
notation $\cD = B A^\ast$.
We need the following observation:
\be
V(\pi) V^\ast(\sigma) = \Channel\left[\ket{\psi_{\pi}} \bra{\psi_\sigma}\right],
\ee
where $\ket{\psi_\pi} = (I \tensor V(\pi)) \frac{1}{\sqrt{d^t}} \sum_{i \in [d]^t} \ket i \ket i$.

We need the following lemma:
\begin{lemma}\OldNormalFont{}
If $A$ is a (possibly rectangular) matrix, then $A A^\dagger$ and $A^\dagger A$ have the same spectra.
\label{lem:spectrum}
\end{lemma}

\begin{lemma}\OldNormalFont{}
If $A$ and $B$ are matrices and $\|\cdot\|_\ast$ is a unitarily invariant norm, then 
$\|A B\|_\ast \leq \|A\|_\ast \|B\|_\infty$.
\label{lem:infinitypalace}
\end{lemma}
\begin{proof}
This lemma can be viewed as a consequence of Russo-Dye theorem, which states that the
extreme points of the unit ball for $\|\cdot \|_\infty$ are the unitary matrices.  Thus we
can write $B = \|B\|_\infty \sum_i p_i U_i$ for $\{p_i\}$ a probability distribution and
$\{U_i\}$ a set of unitary matrices.  We use this fact along with the triangle inequality and
then unitary invariance to obtain
\be 
\|A B\|_\ast 
=    \| A \cdot \|B\|_\infty \sum_i p_i U_i   \|_\ast
\leq \|B\|_\infty \sum_i p_i \|A U_i \|_\ast
= \|B\|_\infty \sum_i p_i \|A  \|_\ast
= \|A\|_\ast \|B\|_\infty.\ee
\end{proof}

A similar argument applies to superoperators.
\begin{lemma}\OldNormalFont{}
If $\cA$ is a superoperator and $\cB $ is a tpcp superoperator
 then $\|\cA \cB\|_\diamond \leq \|\cA\|_\diamond$.
\label{lem:diamondpalace}
\end{lemma}
\begin{proof}
Let $d$ be $\geq$ the input dimensions of both $\cA$ and $\cB$.  Then
$\|\cA\|_\diamond = \max_{\|X\|_1 \leq 1} \|(\cA \ot \id_d)(X)\|_1$ and 
$\|\cA\cB\|_\diamond = \max_{\|X\|_1 \leq 1} \|(\cA \ot \id_d)(\cB \ot \id_d)(X)\|_1$.
Since $\cB$ is a tpcp superoperator
$\|(\cB \ot \id_d)(X)\|_1 \leq 1$ and so $\|\cA\cB\|_\diamond$ is maximizing over a set
which is contained in the set maximized over by $\|\cA\|_\diamond$.
\end{proof}

These give rise to the following well-known bound, which often is called ``the hybrid argument.''
\begin{lemma}\OldNormalFont{}
Let $\| \cdot \|_\ast$ be a unitarily invariant norm. If $A_1, \ldots,A_t$ and $B_1,
\ldots, B_t$ have $\infty$-norm $\leq 1$. Then
\be
\|A_1\ldots A_t - B_1 \ldots B_t\|_\ast \leq \sum_{i}\|A_i- B_i\|_\ast.
\ee
This is also true for superoperators and the diamond norm, if each superoperator is a
tpcp map.
\label{lem:bv}
\end{lemma}

We will need a similar bound for tensor products.  
\begin{lemma}\OldNormalFont{}
Suppose $\|A - B\|_\ast \leq \eps$ for some norm $\|\cdot \|_\ast$ that is multiplicative
under tensor product. Then for any integer $M > 0$
\be
\left \|A^{\tensor M} - B^{\tensor M}\right\|_\ast \leq (\|B\|_\ast +\eps)^M -\|B\|_\ast.
\ee
The same holds for superoperators and the diamond norm. In particular $\|A^{\tensor M} - B^{\tensor M}\|_\ast \leq 2 M \|B\|_\ast^M  \eps$ for $\eps \leq \frac{1}{2 M}$.
\label{lem:tensorerror}
\end{lemma}

We need the following definition and lemma:
\begin{definition}\OldNormalFont{}
Let $X$ and $Y$ be two real valued random variables on the same totally ordered sample space $\Omega$. Then we say $X$ is stochastically dominated by $Y$, if for all $x \leq y \in \Omega$, $\Pr [X \geq x] \leq \Pr [Y \geq y]$. We represent this by $X \preceq Y$.
\end{definition}

\begin{lemma}[Coupling]\OldNormalFont{}
$X \preceq Y$ if and only if there exists a coupling (a joint probability distribution) between $X$ and $Y$ such that the marginals of this coupling are exactly $X$ and $Y$ and that with probability $1$, $X\leq Y$.
\end{lemma}

\subsection{Various measures of convergence to the Haar measure}
\label{sec:norms}

\begin{definition}\OldNormalFont{}
Let $\mu$ be a distribution over $n$-qudit gates. 
 Let $\eps$ be a positive real number.
\begin{enumerate}
\item (Strong designs) $\mu$ is a strong $\epsilon$-approximate $t$-design if
\be
(1-\epsilon) \cdot \Channel\left[G^{(t)}_{\Haar}\right]\preceq \Channel\left[G^{(t)}_{\mu}\right] \preceq (1+\epsilon) \cdot \Channel\left[G^{(t)}_{\Haar}\right],
\ee
or equivalently if
\be
(1-\epsilon)\cdot   \left (\Channel\left[G^{(t)}_{\Haar}\right]\tensor \mathrm{id} \right ) \Phi^{\tensor t}_{d^n}\preceq  \left(\Channel\left[G^{(t)}_{\mu}\right] \tensor \mathrm{id}  \right)\phi_{d^n}^{\tensor t}\preceq (1+\epsilon) \cdot  \left (\Channel\left[G^{(t)}_{\Haar}\right]\tensor \mathrm{id}  \right) \Phi_{d^n}^{\tensor t}.
\ee
The first $\preceq$ is cp ordering and the second $\preceq$ is psd ordering.

 \item (Monomial definition) $\mu$ is a monomial based $\epsilon$-approximate $t$-design if for any balanced monomial $m(C)$ of degree at most $t$
\be
  \left \|\vvec \left[G^{(t)}_\mu\right]-\vvec \left[G^{(t)}_\Haar\right] \right\|_\infty \leq \frac{\epsilon}{d^{nt}}.
\ee
Here for a matrix $A$, $\vvec(A)$ is a vector consisting of the entries of $A$ (in the computational basis).

\item (Diamond definition) $\mu$ is an $\epsilon$-approximate $t$-design in the diamond measure if
\be
\left  \| \Channel\left[G^{(t)}_{\mu}\right] - \Channel\left[G^{(t)}_{\Haar}\right]\right \|_\diamond \leq \epsilon.
\ee

\item (Trace definition) $\mu$ is an $\epsilon$-approximate $t$-design in the trace measure if
\be
 \left \| G_{\mu}^{(t)} - G_{\Haar}^{(t)}\right \|_1 \leq \eps.
\ee
 
 \item (TPE) $\mu$ is a $(d,\eps,t)$ $t$-copy tensor product expander (TPE) if
\be
\left  \| G_{\mu}^{(t)} - G_{\Haar}^{(t)}\right \|_\infty \leq \eps.
\ee

\item (Anti-concentration) $\mu$ is an $\epsilon$ approximate anti-concentration design if
\be
\E_{C\sim \mu} |\braket{0|C|0}|^4 \leq  \E_{C\sim \Haar} |\braket{0|C|0}|^4 \cdot (1+\eps).
\ee

\item (Approximate scramblers) $\mu$ is an $\eps$-approximate scrambler if for any density matrix $\rho$ and subset $S$ of qubits with $|S| \leq n/3 $
\be
\E_{C \sim \mu} \left \|\rho_S(C) - \frac{I}{2^{|S|}} \right \|^2_1 \leq \eps.
\ee
where $\rho_S(C) = \Tr_{\backslash S}C \rho C^\dagger$ and $\Tr_{\backslash S}$ is trace over the subset of qubits that is complimentary to $S$.

\item (Weak approximate decouplers) Let $M, M', A, A'$ be systems composed of $m, m, n-m$ and $n-m$, and let $\phi_{MM'}$, $\phi_{AA'}$ and $\psi_{A'}$ be respectively maximally entangled states along $M, M'$, maximally entangled state along $AA'$ and a pure state along $A'$.  $\mu$ is an $(m,\alpha,\eps)$-approximate weak decoupler if for any subsystem $S$ of $M'A'$ with size $\leq \alpha \cdot n$, when $\mu$ applies to $M'A'$,
\be
\E_{C \sim \mu}\left \|\rho_{MS}(C) - \frac{I}{2^{m}} \tensor \frac{I}{2^{|S|}}\right\|_1\leq \eps.
\ee
We consider two definitions. In the first definition the initial state is $\phi_{MM'} \tensor \phi_{AA'}$ and in the second model it is $\phi_{MM'} \tensor \psi_{A'}$. Here $\rho_{MS}(C)$ is the reduced density matrix along $MS$ after the application of $C \sim \mu$.

\end{enumerate}
\end{definition}

\section{Approximate $t$-designs by random circuits with nearest-neighbor gates on $D$-dimensional lattices}
\label{sec:lattices-designs}

In this section we prove theorems \ref{thm:grid} and
\ref{thm:lattice}, which state that our random circuit models defined
for $D$-dimensional lattices (definitions \ref{def:modelsD}) form
approximate $t$-designs in several measures. 

We begin in \secref{basic-lemma} by outlining some basic utility
lemmas.  The technical core of the proof is contained in the lemmas in
\secref{overlap-lemmas} in which we bound various norms of products of
Haar projectors onto overlapping sets of qubits.  These are proved in
Sections \ref{sec:lemma-proofs} and \ref{sec:overlap-proofs}
respectively.  We show how to use these lemmas to prove our main
theorems in \secref{proofs} (for a 2-D grid) and in \secref{proof2}
(for a lattice in $D>2$ dimensions).

\subsection{Basic lemmas}
\label{sec:basic-lemma}

In this section we state some utilities lemmas which are largely
independent of the details of our circuit models.

\subsubsection{Comparison lemma for random quantum circuits}
\label{sec:comparison}

\begin{definition}\OldNormalFont{}
A superoperator $\cC$ is completely positive (cp) if for any psd
matrix $X$, $(\cC \tensor \Id )(X)$ is also psd.
 For superoperators $\cA$ and $\cB$, $\cA \preceq \cB$ if $\cB-\cA$ is cp.
\label{def:cpmap}
\end{definition}

Our comparison lemma is simply the following:
\def\LemComparison{}
\begin{lemma}[Comparison]
\torestate{\OldNormalFont{} Suppose we have the following cp ordering between superoperators $\cA_1 \preceq \cB_1, \ldots,\cA_t \preceq \cB_t$. Then $\cA_t \ldots \cA_1 \preceq \cB_t \ldots \cB_1$.
\label{lem:comparison}}
\end{lemma}

\def\CorOverlapDesign{
}
\begin{corollary}[Overlapping designs] 
\torestate{\OldNormalFont{}
If $K_1, \ldots, K_t$ are respectively the moments superoperators of $\eps_1, \ldots, \eps_t$-approximate strong $k$-designs each on a potentially different subset of qudits, then
\be
\Channel\left[G^{(t)}_{\Haar(S_1)} \ldots G^{(t)}_{\Haar(S_t)}\right] (1-\eps_1)\ldots (1-\eps_t) \preceq K_1 \ldots K_t \preceq \Channel\left [G^{(t)}_{\Haar(S_1)} \ldots G^{(t)}_{\Haar(S_t)}\right] (1+\eps_1)\ldots (1+\eps_t).
\ee
\label{cor:overlappingdesigns}}
\end{corollary}

\subsubsection{Bound on the value of off-diagonal monomials}

 We first formally define an off-diagonal monomial.
\begin{definition}[Off-diagonal monomials] \OldNormalFont{}
A diagonal monomial of balanced degree $t$ of a unitary matrix $C$ is a balanced monomial that can be written as product of absolute square of terms, i.e., $|C_{a_1,b_1}|^2 \ldots |C_{a_t,b_t}|^2$. A monomial is off-diagonal if it is balanced and not diagonal.
\label{def:offdiagonal}
\end{definition}

We now define the set of diagonal indices as $\cD = \{\ket{i,j}\bra{i',j'} : i=i',j=j', i,i',j,j' \in
  [d]^{nt}\}$ and the set of off-diagonal indices as $\cO = \{\ket{i,j}\bra{i',j'} : i\neq i' \text{ or } j\neq j', i,i',j,j' \in
  [d]^{nt}\}$. We note that a diagonal monomial can be written as $\Tr (C^{\tensor t,t} x)$ for some $x \in \cD$ and similarly,  an off-diagonal monomial can be written as $\Tr (C^{\tensor t,t} x)$ for some $x \in \cO$.

We relate the strong definition of designs to the monomial definiton
via the following lemma.

\begin{lemma}
\torestate{\label{lem:offdiag}
Let $\delta > 0$. Assume that $\Channel\left[G_\mu^{(t)}\right]$ and $\Channel\left[G_\nu^{(t)}\right]$ are two moment superoperators that satisfy the following completely positive ordering
\be
(1 - \delta)  \cdot \Channel\left[G_\nu^{(t)}\right] \preceq \Channel\left[G_\mu^{(t)}\right] \preceq (1 + \delta)  \cdot \Channel\left[G_\nu^{(t)}\right].
\ee
Let $\mathcal{O}$ and $\mathcal{D}$ be respectively the set of off-diagonal and diagonal indices for monomials. Then
\be
\max_{x \in \cO} |\Tr \left(x G^{(t)}_\mu\right)|
\leq \max_{x \in \cO} |\Tr \left(x G^{(t)}_\nu\right)| (1+\delta) + 2\delta \cdot \max_{y \in \cD} |\Tr \left(y G^{(t)}_\nu\right)|.
\ee
}
\end{lemma}

\subsubsection{Bound on the moments of the Haar measure}

We need the following bound on the $t$-th monomial moment of the Haar
measure. Assume we have $m$ qudits.
\begin{lemma}[Moments of the Haar measure] 
\torestate{\OldNormalFont{}
Let $G_{\Haar(m)}^{(t)}$ be
  the quasi-projector operator for the Haar measure on $m$ qudits. Then
\be
\max_y \left \|G_{\Haar(m)}^{(t)} y G_{\Haar(m)}^{(t)}\right\|_1 \leq \frac{t^{O(t)}}{d^{mt}}.
\ee
Here the maximization is taken over matrix elements in the computational basis like 
\\
$y= \ket{i_1,\ldots, i_t, i'_1,\ldots, i'_t}\bra{j_1,\ldots, j_t, j'_1,\ldots, j'_t}$. Each label (e.g. $i_j$) is in $[d]^m$.
\label{lem:Haarmoment}}
\end{lemma}

\subsection{Gap bounds for the product of overlapping Haar projectors}
\label{sec:overlap-lemmas}

We will later need the following results, with proofs deferred until
\secref{overlap-proofs}.

\begin{lemma} 
\torestate{\OldNormalFont{} $\|G_C G_R -G_\Haar^{(t)}\|_\infty \leq 1/d^{\Omega(\sqrt{ n})}$. 
\label{lem:2Dproj}}
\end{lemma}

\begin{lemma}
\torestate{\OldNormalFont{}
Let $D = O(\ln n / \ln \ln n)$ with small enough constant factor, then $\| G_{\text{Planes}(D)} G_{\text{Rows}(D,n)}  - G_{\Haar} \|_\infty \leq 1/d^{\Omega(n^{1-1/D})}$.
\label{lem:generalDproj}}
\end{lemma}

\begin{lemma}
\torestate{\OldNormalFont{}
Let $\ket{x} $ and $\ket{y}$ be two computational basis states. For small enough $D = O(\ln n / \ln \ln n)$ and large enough  $c$, $|\braket{x | \tilde{G}_{n,D,c}  -G_\Haar | y}| \leq \frac{\eps}{d^{nt}}$ for some $\eps = 1/d^{\Omega(n^{1/D})}$.
\label{lem:basicDmonomial}}
\end{lemma}

\begin{lemma}\OldNormalFont{}
For large enough $c$, $\left \|\Channel\left[( G_R G_C G_R )^{c} -G_\Haar^{(t)}\right] \right \|_\diamond = \frac{t^{O(\sqrt n t)}}{d^{\Omega(c \sqrt n)}}$.
\label{lem:diamond2}
\end{lemma}

\begin{lemma}
\torestate{\OldNormalFont{} For small enough $D = O(\ln n / \ln \ln n)$ and large enough  $c$, 
\be
\left \|\Channel\left[( G_{\text{Rows}(D,n)} G_{\text{Planes}(D)} G_{\text{Rows}(D,n)})^{c} - G_\Haar^{(t)}\right] \right\|_\diamond = \frac{t^{O(t n^{1-1/D})}}{d^{\Omega(c n^{1-1/D})}}.
\ee
\label{lem:diamondD}}
\end{lemma}

In these last two lemmas, we see that $c$ will need to grow with $t$.
We believe that a sharper analysis could reduce this dependence, but
since we already have a $\poly(t)$ dependence in $s$, improving Lemmas
\ref{lem:diamond2} and \ref{lem:diamondD} would not make a big
difference.  In fact, even in 1-D, \cite{BHH-designs} found a sharp
$n$ dependence but their factor of $\poly(t)$ (which we inherit) is
probably not optimal.


\subsection{Proof of Theorem \ref{thm:grid}; $t$-designs on two-dimensional lattices}
\label{sec:proofs}
\restatetheorem{thm:grid}

\begin{proof}
\begin{enumerate}
\item This item corresponds to convergence of the individual moments of the Haar measure. A balanced moment of a distribution $\mu$ can be written as
\be
\E_{C \sim \mu}  [C_{i_1, j_1} \ldots C_{i_t,j_t} C^{\ast}_{i'_1,
  j'_1} \ldots C^{\ast}_{i'_t,j'_t} ] 
=\bra{i, i'} G_\mu^{(t)} \ket{ j,j'}
= \Tr [G_\mu^{(t)} \cdot \ket{j,j'}\bra{i,i'}]
\label{eq:balanced-mom}
\ee 
where $\ket i := \ket{i_1,\ldots,i_t}$ and so on for $\ket{i'},\ket
j,\ket{j'}$.   The same moment can also be written as
\be
 \Tr  \left(\ket{j}\bra{j'} \Channel\left[G_\mu^{(t)}\right](\ket{i}\bra{i'})\right)
\label{eq:channel-mom}
\ee
We will see that the strong design condition established by gives
us strong bounds first for the ``diagonal'' case ($i=i',j=j'$) then the
off-diagonal case.  This is because when we interpret $G_\mu^{(t)}$ as
a quantum operation, the diagonal monomials correspond to $\Tr Y
G_\mu^{(t)} X$ for psd matrices $X,Y$, and so the strong design
condition applies directly.  For off-diagonal moments we need to do a
bit more work.

For each the diagonal and off-diagonal monomials, our strategy will be
to first compare with the entries of $G_R (G_C G_R)^c G_R$ and then to
compare to $G_{\Haar}$.

First observe that 
since $\Channel\left[G_{\mu^{\lattice,n}_{2,c,s}}^{(t)}\right] = (g^s_R g^s_C)^c
g^s_R$ and  $s = \poly(t) \cdot (\sqrt n + \ln (1/\delta))$ then
corollary 6 of \cite{BHH-designs}  implies that each $g^s_i$ for $i \in \{R,C\}$ is an $\delta$-approximate $t$-design. Hence, using corollary \ref{cor:overlappingdesigns},
\be
\Channel\left[G_R \left( G_C G_R\right]\right)^c G_R (1 - \frac{\delta}{4t!})] \preceq \Channel\left[G_{\mu^{\lattice,n}_{2,c,s}}^{(t)}\right] \preceq \Channel\left[G_R ( G_C G_R\right])^c G_R (1 + \frac{\delta}{4t!}).
\label{eq:gRCR-ineq}
\ee
Note that we chose $\poly(t)$ large enough so that the error is as small as $\frac{\delta}{4t!}$. This choice will be helpful later.

%
%
Focusing first on diagonal monomials $\proj i,\proj j$ we  can bound
\bea
&&(1+ \frac{\delta}{4t!}) \Tr \left(\proj j \Channel\left[G_R (G_C G_R])^c G_R\right]
  (\proj i)\right) - 
\bra{i,j}G_{\mu^{\lattice,n}_{2,c,s}}^{(t)}\ket{i,j}\nonumber \\
 &=& \Tr \left (\proj j  [\Channel\left[G_R ( G_C G_R\right]\right)^c G_R (1+\frac{\delta}{4t!})-
 G_{\mu^{\lattice,n}_{2,c,s}}^{(t)}] ](\proj i)) \geq 0.
\eea
In other words, for diagonal monomials
\ba
\Tr \left (\proj j
  \Channel\left [G_{{\mu^{\lattice,n}_{2,c,s}}^{(t)}}^{(t)}\right](\proj i)\right)
&\leq (1+ \frac{\delta}{4t!}) \Tr \left (\proj j \Channel\left [G_R \left( G_C G_R\right]\right)^c G_R] (\proj i)\right)\nonumber \\ & =  (1+ \frac{\delta}{4t!}) \Tr \left (G_R ( G_C G_R)^c G_R \proj{i,j}\right).
\ea

Similarly, using the first inequality in \eq{gRCR-ineq}
\be
\Tr \left (\proj j
  \Channel\left [G_{{\mu^{\lattice,n}_{2,c,s}}^{(t)}}^{(t)}\right](\proj i)\right)
  \geq (1- \frac{\delta}{4t!}) \Tr \left(G_R ( G_C G_R)^c G_R \proj{i,j}\right).
\ee

The next step is to bound $ \Tr \left( y G_R (G_C G_R)^c G_R x\right)$:
\bea
\left | \Tr \left (G_R (G_C G_R)^c G_R x\right) -  \Tr \left (G_\Haar^{(t)} x\right)  \right | &=&   \left | \Tr  \left( (G_R (G_C G_R)^c G_R - G_\Haar^{(t)} ) x \right)  \right|\nonumber \\
 &=&  \left | \Tr \left  ( ((G_C G_R)^c - G_\Haar^{(t)} ) G_R x G_R \right) \right|\nonumber \\
  &\leq&   \left \| ((G_C G_R)^c - G_\Haar^{(t)} ) \|_\infty\cdot  \|G_R x G_R \right\|_1\nonumber \\
   &\leq&   \left \|G_C G_R- G_\Haar^{( t)} \right\|^c_\infty \cdot  \left (\max_{y \in [d]^{2 \sqrt n t}} \|G_{r_{1,1}} y G_{r_{1,1}}\|_1 \right)^{\sqrt n}.
\eea

In the third line we have used the H\"older's inequality. In the last inequality we have used the fact that $G_1$ is a tensor product of $G_{r_{1,i}}$ across each column in the first direction; by symmetry we can just consider $G_{r_{1,1}}$.

Using Lemma \ref{lem:Haarmoment}
\be
\max_{y\in [d]^{2 \sqrt n t} } \left \|G_{r_{1,1}} y G_{r_{1,1}}\right \|_1 = \frac{t^{O(t)}}{d^{t \sqrt n}}.
\ee
Furthermore, using Lemma \ref{lem:2Dproj}
\be
 \left\|G_C G_{R}- G_\Haar^{( t)} \right\|_\infty \leq \frac{1}{d^{\Omega(\sqrt n)}}.
\ee
therefore
\be
 \left\|G_C G_R- G_\Haar^{( t)} \right \|^c_\infty \cdot  (\max_{y \in [d]^{2 \sqrt n t}} \|G_{r_{1,1}} y G_{r_{1,1}}\|_1 )^{\sqrt n} \leq \frac{1}{d^{O(c \cdot \sqrt n)}} \cdot  (\frac{t^{O(t)}}{d^{t \sqrt n}} )^{\sqrt n}.
\ee
As a result, for some large enough $c = O(t \ln t + \frac{\ln 1/\delta}{\sqrt n})$ we conclude
\be
 | \Tr \left(G_R (G_C G_R)^c G_R x\right) -  M^{(\Haar,t)}_x   | \leq  \|G_C G_R- G_\Haar^{(t)} \|^c_\infty \cdot  (\max_{y\in[d]^{2 \sqrt n t}} \|G_{r_{1,1}} y G_{r_{1,1}}\|_1 )^{\sqrt n} \leq \frac{\delta}{4 d^{nt}}.
\label{eq:83}
\ee
As a result, using Lemma \ref{lem:Haarmoment} any diagonal monomial satisfies 
\bea
|
\Tr \left(G^{(t)}_{\mu^{\lattice,n}_{2,c,s}}\right) -  \Tr \left(G^{(t)}_{\Haar}\right)| &\leq& |\Tr \left(G_R ( G_C G_R)^c G_R x\right )-\Tr \left(G_R ( G_C G_R)^c G_R x\right)| \nonumber\\
 &&+ \frac{\delta}{4t!} |\Tr \left(G_R ( G_C G_R)^c G_R x\right)|\nonumber\\
 &\leq& \frac{\delta}{4d^{nt}}+ \frac{\delta}{4t!} (M^{(\Haar,t)}_x+ \frac{\delta}{4d^{nt}})\nonumber\\
  &\leq& \frac{\delta}{4d^{nt}}+ \frac{\delta}{4t!} (t!/d^{nt}+ \frac{\delta}{4d^{nt}})\nonumber\\
&\leq& \frac{\delta}{d^{nt}}.
\label{eq:84}
\eea


Next, we bound the expected off-diagonal monomials of the distribution. The value of the off-diagonal monomials according to the Haar measure is zero. So it is enough to bound $\max_{x \in \cO} | \Tr \left( G^{(t)}_\mu x \right) |$, where $\cO$ is the set of off-diagonal indices for moments. In order to do this we use Lemma \ref{lem:offdiag} for $\mu = \mu^{\lattice,n}_{2,c,s}$ and $\nu$ being a distribution with moment superoperator $\Channel[G_R] (\Channel[G_R] \Channel[G_C])^c \Channel[G_R]$.
\be
\max_{x \in \cO} |\Tr (G^{(t)}_{\mu^{\lattice,n}_{2,c,s}) x }| \leq \max_{x \in \cO}  \Tr (G_R (G_C G_R)^c G_R x) (1+\frac{\delta}{4t!}) + \delta/t! \cdot \max_{y \in \cD} \Tr (G_R (G_C G_R)^c G_R y).
\label{eq:85}
\ee
Here $\cD$ is the set of diagonal monomials. Using \eqref{eq:83}
\be
\max_{y \in \cD} \Tr (G_R (G_C G_R)^c G_R y) \leq \max_{y \in \cD} \Tr \left (G_\Haar^{(t)} y\right) + \frac{\delta}{4 d^{nt}} \leq \frac{t!}{d^{nt}} + \frac{\delta}{4 d^{nt}}.
\ee

In order to bound $\max_{x \in \cO}  \Tr (G_R (G_C G_R)^c G_R x)$, we first make the observation that since $x \in \cO$, $\Tr (G_\Haar^{(t)} x) =0$. Therefore
\bea
\max_{x \in \cO}  |\Tr (G_R (G_C G_R)^c G_R x)| &=& \max_{x \in \cO}  |\Tr ((G_R (G_C G_R)^c G_R - G_\Haar^{(t)})x)|\nonumber\nonumber\\
 &\leq& \max_{x \in \cO}  |\Tr ((G_R G_C)^c- G_\Haar^{(t)}) G_R x G_R)|\nonumber\nonumber\nonumber\\
  &\leq& \frac{\delta}{4 d^{nt}}.
\label{eq:87}
\eea
therefore using \eqref{eq:83}, \eqref{eq:85} and \eqref{eq:87} we conclude
\bea
\max_{x \in \cO} |\Tr (G^{(t)}_{\mu^{\lattice,n}_{2,c,s}}x )| &\leq& \frac{\delta}{4 d^{nt}} (1+\delta/ (4t!)) + \frac{\delta}{4t!} \cdot (\frac{t!}{d^{nt}} + \frac{\delta}{4 d^{nt}})\nonumber\\
&\leq& \frac{\delta}{2 d^{nt}} + 2 \delta/(4d^{nt}) \leq \frac{\delta}{d^{nt}}.
\eea

\item 
\bea
\left \| \Channel\left[G_{\mu^{\lattice,n}_{2,c,s}}^{(t)}-G_\Haar^{(t)}\right] \right\|_\diamond &\leq& \left\|  \Channel\left[g^s_R (g^s_C g^s_R)^c - (G_R G_C G_R\right])^c \right\|_\diamond+ \left\| \left(\Channel\left[G_R G_C G_R\right]\right)^c - G_\Haar^{(t)}] \right\|_\diamond\nonumber\\
&\leq& 4c \cdot \left \|  \Channel\left[g^s_{r_{1,1}}\right]^{\tensor \sqrt n} - \left[G_{r_{1,1}}\right]^{\tensor \sqrt n} \right\|_\diamond + (\frac{t^t}{d^{c}})^{O(\sqrt n)}\nonumber\\
&\leq& 4c \cdot \sqrt n \cdot \left \|  \Channel\left[g^s_{r_{1,1}} - G_{r_{1,1}}\right] \right\|_\diamond + (\frac{t^t}{d^{c}})^{O(\sqrt n)}\nonumber\\
&\leq& \delta/2 + \delta/2\nonumber\\
&\leq& \delta.
\eea

In the first line we have used triangle inequality and the definition $K_{\mu^{\lattice,n}_{2,c,s}}^{(t)} = (\prod_\alpha g^s_{\text{Rows}(\alpha,n)})^c$. In the second line, for the first term we have used Lemma \ref{lem:bv} and that all operators are compositions of moment superoperators. For the second part we have used Lemma \ref{lem:diamond2}.
In the third inequality we have used Lemma \ref{lem:tensorerror}. In fourth inequality, the first term ($\delta/2$) comes from lemma 3 and corollary 6 of \cite{BHH-designs} for  $s= \poly(t) \cdot (\sqrt n + \ln \frac{1}{\delta} )$, and the second $\delta/2$ is by the choice $c = O(t \ln t + \frac{\ln (1/\delta)}{\sqrt  n})$.


\item Let $Q_0 := G_{r_{1,1}}$ and $Q_1 := G_{r_{1,1}} - g^s_{r_{1,1}}$, and for $x \in \{0,1\}^{\sqrt n}$ let $Q_x = Q_{x_1} \ldots Q_{x_{\sqrt n}}$. Here, $\|Q_0\|_1 = t!$ and $\|Q_x\|_1 = t!^{\sqrt n - |x|} \cdot \|G_{r_{1,1}} - g^s_{r_{1,1}}\|^{|x|}$. 
\bea
\| G_{\mu^{\lattice,n}_{2,c,s}}^{(t)} - G_\Haar^{(t)} \|_1 &\leq& \|  (g^s_C g^s_R)^c - (G_C G_R)^c \|_1 + \| (G_C G_R)^c - G_\Haar^{(t)} \|_1\nonumber\\
&\leq& 4 c \cdot \|  (g^s_{r_{1,1}})^{\tensor \sqrt n} -
G_{r_{1,1}}^{\tensor \sqrt n} \|_1 + t^{O(t) \sqrt n} \|G_C G_R -
G_\Haar^{(t)} \|^c_\infty\label{eq:two-terms}
\eea
We bound the two terms separately.
First
\ba
4 c \|  (g^s_{r_{1,1}})^{\tensor \sqrt n} - G_{r_{1,1}}^{\tensor
  \sqrt n} \|_1 
& \leq 4 c \cdot \sum_{x \in \{0,1\}^{\sqrt n} : x \neq 0} \|Q_x\|_1
\nn \\
&\leq 4 c \cdot [(t! + \|g^s_{r_{1,1}} - G_{r_{1,1}}\|_1) ^{\sqrt n} -
t!^{\sqrt n}] \nn \\ 
& = 4c t! ((1 + \|g^s_{r_{1,1}} - G_{r_{1,1}}\|_1 / t!)^{\sqrt n}-1)
\nn \\ & \leq 4c \cdot 2 \sqrt n \|g^s_{r_{1,1}} - G_{r_{1,1}}\|_1 \nn
\\
\intertext{
\hfill The last line needs $s$ to be large enough
that $\sqrt n \|g^s_{r_{1,1}} - G_{r_{1,1}}\|_1 \leq 1/(2\sqrt
n)$.  }
& \leq 8c\sqrt n (dt!)^{\sqrt n}\cdot \|g^s_{r_{1,1}} - G_{r_{1,1}}
\|_\infty
\nn \\ & \leq
8c\sqrt n (dt!)^{\sqrt n} (1-1/\poly(t))^s
\nn \\ & \leq \delta /2 \
\ea

Now we bound the second term of \eq{two-terms}.
\ba  t^{O(t) \sqrt n} \|G_C G_R -
G_\Haar^{(t)} \|^c_\infty
&  \leq
t^{O(t) \sqrt n} (d^{-\Omega(\sqrt n)})^c & \text{using \lemref{2Dproj}} \\
& \leq t^{C_1t\sqrt n} d^{-cC_2\sqrt n} & \text{for some universal constants
  $C_1,C_2>0$} \\
& = (t^{C_1t}/d^{cC_2})^{\sqrt{n}} \nn\\
& \leq \delta/2.\ea
In the last step we need to choose the implicit constant in the
definition of $c$ based on $C_1,C_2$.

\item
\bea
\| G_{\mu^{\lattice,n}_{2,c,s}}^{(t)} - G_\Haar^{(t)} \|_\infty &\leq& \|  (g^s_C g^s_R )^c - (G_C G_R)^c \|_\infty + \| (G_C G_R)^c - G_\Haar^{(t)} \|_\infty\nonumber\\
&\leq&4 c \cdot \|  (g^s_{r_{1,1}})^{\tensor \sqrt n} - G_{r_{1,1}}^{\tensor \sqrt n} \|_\infty +  \|G_C G_R - G_\Haar^{(t)} \|^c_\infty\nonumber\\
&\leq&4 c \cdot \sqrt n \cdot \| g^s_{r_{1,1}} - G_{r_{1,1}} \|_\infty + \frac{1}{d^{\Omega( c \sqrt n )}}\nonumber\\
&\leq& 4 c \cdot \sqrt n \cdot e^{-s/\poly(t)} + \frac{1}{d^{\Omega( c \sqrt n )}}.
\eea
These steps follow from the proof of part 1.

\end{enumerate}
\end{proof}

\newpage

\subsection{Proof of Theorem \ref{thm:lattice}; $t$-designs on $D$-dimensional lattices}
\label{sec:proof2}

Throughout this section we treat $D$ and $t$ as constants. 

\restatetheorem{thm:lattice}

\begin{proof}
\begin{enumerate}
\item 
Consider the moment superoperator for the $D$-dimensional random circuit distribution $\Channel \left[G_{\mu^{\lattice,n}_{D,c,s}}\right]$, where for $3\leq \alpha \leq D$, $\kappa_{\text{Rows}(\alpha,n)}$ is defined according to the recursive formula $\kappa_{\alpha} =  \kappa^{\otimes n^{1/\alpha}}_{\alpha-1} ((\Channel[g_i])^s \kappa^{\otimes n^{1/\alpha}}_{\alpha-1})^c $.

Using corollary 6 of \cite{BHH-designs}, if $s = O (n^{1/D})$ then each $g^s_{\text{Rows}(\alpha,n)}$ for $1\leq \alpha \leq D$ satisfies a $1/d^{\Omega(n^{1/D})}$-approximate $t$-design property. Hence, using corollary \ref{cor:overlappingdesigns}
\be
\Channel[\tilde{G}_{n,D,c}] (1 - 1/d^{\Omega(n^{1/D})}) \preceq \Channel\left[G_{\mu^{\lattice,n}_{2,c,s}}^{(t)}\right] \preceq \Channel\left[\tilde{G}_{n,D,c}\right] (1 + 1/d^{\Omega(n^{1/D})}).
\ee

Therefore,
\be
(1 - 1/d^{\Omega(n^{1/D})}) \Tr (\tilde{G}_{n,D,c}  x) \leq
\Tr \left(G^{(t)}_{\mu^{\lattice,n}_{D,c,s}}  x\right)
 \leq  (1+ 1/d^{\Omega(n^{1/D})}) \Tr (\tilde{G}_{n,D,c}  x).
\label{eq:diag3}
\ee
Where $x$ is a matrix $\ket {i , j}\bra{i',j'}$ for $i, j , i' ,j' \in [d]^{nt}$.

Next, we use Lemma \ref{lem:basicDmonomial}.
This lemma along with the bound in \eqref{eq:diag3} and Lemma \ref{lem:Haarmoment} proves the stated bound for diagonal monomials:
\bea
|\Tr \left(G^{(t)}_{\mu^{\lattice,n}_{D,c,s}}  x\right) -\Tr (G^{(t)}_{\Haar}  x)|&\leq& |\Tr (\tilde{G}_{n,D,c}  x) - \Tr (G^{(t)}_{\Haar}  x)| +|\Tr (\tilde{G}_{n,D,c}  x)|1/d^{\Omega(n^{1/D})}\nonumber\\
&\leq&  \frac{1/d^{\Omega(n^{1/D})}}{d^{nt}} + (|\Tr (G^{(t)}_{\Haar}  x)| + \frac{1/d^{\Omega(n^{1/D})}}{d^{nt}})1/d^{\Omega(n^{1/D})}\nonumber\\
&\leq&  \frac{1/d^{\Omega(n^{1/D})}}{d^{nt}} + (t!/d^{nt} + \frac{1/d^{\Omega(n^{1/D})}}{d^{nt}})1/d^{\Omega(n^{1/D})}\nonumber\\
&\leq&   \frac{1/d^{\Omega(n^{1/D})}}{d^{nt}}.
\eea

Next, we bound off-diagonal monomials $\max_{x \in \cO} | \Tr \left( G^{(t)}_\mu x \right) |$. Again, we use Lemma \ref{lem:offdiag} for $\mu = \mu^{\lattice,n}_{D,c,s}$ and $\nu$ being a distribution with moment superoperator $K_{\mu^{\lattice,n}_{D,c,s}}$ (or the quasi-projector $\tilde{G}_{n,D,c} $):
\be
\max_{x \in \cO} |\Tr \left(G^{(t)}_{\mu^{\lattice,n}_{D,c,s}}  x\right)| \leq \max_{x \in \cO}  \Tr (\tilde{G}_{n,D,c}  x) (1+1/d^{\Omega(n^{1/D})}) + 1/d^{\Omega(n^{1/D})} \cdot \max_{y \in \cD} \Tr (\tilde{G}_{n,D,c}  y).
\label{eq:85D}
\ee
Using Lemma \ref{lem:basicDmonomial}
\be
\max_{y \in \cD} \Tr (\tilde{G}_{n,D,c}  y) \leq \max_{y \in \cD} \Tr \left(G_\Haar^{(t)} y\right) + \frac{1/d^{\Omega(n^{1/D})}}{ d^{nt}} \leq \frac{t!}{d^{nt}} + \frac{1/d^{\Omega(n^{1/D})}}{ d^{nt}}.
\label{eq:97}
\ee

Similar to \eqref{eq:87} we can show
\bea
\max_{x \in \cO}  |\Tr \left(\tilde{G}_{n,D,c}  x\right)| &=& \max_{x \in \cO}  |\Tr \left((\tilde{G}_{n,D,c}  - G_\Haar^{(t)})x\right)|\nonumber\\
 &\leq& \max_{x \in \cO}  |\Tr \left(({\hat {G}}_{n,D,c} )^c- G_\Haar^{(t)}\right) \tilde{G}_{n^{1-1/D},D-1,c}^{\otimes n^{1/D}} x \tilde{G}_{n^{1-1/D},D-1,c}^{\otimes n^{1/D}})|\nonumber\\
 &\leq& \frac{1/d^{\Omega(n^{1/D})}}{ d^{nt}},
\label{eq:87D}
\eea
therefore using \eqref{eq:85D}, \eqref{eq:97} and \eqref{eq:87D} we conclude that any monomial $M^{(\mu^{\lattice,n}_{D,c,s} ,t)}_x$ satisfies
\bea
\max_{x \in \cO} |\Tr \left(G^{(t)}_{\mu^{\lattice,n}_{D,c,s}}  x\right)| &\leq& \frac{1/d^{\Omega(n^{1/D})}}{d^{nt}}.
\eea

\item Let $\eps_{D,n} := \left \| \Channel\left[G_{\mu^{\lattice,n}_{D,c,s}}\right]- \Channel\left[G_\Haar^{(t)}\right] \right \|_\diamond$. We use induction to show that $\eps_{D,n} = 1/d^{\Omega(n^{1/D})}$ for any integers $n$ and $D$. This is true for $D=2$ by Theorem \ref{thm:grid}. Assuming $\eps_{D-1, n} = 1/d^{\Omega(n^{1/(D-1)})}$ for any $n$, we show that $\eps_{D,n} = 1/d^{\Omega(n^{1/D})}$.
\bea
\eps_{D,n} &:=& \left \| \Channel\left[G_{\mu^{\lattice,n}_{D,c,s}}- G_\Haar^{(t)}\right] \right\|_\diamond\nonumber\\
  &\leq& \left\| \Channel\left[ G_{\mu^{\lattice,n}_{D,c,s}} - \tilde{G}_{n,D,c}\right] \right\|_\diamond + \left\| \Channel\left[\tilde{G}_{n,D,c} - G_\Haar^{(t)}\right] \right\|_\diamond\nonumber\\
&\leq& \poly(n) \cdot \left\|  \Channel\left[(g^s_{r_{1,1}})^{\tensor n^{1-1/D}} - G_{r_{1,1}}^{\tensor n^{1-1/D}}\right] \right\|_\diamond\nonumber\\
 &&+ \left\| \Channel\left[\tilde{G}_{n,D,c}-G_\Haar^{(t)}\right] \right\|_\diamond\nonumber\\
&\leq& O(n^{1-1/D}) \cdot \|  \Channel\left[g^s_{r_{1,1}} - G_{r_{1,1}}\right]\|_\diamond\nonumber\\
&&+\|  \Channel \left[(\tilde{G}^{\tensor n^{1/D}}_{n^{1-1/D},D-1, c, s}G_{\text{Rows}(D,n)}\tilde{G}^{\tensor n^{1/D}}_{n^{1-1/D},D-1, c, s})^c - G_\Haar^{(t)}\right] \|_\diamond\nonumber\\
&\leq& O(n) \cdot 1/d^{\Omega(n^{1/D})}\nonumber\\
 &&+\left\|  \Channel\left[(\tilde{G}^{\tensor n^{1/D}}_{n^{1-1/D},D-1, c, s}G_{\text{Rows}(D,n)}\tilde{G}^{\tensor n^{1/D}}_{n^{1-1/D}, D-1, c,s})^c - G_\Haar^{(t)}\right] \right\|_\diamond\nonumber\\
  &\leq& 1/d^{\Omega(n^{1/D})}+\left\|  \Channel\left[(\tilde{G}^{\tensor n^{1/D}}_{n^{1-1/D},D-1, c, s} G_{\text{Rows}(D,n)} \tilde{G}^{\tensor n^{1/D}}_{n^{1-1/D}, D-1, c,s})^c - G_\Haar^{(t)}\right] \right\|_\diamond
 \eea
The third line is by triangle inequality. The fourth inequality is by Lemma \ref{lem:bv}. The fifth line is by Lemma \ref{lem:tensorerror} and the definition $\tilde{G}_{n,D,c} = ([\tilde{G}^{\tensor n^{1/D}}_{n^{1-1/D},D-1, c, s}G_{\text{Rows}(D,n)}\tilde{G}^{\tensor n^{1/D}}_{n^{1-1/D}, D-1, c,s}])^c$. The sixth line is by lemma 3 and corollary 6 of \cite{BHH-designs}, which assert that after linear depth in the number of qudits ($n^{1/D}$), the random circuit model we consider is $\eps$-approximate $t$-design in the diamond measure, and that $\eps$ can be made exponentially small in $n^{1/D}$. 

Next, we bound $\left\|  \Channel\left[(\tilde{G}^{\tensor n^{1/D}}_{n^{1-1/D},D-1, c, s}G_{\text{Rows}(D,n)}\tilde{G}^{\tensor n^{1/D}}_{n^{1-1/D}, D-1, c,s})^c - G_\Haar^{(t)}\right] \right\|_\diamond$. We first relate this expression to the superoperator $\Channel\left[G_{\text{Planes}(D)}\right]$. Using triangle inequality and Lemma \ref{lem:bv}:
\bea
&&\left\|  \Channel\left[(\tilde{G}^{\tensor n^{1/D}}_{n^{1-1/D},D-1, c, s} G_{\text{Rows}(D,n)}\tilde{G}^{\tensor n^{1/D}}_{n^{1-1/D}, D-1, c,s})^c - G_\Haar^{(t)}\right] \right\|_\diamond \nonumber\\
&\leq& \left\|  \Channel\left[ ( G_{\text{Rows}(D,n)}\tilde{G}^{\tensor n^{1/D}}_{n^{1-1/D}, D-1, c,s}G_{\text{Rows}(D,n)})^{c-1} - G_\Haar^{(t)}\right] \right\|_\diamond\nonumber\nonumber\\
&\leq& \left\|  \Channel\left[( G_{\text{Rows}(D,n)} (\tilde{G}^{\tensor n^{1/D}}_{n^{1-1/D}, D-1, c,s}- G_{\text{Planes({D,n})}}) G_{\text{Rows}(D,n)}
 +G_{\text{Rows}(D,n)} G_{\text{Planes}(D)} G_{\text{Rows}(D,n)} )^{c-1} - G_\Haar^{(t)}\right] \right\|_\diamond\nonumber\nonumber\\
&\leq& O\left (\| \Channel\left[\tilde{G}^{\tensor n^{1/D}}_{n^{1-1/D}, D-1, c,s}- G_{\text{Planes}(D)}\right] \|_\diamond\right)+\left\|\Channel\left[( G_{\text{Rows}(D,n)} G_{\text{Planes}(D)} G_{\text{Rows}(D,n)} )^{c-1} - G_\Haar^{(t)}\right] \right\|_\diamond\nonumber\nonumber\\
&\leq& O(n) \left\| \Channel\left[\tilde{G}_{n^{1-1/D}, D-1, c,s}- G_{\Haar(p_1)}\right] \right\|_\diamond+\left\|\Channel\left [( G_{\text{Rows}(D,n)} G_{\text{Planes}(D)} G_{\text{Rows}(D,n)} )^{c-1} - G_\Haar^{(t)}\right] \right\|_\diamond\nonumber\nonumber\\
&\leq& O(n) \eps_{D-1, n^{1-1/D}} +\left\|\Channel\left [ ( G_{\text{Rows}(D,n)} G_{\text{Planes}(D)} G_{\text{Rows}(D,n)} )^{c-1} - G_\Haar^{(t)}\right] \right\|_\diamond\nonumber\nonumber\\
&\leq& O(n) \frac{1}{d^{n^{1/D}}} +\left\|( \Channel\left[G_{\text{Rows}(D,n)} G_{\text{Planes}(D)} G_{\text{Rows}(D,n)} )^{c-1} - G_\Haar^{(t)}\right] \right\|_\diamond\nonumber\nonumber\\
&\leq& \frac{1}{d^{n^{1/D}}} +\left\|( \Channel\left[G_{\text{Rows}(D,n)} G_{\text{Planes}(D)} G_{\text{Rows}(D,n)})^{c-1} - G_\Haar^{(t)}\right] \right\|_\diamond.\nonumber
\eea
The first line is by Lemma \ref{lem:diamondpalace}. The third line is by triangle inequality and Lemma \ref{lem:diamondpalace}. The fourth line is by Lemma \ref{lem:tensorerror} and that $\tilde{G}_{\text{Planes}(D)}$ is a tensor product of Haar moment operators. Note in the sixth line we have used the induction hypothesis: $\eps_{D-1, n^{1-1/D}} = 1/d^{O(\frac{n^{1-1/D}}{D-1})} = \frac{1}{d^{\Omega(n^{1/D})}}$.

Using Lemma \ref{lem:diamondD} $\left\|( \Channel[G_{\text{Rows}(D,n)}] \tilde{G}_{\text{Planes}(D)} \Channel[G_{\text{Rows}(D,n)}] )^{c-1} - \Channel[G_\Haar^{(t)}] \right\|_\diamond = \frac{1}{d^{\Omega(n^{1/D})}}$ and this completes the proof.

\item Define $\eps_{D,n} := \left\| G_{\mu^{\lattice,n}_{2,c,s}}^{(t)} - G_\Haar^{(t)}
  \right\|_1$. By induction assume $\eps_{D-1, n} = 1/d^{\Omega(n^{1/{D-1}})}$ for all $n$.
 We would like to show that $\eps_{D,n} = 1/d^{\Omega(n^{1/D})}$.

\bea
\eps_{D,n} &:= &\left\| G_{\mu^{\lattice,n}_{D,c,s}}^{(t)} - G_\Haar^{(t)} \right\|_1 \nonumber\\
&=& \left\|G^{(t) \tensor n^{1/D}}_{\mu^{\lattice,n^{1-1/D}}_{D-1,c,s}}  (g_{\text{Rows}(D,n)}^s G^{(t)\tensor n^{1/D}}_{\mu^{\lattice,n^{1-1/D}}_{D-1,c,s}} )^c-G_\Haar^{(t)}\right\|_1
\label{eq:blahld}
\eea
Write $G^{(t) \tensor n^{1/D}}_{\mu^{\lattice,n^{1-1/D}}_{D-1,c,s}} = 
G_{\text{Planes}(D)}  + (G^{(t) \tensor n^{1/D}}_{\mu^{\lattice,n^{1-1/D}}_{D-1,c,s}} -
G_{\text{Planes}(D)})=: Z_0 + Z_1$. Our strategy is to expand \eqref{eq:blahld} in
terms of $G_{\text{Planes}(D)}$: 
\ba 
& \left\|(\delta + G_{\text{Planes}(D)})  (g_{\text{Rows}(D,n)}^s
 (\delta + G_{\text{Planes}(D)}) )^c-G_\Haar^{(t)}\right\|_1 \nonumber\\
= &\sum_{\phi \in \{0,1\}^{c+1}}
\left\|Z_{\phi_0} \prod_{i=1}^c (g_{\text{Rows}(D,n)}^s Z_{\phi_i}) -  G_\Haar^{(t)} \right\|_1
\nonumber\\ \leq &
 \underbrace{\sum_{\phi \in \{0,1\}^{c+1} \backslash 0^{c+1}}
\left\|Z_{\phi_0} \prod_{i=1}^c (g_{\text{Rows}(D,n)}^s Z_{\phi_i}) \right\|_1}_{(1)}
+ \underbrace{\left\|Z_0 (g_{\text{Rows}(D,n)}^s Z_0)^c -  G_\Haar^{(t)} \right\|_1}_{(2)}
\ea
To bound (1), observe that each term contains at least one $Z_1$. We
would like to bound $\|Z_1\|_1$.  
Observe that $G_{\text{Planes}} =G^{(t)\tensor
  n^{1/D}}_{\Haar(n^{1-1/D})}$, so
\ba \|Z_1\|_1 & = 
  \left\|G^{(t) \tensor n^{1/D}}_{\mu^{\lattice,n^{1-1/D}}_{D-1,c,s}} -
G^{(t)\tensor n^{1/D}}_{\Haar(n^{1-1/D})}  \right\|_1\nonumber  \\
&  = \left\| \sum_{i=1}^{n^{1/D}}
G^{(t) \tensor i-1}_{\mu^{\lattice,n^{1-1/D}}_{D-1,c,s}}
(G^{(t) }_{\mu^{\lattice,n^{1-1/D}}_{D-1,c,s}} -
G^{(t)}_{\Haar(n^{1-1/D})})
G^{(t)\tensor n^{1/D}-i}_{\Haar(n^{1-1/D})}
  \right\|_1\nonumber  \\
&  \leq \sum_{i=1}^{n^{1/D}}
\left\|G^{(t)}_{\mu^{\lattice,n^{1-1/D}}_{D-1,c,s}}\right\|_1^{i-1}
\left\| G^{(t) }_{\mu^{\lattice,n^{1-1/D}}_{D-1,c,s}} -
G^{(t)}_{\Haar(n^{1-1/D})}\right\|_1
\left\|G^{(t)}_{\Haar(n^{1-1/D})}\right\|_1^{ n^{1/D}-i}
\label{eq:mult-1-norm}\\
&  \leq \sum_{i=1}^{n^{1/D}}
(t! + \eps_{D-1,n})^{i-1}
\eps_{D-1,n}
t!^{n^{1/D}-i}
\label{eq:induction-1-norm}\\
&  \leq n^{1/D} 
(t! + \eps_{D-1,n})^{n^{1/D}}
d^{-\Omega(n^{1/(D-1)})}.
\ea
This final expression is $\leq d^{-\Omega(n^{1/D})}$ for $n$
sufficiently large relative to $d,t,D$.  Eq.~\eqref{eq:induction-1-norm}
uses the induction hypothesis as well as the fact that 
$G^{(t)}_{\Haar(m)}$ is a projector of rank $\leq
t!$ for any $m$.  (In fact this is an equality when $m \geq \ln(t)$.)
This last fact is standard and can be found in Lemma 17 of
\cite{BHH-designs}, with the relevant math background in \cite{GW98,Har-sym}.

For (2), we observe that $(G_{\text{Planes}(D)} g_{\text{Rows}(D,n)}^s)^c- G_\Haar^{(t)}$ has rank $t!^{O(n^{1/D})}$ so the cost of moving to the infinity norm is moderate:
\bea
 \left\|(G_{\text{Planes}(D)} g_{\text{Rows}(D,n)}^s)^c- G_\Haar^{(t)}\right\|_1 &\leq& t!^{O(n^{1/D})} \left\|(G_{\text{Planes}(D)} g_{\text{Rows}(D,n)}^s)^c- G_\Haar^{(t)}\right\|_\infty\\
  &=& t!^{O(n^{1/D})} \left\|G_{\text{Planes}(D)} g_{\text{Rows}(D,n)}^s- G_\Haar^{(t)}\right\|^c_\infty
\eea
We now bound $\left\|G_{\text{Planes}(D)} g_{\text{Rows}(D,n)}^s-
G_\Haar^{(t)}\right\|_\infty$ using a variant of the proof of part 3 of this theorem.
\be
\left\|G_{\text{Planes}(D)} g_{\text{Rows}(D,n)}^s- G_\Haar^{(t)}\right\|_\infty
\leq \left\|g_{\text{Rows}(D,n)}^s- G_{r_{1,1}}^{\tensor  n^{1-1/D}}\right\|_\infty
+ \left\|G_{\text{Planes}(D)} G_{r_{1,1}}^{\tensor  n^{1-1/D}}-
G_\Haar^{(t)}\right\|^c_\infty
\label{eq:1-norm-bound}
\ee
Using \cite{BHH-designs} and Lemma \ref{lem:tensorerror}
$\left\|g_{\text{Rows}(D,n)}^s- G_{r_{1,1}}^{\tensor  n^{1-1/D}}\right\|_\infty
\leq O(n^{1-1/D})\left\|g_{r_{1,1}}^s- G_{r_{1,1}}\right\|_\infty =
\frac{1}{d^{\Omega(n^{1/D})}}$. Moreover, using lemma
\ref{lem:generalDproj} $\left\|G_{\text{Planes}(D)} G_{r_{1,1}}^{\tensor
  n^{1-1/D}}- G_\Haar^{(t)}\right\|^c_\infty =
\frac{1}{d^{\Omega(n^{1/D})}}$. 

This completes the proof by taking the constant in the
$\Omega(n^{1/D})$ in the last exponent sufficiently larger than the
constant in the $O(n^{1/D})$ exponent in \eqref{eq:1-norm-bound}.
Here we are ignoring the dependence on $d,t,D$.  Taking this into
account properly would yield a depth that scales polynomially with
with $t$, with the degree of the polynomial depending on $D$.

\item Define $\eps_{D,n} := \left\| G_{\mu^{\lattice,n}_{2,c,s}}^{(t)} - G_\Haar^{(t)} \right\|_\infty$. By induction assume $\eps_{D,n} = 1/d^{\Omega(n^{1/D})}$ for any integers $n$ and $D$. Assuming $\eps_{D-1, n} = 1/d^{\Omega(n^{1/{D-1}})}$ for all $n$, we show that $\eps_{D,n} = 1/d^{\Omega(n^{1/D})}$.
\bea
\eps_{D,n} &:= &\left\| G_{\mu^{\lattice,n}_{2,c,s}}^{(t)} - G_\Haar^{(t)} \right\|_\infty \nonumber \\
&\leq& \left\|  G_{\mu^{\lattice,n}_{D,c,s}} - \tilde{G}_{n,D,c}  \right\|_\infty + \left\| \tilde{G}_{n,D,c} - G_\Haar^{(t)} \right\|_\infty\nonumber \\
&\leq&\poly(n) \cdot \left\|  (g^s_{r_{1,1}})^{\tensor n^{1-1/D}} - G_{r_{1,1}}^{\tensor n^{1-1/D}} \right\|_\infty\nonumber \\
 &&+  \left\|G_{\text{Rows}(D,n)} \tilde{G}_{n^{1-1/D},D-1,c} ^{\otimes n^{1/D}} - G_\Haar^{(t)} \right\|^c_\infty\nonumber \\
&\leq&\poly(n) \cdot \left\|  (g^s_{r_{1,1}})^{\tensor n^{1-1/D}} - G_{r_{1,1}}^{\tensor n^{1-1/D}} \right\|_\infty\nonumber \\
&& + \Big\|G_{\text{Rows}(D,n)} (\tilde{G}_{n^{1-1/D},D-1,c} ^{\otimes n^{1/D}} -F_{\text{Rows}(D,n)})\nonumber \\
 &&+ G_{\text{Rows}(D,n)} F_{\text{Rows}(D,n)}- G_\Haar^{(t)} \Big\|^c_\infty\nonumber \\
&\leq&O(n) 1/d^{\Omega(n^{1/D})}\nonumber \\
 &&+  O(n) \eps_{D-1, n^{1-1/D}} + 1/d^{\Omega(n^{1-1/D}) c}\nonumber \\
&\leq& d^{- \Omega(n^{1/D})}+ 1/d^{\Omega(n^{1/D})} + 1/d^{\Omega(n^{1-1/D}) c}\nonumber \\
&\leq& d^{- \Omega(n^{1/D})}.
\eea
These steps follow from the proof of part 2.

\end{enumerate}
\end{proof}


\newpage

\subsection{Proofs of the basic lemmas stated in Section \ref{sec:basic-lemma}}
\label{sec:lemma-proofs}

\subsubsection{Comparison lemma for random quantum circuits}
\label{sec:comparison-proofs}

\restatelemma{lem:comparison}

\begin{proof}
We first prove the following claim
\begin{claim}
If $\cA \preceq \cB$ and $\cC \preceq D$ are cp maps, then $\cA \cC \preceq \cB\cD$.
\end{claim}
\proof{The class of cp maps is closed under composition and addition. Therefore $\cB \cD - \cA\cC =  (\cB-\cA) D + \cA (\cD-\cC)$ is cp.}

The proof (of Lemma \ref{lem:comparison}) is by induction. We show for all $1 \leq i \leq t$
\be
\cA_i \ldots \cA_1 \preceq \cB_i \ldots \cB_1.
\ee
Clearly this is true for $i=1$. Suppose also this is true for $1 < k < t$. So $\cA_i \ldots \cA_1 \preceq \cB_i \ldots \cB_1$ and $\cA_{i+1} \preceq \cB_{i+1}$, and using the claim $\cA_{i+1} \ldots \cA_1 \preceq \cB_{i+1} \ldots \cB_1$.
\end{proof}

\restatecorollary{cor:overlappingdesigns}

\begin{proof}
This is immediate from Lemma \ref{lem:comparison}, Definition \ref{def:moments}, and the observation that if $A \preceq B$ then $A \tensor \id \preceq B \tensor \id$.
\end{proof}

\subsubsection{Bound on the value of off-diagonal monomials}
\label{sec:equiv}

\restatelemma{lem:offdiag}

\begin{proof}
Let $\phi_N := \ket{\phi_N} \bra{\phi_N}$ for 
\be
\ket{\phi} := \frac{1}{\sqrt{N}} \sum_{x\in [d]^n} \ket{x}\ket{x}
\ee
 be the $n$-qudit maximally entangled state, and $N = d^n$.
 
We use the following standard lemma which we leave without proof (see \cite{BHH-designs} for e.g.)
\begin{lemma} \OldNormalFont{}
Let $\mu$ and $\nu$ be two distributions over the $n$-qudit unitary group then $\Channel[G^{(t)}_\mu] \preceq \Channel[G^{(t)}_\nu]$ if and only if 
\be
 \left (\Channel\left [G^{(t)}_\nu\right]\tensor \mathrm{id} - \Channel\left[G^{(t)}_\mu\right] \tensor \mathrm{id} \right) \phi^{\tensor t}_N
\ee
 is a psd matrix.
 \label{lem:psd1}
\end{lemma}

We now adapt \lemref{offdiag} to Lemma \ref{lem:psd1}. First,
\be
\phi_N^{\tensor t} = \frac{1}{N^{t}} \sum \ket{i_1, \ldots , i_t}\bra{j_1, \ldots , j_t}\tensor \ket{i_1, \ldots , i_t}\bra{j_1, \ldots , j_t} \equiv \frac{1}{N^{t}} \sum \ket{i}\bra{j}\tensor \ket{i}\bra{j}.
\ee
For $i,j,k,l \in [d]^{nt}$, if we define 
\be
M^{(\mu,t)}_{k, i,l,j} = \bra{k}\Channel\left [G^{(t)}_\mu\right]  \left(\ket{i}\bra{j} \right) \ket{l}.
\ee
Therefore
\be
 (\Channel[G_\mu^{( t)}] \tensor \mathrm{id} ) \phi_N^{\tensor t} = \frac{1}{N^{t}} \sum M^{(\mu,t)}_{a,b,c,d}  \ket{a}\bra{c}\tensor \ket{b}\bra{d},
\ee
and
\be
 (\Channel[G^{(t)}_\Haar]  \tensor \mathrm{id} ) \phi_N^{\tensor t}= \frac{1}{N^{t}} \sum M^{(\mathrm{Haar},t)}_{a,b,c,d}  \ket{a}\bra{c}\tensor \ket{b}\bra{d}.
\ee
Therefore since $\Channel\left[G_\mu^{(t)}\right] \leq (1+\delta) \Channel[G_\nu^{( t)}]$ the following matrix
\be
A =  ( (1+\delta) \Channel[G_\Haar^{(t)} ]\tensor \mathrm{id}-\Channel\left[G_\mu^{(t)}\right] \tensor \mathrm{id} ) \phi_N^{\tensor t}= \frac{1}{N^{t}} \sum ( (1+\delta) M^{(\mathrm{Haar},t)}_{a,b,c,d} - M^{(\mu,t)}_{a,b,c,d} )  \ket{a} \ket{b}\bra{c}\bra{d}.
\label{eq:sd}
\ee
Is psd. We use the following fact about psd matrices which we leave without proof.

\textbf{Fact}--- if $A$ is psd then the absolute maximum of off-diagonal terms in $A$ is at most the absolute maximum diagonal term. 

Then using the above fact
\be
\max_{x \in \cO} |(1 + \delta) \Tr \left( G^{(t)}_\nu x \right) - \Tr \left( G^{(t)}_\mu x \right)| \leq \max_{y \in \cD} |(1 + \delta) \Tr ( G^{(t)}_\nu y) - \Tr \left( G^{(t)}_\mu y \right)|.
\ee
Hence
\be
\max_{x \in \cO} | \Tr \left( G^{(t)}_\mu x \right) | \leq \max_{x \in \cO} | \Tr \left( G^{(t)}_\nu x \right) | (1+\delta) + \max_{y \in \cD} |(1 + \delta)  \Tr \left( G^{(t)}_\nu y \right)  -  \Tr \left( G^{(t)}_\mu y \right) |.
\label{eq:dsss}
\ee

Now if $y \in \cD$
\be
 \Tr \left( G^{(t)}_\nu x \right) (1-\delta) \leq  \Tr \left( G^{(t)}_\mu x \right)  \leq  \Tr \left( G^{(t)}_\nu x \right) (1+\delta).
\ee
then using this in \eqref{eq:dsss}
\be
\max_{x \in \cO} | \Tr \left( G^{(t)}_\mu x \right) | \leq \max_{x \in \cO} | \Tr \left( G^{(t)}_\nu x \right) | (1+\delta) + 2\delta \cdot \max_{y \in \cD}  \Tr \left( G^{(t)}_\nu y \right) .
\ee
\end{proof}

\subsubsection{Bounds on the moments of the Haar measure}

\restatelemma{lem:Haarmoment}

\begin{proof}
First observe that
\bea
\max_y \|G_\Haar^{(t)} y G_\Haar^{(t)}\|_1 &=& \max_{a , b} \Tr \sqrt{G_\Haar^{(t)} \ket{a}\bra{b} G_\Haar^{(t)} G_\Haar^{(t)} \ket{b}\bra{a} G_\Haar^{(t)}}\\
&=& \max_{a , b} \sqrt{\braket{a| G_\Haar^{(t)} |a}\cdot \braket{b| G_\Haar^{(t)} |b }}\\
&=& \max_{i} \braket{i|G_\Haar^{(t)}|i}.
\eea
The below lemma concludes the proof.
\begin{lemma}[Moments of the Haar measure] \OldNormalFont{} The largest $t$-th monomial moment of the Haar measure is at most $\frac{t!}{d^{tm}}$.
\label{lem:Haarlargest}
\end{lemma}
\begin{proof}
Consider a particular balanced moment of Haar, using H\"older's inequality
\be
\E_{C\sim\Haar} | C_{a_1, b_1} \ldots C_{a_t,b_t}|^2 \leq \prod_{i \in [k]} (\E |C_{a_i,b_i}|^{2k})^{1/k} \leq k! / d^{km}.
\ee
If the moment is not balanced the expectation is zero and hence the bound still works. Here we have used a closed form expression for $\E |C_{a_i,b_i}|^{2k}$, see corollary 2.4 and proposition 2.6 of \cite{C06} for a reference.
\end{proof}

\end{proof}

\subsection{Proofs of the projector overlap lemmas from section
  \ref{sec:overlap-lemmas}}
\label{sec:overlap-proofs}

\subsubsection{Extended quasi-orthogonality of permutation operators with application to random circuits on $2$-dimensional lattices}
\label{overlap}
In this section we prove Lemma \ref{lem:2Dproj}

\restatelemma{lem:2Dproj}


First, we need a description of the subspaces the projectors $G_R, G_C$ and $G_\Haar$ project onto. Consider a $\sqrt{n}\times \sqrt{n}$ square lattice with $n$ qudits as the collection of points $A := [\sqrt{n}]\times [\sqrt{n}]$. We use the following interpretation of the Hilbert space a quasi-projector acts on. This interpretation is also used in \cite{BHH-designs}. Denote $R_j (C_j)$ as the $j$-th row (column) of $A$ for $j \in [\sqrt{n}]$. Here we assume each point of $A$ consists of $t$ pairs of qudits, each with local dimension $d$. Thereby, the lattice becomes the Hilbert space 
$\mathcal{H}:= \bigotimes_{(x,y) \in A} \C_{(x,y)}^{d^{2 t}}$, 
and has dimension $d^{2 t n}$.

We are interested in a certain subspace of $\cH$, and in order to understand it we need the following notation. For each point $(x,y) \in A$ we assign the quantum state $\ket{\psi_\pi}:= (I\tensor V(\pi))\ket{\Phi_{d,t}}$, for each permutation $\pi \in S_t$. $\ket{\Phi_{d,t}}$ is the maximally entangled state $\frac{1}{\sqrt{d^t}}\sum_{x \in [d]^t} \ket{x,x}$, $V: S_t  \rightarrow GL(\C^{d^{2t}})$, is a representation of $S_t$ with the map $V(\pi): \ket{x(1), x_2, \ldots, x_t} \mapsto \ket{x_{\pi^{-1}(1)}, x_{\pi^{-1}(2)}, \ldots, x_{\pi^{-1}(t)}}$, and $S_t$ is the symmetric group over $t$ elements.

Given these definitions define the following basis states in $\mathcal{H}$:
\be
\ket{R_{\pi_1, \pi_2, \ldots, \pi_{\sqrt{n}}}} :=  \bigotimes_{v_1 \in R_1} \ket{\psi_{\pi_1}}_{v_1} \tensor  \bigotimes_{v_2 \in R_2} \ket{\psi_{\pi_2}}_{v_2} \tensor \ldots \tensor  \bigotimes_{v_{\sqrt{n}} \in R_{\sqrt{n}}} \ket{\psi_{\pi_{\sqrt{n}}}}_{v_{\sqrt{n}}},
\label{wow1}
\ee
\noindent and,
\be
\ket{C_{\pi_1, \pi_2, \ldots, \pi_{\sqrt{n}}}} :=  \bigotimes_{v_1 \in C_1} \ket{\psi_{\pi_1}}_{v_1} \tensor  \bigotimes_{v_2 \in C_2} \ket{\psi_{\pi_2}}_{v_2} \tensor \ldots \tensor  \bigotimes_{v_{\sqrt{n}} \in C_{\sqrt{n}}} \ket{\psi_{\pi_{\sqrt{n}}}}_{v_{\sqrt{n}}},
\label{wow2}
\ee
for each $\sqrt{n}$ tuple of permutations $(\pi_1, \pi_2, \ldots, \pi_{\sqrt{n}}) \in  S_t^{\sqrt{n}}$. Here $S_t^{\sqrt{n}}$ is the $\sqrt{n}$-fold Cartesian product $S_t \times \ldots \times S_t$ of $S_t$ with itself. Denote $H_{t,n}$ as the subset consisting of tuples of permutations in which not all of the permutations are equal. For example, elements like $(\pi, \pi, \ldots, \pi)$ are not contained in this set.  Notice that these basis are not orthogonal to each other and if $t>d^n$ these are not even linearly independent.

Here we define two vector spaces $V_R , V_C \subseteq \mathcal{H}$, with:
\be
V_R:=\operatorname*{span}_\C  \left\{ \ket{R_{\pi_1, \pi_2, \ldots, \pi_{\sqrt{n}}}}: (\pi_1, \pi_2, \ldots, \pi_{\sqrt{n}}) \in S_t^{\sqrt{n}} \right\},
\ee
and,
\be
V_C:=\operatorname*{span}_\C  \left \{ \ket{C_{\pi_1, \pi_2, \ldots, \pi_{\sqrt{n}}}}: (\pi_1, \pi_2, \ldots, \pi_{\sqrt{n}}) \in S_t^{\sqrt{n}} \right\},
\ee
and we call them row and column vector spaces, respectively. Also, denote the intersection between them by $V_\Haar := V_R \cap V_C$. Equivalently:
\be
V_\Haar= \operatorname*{span}_\C  \left \{ \bigotimes_{v\in A} \ket{\psi_\pi}_v: \pi \in S_t \right\}.
\ee
Then define $\tilde{V}_R:= V_R\cap V^\perp_H$ and $\tilde{V}_C:= V_C\cap V^\perp_H$. Define the angle between two vector spaces $A$ and $B$ as 
\be
\cos \measuredangle (A,B):= \max_{x\in A, y \in B} \braket{x,y}.
\ee

We need the following definition of a Gram matrix
\begin{definition}[Gram matrix] \OldNormalFont{}
Let $v_1, \ldots, v_{\text{Rows}(D,n)}$ be normal vectors that are not necessarily orthogonal to each other. Then the Gram matrix corresponding to this set of vectors is defined as $[J_{ij}] = \braket{v_i | v_j}$.
\end{definition}

We also need the following lemma
\begin{lemma}\OldNormalFont{}
(Perron-Frobenius \cite{boaz-PF}) If $A$ is a (not necessarily symmetric) $d$-dimensional matrix, then:

\be
||A||_\infty \leq \sqrt{\max_{i\in [d]} \sum_j |A_{i,j}|\cdot \max_{j\in [d]} \sum_i |A_{i,j}|}.
\ee
\label{sumofrows}
\end{lemma}

Let $G_R, G_C$ and $G_{\Haar}$ be the quasi-projectors defined in Section \ref{sec:definitions}. From \cite{BHH-designs} we know that $G_R$, $G_C$ and $G_\Haar$ are indeed projectors onto $V_R, V_C$ and $V_\Haar$, respectively. Define the inner-product matrix between $V_R$ and $V_C$ with matrix $Q$ with entries:

\be
[Q]_{g,h}:=\braket{R_g|C_h}, \hspace{1cm} g,h \in H_{t,n}.
\ee

The goal is to prove $\|G_C G_R -G_\Haar^{(t)}\|_\infty \leq 1/d^{\Omega(\sqrt{ n})}$.  This basically means that the composition of $G_R$ and $G_C$ is close to $G_{\Haar}$.

Also let $c_{d,n,t} =\frac{1}{1-\frac{\sqrt n t(t-1)}{2 d^{\sqrt n}}}$ 
be a number very close to $1$.

The proof is in three main steps. First we relate $\|G_CG_R-G_{\Haar}\|_\infty$ to $\measuredangle (\tilde V_R, \tilde V_C)$:
\begin{proposition}\OldNormalFont{}
$\|G_CG_R-G_{\Haar}\|_\infty \leq \cos^2 \measuredangle (\tilde V_R, \tilde V_C)$.
\label{prop:angle}
\end{proposition}

Next, we relate $\measuredangle (\tilde V_R, \tilde V_C)$ to $||Q||_\infty$
\begin{proposition}\OldNormalFont{}
$|\cos \measuredangle (\tilde V_R, \tilde V_C)| \leq c_{d,n,t} . ||Q||_\infty$
\label{gapandQ}
\end{proposition}

Then we bound $||Q||_\infty$:
\begin{proposition}\OldNormalFont{}
$||Q||_\infty\leq    (\frac{1}{d}+\frac{1}{d^{\sqrt{n}-1}}+\frac{2 t^2}{d^{\sqrt{n}}}   )^{\sqrt{n}}$. 
\label{Qissmall}
\end{proposition}

Propositions \ref{prop:angle}, \ref{gapandQ} and \ref{Qissmall} imply the proof of Lemma \ref{lem:2Dproj}.

\begin{proof}[\OldNormalFont{} Proof of Proposition \ref{prop:angle}]
We use the following result of Jordan
\begin{proposition} (Jordan) if $P$ and $Q$ are two projectors, then the Hilbert space $V$ they act on can be decomposed, as a direct sum, into one-dimensional or two-dimensional subspaces, all of which are invariant under the action of both $P$ and $Q$ at the same time.
\end{proposition}

\noindent which implies
\begin{corollary}\OldNormalFont{}
There are orthonormal basis $e_1, \ldots, e_K$, $f_1, \ldots, f_K$, $q_1, \ldots, q_T$, and angles $0\geq \theta_1 \geq \theta_2 \geq \ldots \geq \theta_K \leq \pi/2$ such that:
\be
V_R = \operatorname*{span}_{\C}  \left\{e_1,\ldots, e_K, q_1, \ldots, q_T  \right\},
\ee
and
\be
V_C = \operatorname*{span}_{\C}  \{\cos \theta_1 e_1 + \sin \theta_1 f_1,\ldots, \cos \theta_K e_K + \sin \theta_K f_K, q_1, \ldots, q_T  \},
\ee
and
\be
V_\Haar = \operatorname*{span}_{\C}  \{q_1, \ldots, q_T  \}.
\ee
\label{jordan}
\end{corollary}

In other words, both $G_R$ and $G_C$ can be decomposed into $2\times 2$ blocks, each corresponding to one of the angles $\theta_i$, such that $G_C$ on this block looks like
\be
G^{2 \times 2}_C=
\begin{pmatrix}
1 & 0\\
0& 0\\
\end{pmatrix},
\ee
and $G_R$
\be
G^{2\times 2}_R =
\begin{pmatrix}
\cos^2 \theta_i & \sin \theta_i \cos \theta_i\\
\sin \theta_i \cos \theta_i & \sin^2 \theta_i\\
\end{pmatrix}.
\ee

Hence $G_C G_R$ looks like
\be
G^{2\times 2}_C G^{2\times 2}_R  =
\begin{pmatrix}
\cos^2 \theta_i & \sin \theta_i \cos \theta_i\\
0 & 0\\
\end{pmatrix},
\ee
which has largest singular value $|\cos^2 \theta_i|$. Propositions \ref{gapandQ} and \ref{Qissmall} along with this observation imply that the largest singular value of $G_C G_R - G_\Haar$ is $1/d^{n^{O(n^{1/D})}}$.

\end{proof}

\begin{proof} [\OldNormalFont{} Proof of Proposition \ref{gapandQ}]
%
%
%
%
%
An arbitrary normal vector in $\tilde{V}_R$ can be written as $\ket {\psi_x} = \frac{\sum_{\tilde \pi \in H_{t,n}} x_{\tilde \pi} \ket {R_{\tilde \pi}}}{\sqrt {\sum_{\tilde \pi, \tilde \sigma \in H_{t,n}}     x_{\tilde \pi} x_{\tilde \sigma} \braket {R_{\tilde \pi}| R_{\tilde{\sigma}}    }}}$. Let $\ket{x}$ be a vector with entries corresponding to $x_{\tilde{\pi}_1, \ldots , \tilde \pi_{\sqrt n}}$. Similarly, a typical vector inside $\tilde V_C$ can be represented as $\ket {\psi_y} = \frac{\sum_{\tilde \pi \in H_{t,n}} y_{\tilde \pi} \ket {R_{\tilde \pi}}}{\sqrt {\sum_{\tilde \pi, \tilde \sigma \in H_{t,n}}     y_{\tilde \pi} y_{\tilde \sigma} \braket {C_{\tilde \pi}| C_{\tilde{\sigma}}    }}}$. Also represent the corresponding vector $\ket{y}$ similarly. 

 Let $\tilde{J}$ and $\tilde{J}'$ be the Gram matrices corresponding to the basis described for $\tilde{V}_R$ and $\tilde{V}_C$, respectively. Then:
 \be
 \braket {\psi_x | \psi_y} = \frac{\sum_{\tilde \pi, \tilde \sigma \in H_{t,n}} x_{\tilde \pi} \braket {R_{\tilde \pi}| R_{\tilde \sigma}} y_{\tilde \sigma}}{\sqrt {\sum_{\tilde \pi, \tilde \sigma \in H_{t,n}}     x_{\tilde \pi} x_{\tilde \sigma} \braket {R_{\tilde \pi}| R_{\tilde{\sigma}}    }} \cdot \sqrt {\sum_{\tilde \pi, \tilde \sigma \in H_{t,n}}     y_{\tilde \pi} y_{\tilde \sigma} \braket {C_{\tilde \pi}| C_{\tilde{\sigma}}    }}} = \frac{\bra{x}Q \ket{y}}{\sqrt{\bra{x} \tilde{J}\ket{x}}.\sqrt{\bra{y}\tilde{J}'\ket{y}}}.
 \ee
 
To see the equality we go through the below calculation. 

\begin{eqnarray}
\cos \phi &=& \sup_{\mathclap{\substack{||x||_2=1\nonumber \\
                   ||y||_2=1}}}\hspace{0.3cm} \frac{\bra{x}Q \ket{y}}{\sqrt{\bra{x} \tilde{J}\ket{x}}.\sqrt{\bra{y}\tilde{J}'\ket{y}}}\nonumber \\
\vspace{1mm}&\leq& c_{d,n,t}\hspace{1mm} .\hspace{2mm} \sup_{\mathclap{\substack{||x||_2=1\nonumber \\
                   ||y||_2=1}}}\hspace{0.3cm} \bra{x}Q \ket{y}\nonumber \\
&\leq& c_{d,n,t}\hspace{1mm} .\hspace{2mm} \sup_{\mathclap{\substack{||x||_2=1\nonumber \\
                   ||y||_2=1}}}\hspace{0.3cm} \sqrt{\bra{x}Q^\dagger Q\ket{x}} \hspace{1mm}. \hspace{1mm} ||y||_2\nonumber \\
&\leq& c_{d,n,t} \hspace{1mm}. \hspace{1mm} ||Q||_\infty.
\end{eqnarray}

For the second line we used the following proposition
\begin{proposition}\OldNormalFont{}
If $\tilde{J}$ is the Gram matrix for the basis states, $\ket{R_{\cdot}}$ or $\ket{C_{\cdot}}$ in \eqref{wow1} and \eqref{wow2} for $\tilde{V}_R$ or $\tilde{V}_R$, then for any $|x\rangle$ with $||x||_2=1$:

\be
\bra{x}\tilde{J}\ket{x} \geq\left (1-\frac{\sqrt n t(t-1)}{2 d^{\sqrt n}}\right) =\frac{1}{c_{d,n,t}}.
\label{eq:196}
\ee
\label{Gramlowerbound}
\end{proposition}
For the third line we used Cauchy-Schwartz. 
\end{proof}

In order to prove Proposition \ref{Qissmall} we need the following tool. If $\vv{x(1)},\vv{x_2}, \ldots,\vv{x_K}$ are $d$-dimensional vectors the multi-product of them is defined to be:

\be
\mathsf{multiprod} (\vv{x_1},\vv{x_2}, \ldots,\vv{x_K}) := \sum^d_{i=1} x_{1 i} x_{2 i}\ldots x_{K i}.
\ee

\begin{proposition} [Majorization] \OldNormalFont{} Let $\vv{x}_1,\vv{x}_2, \ldots,\vv{x}_K$ be $d$-dimensional, non-negative and real vectors. If $\vv{x}_i^{\downarrow}$ is $\vv{x}_i$ in descending order, then:
\be
\mathsf{multiprod} \left (\vv{x}_1,\vv{x}_2, \ldots,\vv{x}_K\right)\leq \mathsf{multiprod} (\vv{x}_1^{\downarrow},\vv{x}_2^{\downarrow}, \ldots,\vv{x}_K^{\downarrow}).
\ee
\label{multiprod}
\end{proposition}

\begin{proof}
The $K=2$ version of claim is that $\langle \vec x(1), \vec x_2\rangle \leq \langle \vec x(1)^\downarrow, \vec x_2^\downarrow\rangle$.  This is a standard fact.  To prove it, observe that WLOG we can assume $\vec x(1) = \vec x(1)^\downarrow$.  Then for any out-of-order pair $x_{2i} < x_{2j}$ with $i<j$, we will increase $\langle \vec x(1), \vec x_2\rangle$ by swapping $x_{2i}$ and $x_{2j}$.  Applying this repeatedly we end with $\langle \vec x(1)^\downarrow, \vec x_2^\downarrow\rangle$.

This same argument works if we replace the inner product with a sum over the first $d' \leq d$ terms, i.e. $\sum_{i=1}^{d'} x_{1i}x_{2i}$.  Thus the same argument shows that
\be \vv{x}_1 \circ \vv{x}_2 \preceq \vv{x}_1^\downarrow \circ \vv{x}_2^\downarrow .
\ee

The proposition now follows by induction on $K$.
\end{proof}

We also need the following upper bound: 
\begin{proposition}\OldNormalFont{}
Let $e\in S_t$ be the identity permutation. Define $f_t: \R_{>1}  \rightarrow \R_{>1}$ with the map:

\be
f_t(\alpha)=\sum_{\sigma \in S_t} \frac{1}{\alpha^{\dist(e,\sigma)}},
\label{ft}
\ee

\noindent for $\alpha > 1$. Then as long as $2t^2 \leq \alpha$

\be
f_t(\alpha) \leq 1 + \frac{2 t^2}{\alpha}.
\ee
\label{ftlowerbound}
\end{proposition}

For $\sigma_1, \ldots, \sigma_M \in S_t$ define the function:
\begin{equation}
h(D,t, \sigma_1, \ldots, \sigma_M):= \sum_{\pi \in S_t} \frac{1}{D^{\dist(\pi, \sigma_1)+\ldots+\dist(\pi, \sigma_M)}}.
\end{equation}

\begin{proposition}\OldNormalFont{}
Let $(\sigma_1,\ldots, \sigma_M) \in H$ be permutations that not all of them are equal to each other then:

\be
h(D,t, \sigma_1,\ldots, \sigma_M) \leq \frac{1}{D}+\frac{1}{D^{M-1}} + \frac{2 t^2}{D^M}.
\ee
\label{maininequality}
\end{proposition}

\begin{proof} [\OldNormalFont{} Proof of Proposition \ref{Qissmall}]
In or to prove this, we show that the sum of terms in each row is a small number. Then use Lemma \ref{sumofrows} to obtain the result. Consider the particular row $(\sigma_1, \ldots, \sigma_{sqrt{n}}) \in H$, then the sum of terms in each row is:

\be
\sum_{(\pi_1, \ldots, \pi_{\sqrt{n}})\in H} \braket{R_{\pi_1, \ldots, \pi_{\sqrt{n}}}|C_{\sigma_1, \ldots, \sigma_{\sqrt{n}}}} = \sum_{\pi_1, \ldots, \pi_{\sqrt{n}}\in S_t} \braket{R_{\pi_1, \ldots, \pi_{\sqrt{n}}}|C_{\sigma_1, \ldots, \sigma_{\sqrt{n}}}}-\sum_{\pi\in S_t} \braket{R_{\pi, \ldots, \pi}|C_{\sigma_1, \ldots, \sigma_{\sqrt{n}}}}.
\ee

The lower bound:

\be
\sum_{\pi\in S_t} \braket{R_{\pi, \ldots, \pi}|C_{\sigma_1, \ldots, \sigma_{\sqrt{n}}}}\geq 0,
\ee

\noindent is good enough. The goal is to find a good upper bound for $S:=\sum_{\pi_1, \ldots, \pi_{\sqrt{n}}\in S_t} \braket{R_{\pi_1, \ldots, \pi_{\sqrt{n}}}|C_{\sigma_1, \ldots, \sigma_{\sqrt{n}}}}$. But $S$ simplifies to:

\be
S=  \left(\sum_{\pi \in S_t} \frac{1}{d^{\dist(\pi,\sigma_1)+\ldots+\dist(\pi, \sigma_{\sqrt{n}})}}  \right)^{\sqrt{n}} = h(d,t,\sigma_1, \ldots, \sigma_{\sqrt{n}})^{\sqrt{n}}.
\ee

\noindent Now we use Proposition \ref{maininequality} and find the upper bound:

\be
S\leq   \left(\frac{1}{d}+\frac{1}{d^{\sqrt{n}-1}}+\frac{2 t^2}{d^{\sqrt{n}}}   \right)^{\sqrt{n}}.
\ee

\noindent Which is a global maximum and in turn is a bound on the $\infty$-norm.
\end{proof}

\begin{proof} [\OldNormalFont{} Proof of Proposition \ref{Gramlowerbound}]
We will prove the statement for the row space, and the same thing works for the column space. First, for any normal vector $\ket x$, $\braket{x | \tilde J_R | x} \geq \lambda_{\min} (\tilde J_R)$. 
Let $J(\sqrt n)$ be the Gram matrix for the Haar subspace on one row of the grid. The entries of $J(\sqrt n)$ are according to:
\be
J(\sqrt n)_{\pi,\sigma} :=   \left ( \frac{1}{d^{\dist(\pi,\sigma)}} \right )^{\sqrt{n}}=  \left (\frac{1}{d^{\sqrt{n}}}  \right)^{\dist(\pi,\sigma)}.
\ee
Let $P$ be the projector that projects out the subspace spanned by $\{\ket{R_{\pi, \ldots, \pi}} : \pi \in S_t\}$. Then $\tilde J = P J(\sqrt n)^{\tensor \sqrt n} P^\dagger$. We first need the following proposition
\begin{proposition} \OldNormalFont{} If $J$ is the Gram matrix of the vector space spanned by $ \{ \ket{\psi_\pi}^{\tensor m}: \pi \in S_t \}.$, then:
\be
1- \frac{t(t-1)}{2 d^m}\leq \lambda_{\min}(J)
\ee
\label{SpectrumJ}
\end{proposition}
Using this proposition $\lambda_{\min} (J(\sqrt n)) \geq 1- \frac{t(t-1)}{2 d^{\sqrt n}}$, and therefore $\lambda_{\min} (J(\sqrt n)^{\tensor \sqrt n}) \geq (1- \frac{t(t-1)}{2 d^{\sqrt n}})^{\sqrt n}  \geq  1- \frac{\sqrt n t(t-1)}{2 d^{\sqrt n}}$. This implies that $J(\sqrt n)^{\tensor \sqrt n} \succeq I (1-\frac{\sqrt n t(t-1)}{2 d^{\sqrt n}})$, and therefore $\tilde J \succeq P P^\dagger (1-\frac{\sqrt n t(t-1)}{2 d^{\sqrt n}})$. This means that restricted to $\tilde V_R$ the minimum eigenvalue of $\tilde J_R$ is at least $(1-\frac{\sqrt n t(t-1)}{2 d^{\sqrt n}})$.

\end{proof}

\begin{proof}[\OldNormalFont{} Proof of Proposition \ref{maininequality}]
Let $C=   \{ \sigma_1,\ldots, \sigma_M  \}$. Then $h=h_1+h_2$, where:

\be
h_1 =  \sum_{\pi \in C} \frac{1}{D^{\dist(\pi, \sigma_1)+\ldots+\dist(\pi, \sigma_M)}},
\ee

\noindent and,

\be
h_2 =  \sum_{\pi \in S_t/C} \frac{1}{D^{\dist(\pi, \sigma_1)+\ldots+\dist(\pi, \sigma_M)}}.
\ee
We then find useful upper bounds for $h_1$ and $h_2$ separately. Suppose that $C$ has distinct elements $\{\tau_1,\ldots,\tau_K\}$ with $\tau_1$ appearing $\mu_1$ times, $\tau_2$ appearing $\mu_2$ times, etc.
Define 
\ba
 S & = \left \{ (\mu_1,\ldots, \mu_K) \in \bbZ_{\geq 0}^K: \mu_1 + \ldots + \mu_K =M, \max(i) \mu_i < M\right\} \\
P &= \Big\{(\mu_1, \ldots, \mu_K) \in S: \exists i,j, \hspace{1mm} \mu_i = M-1 \hspace{1mm} \& \hspace{1mm} \mu_j = 1\Big \}
\ea
Now we can bound $h_1$ by
\bea
h_1&=& \sum_{\pi \in C} \frac{1}{D^{\mu_1 \dist(\pi, \tau_1)+\ldots+\mu_K\dist(\pi, \tau_K)}}\nonumber \\
&\leq& \max_{(\mu_1,\ldots, \mu_K) \in S}\frac{D^{\mu_1}+ \ldots +D^{\mu_K}}{D^M}\nonumber \\
&\leq& \max_{(\mu_1,\ldots, \mu_K) \in \text{conv} (P)}\frac{D^{\mu_1}+ \ldots +D^{\mu_K}}{D^M} \label{eq:S-in-conv-P}\nonumber \\
&\leq& \frac{D^{M-1}+D}{D^M} \label{eq:max-extreme-pt}\nonumber \\
&=& \frac{1}{D}+\frac{1}{D^{M-1}}.
\eea

Here $\conv$ denotes the convex hull and \eq{S-in-conv-P} uses the fact that $K\geq 2$ since $\sigma_1,\ldots,\sigma_M$ are assumes to be not all equal.
To justify \eq{max-extreme-pt}, observe that $f (\mu) = D^{\mu_1} + \ldots + D^{\mu_K}$ is a convex function and the maximization is over a convex set whose extreme points are $P$. Therefore the maximum is achieved at a point in $P$. 

In order to find a bound on $h_2$, for each $\sigma\in C$ we will define the following vector $\vec X_\sigma$ whose entries are labeled by $\pi\in S_t$.
\be \vec X_{\sigma,\pi} = \begin{cases}
0 & \text{ if }\pi \in C \\
D^{-\dist(\sigma,\pi)}& \text{ if }\pi \not\in C
\end{cases}
\ee
Then $h_2 = \multiprod(\vec X_{\sigma_1},\ldots,\vec X_{\sigma_M})$.
We can use Proposition \ref{multiprod} to show that
\begin{equation}
h_2= \multiprod (\vv{X}_{\sigma_1},\ldots,\vv{X}_{\sigma_M} )
\leq\multiprod (\vv{X}^{\downarrow}_{\sigma_1},\ldots, \vv{X}^{\downarrow}_{\sigma_M} ).
\end{equation}

We will also define $\vec X_e$ (where $e$ denotes the identity element of $S_t$) by
\be \vec X_{e,\pi} = D^{-\dist(e,\pi)}.\ee
Observe that $\vec X_\sigma$ can be obtained from $\vec X_e$ by zeroing out the elements in locations corresponding to $C$ and reordering the remaining elements.  Thus
for each $\sigma\in C$
\be \vec X_{\sigma}^{\downarrow} \preceq \vec X_e^{\downarrow}.\ee
We use Proposition \ref{multiprod} again to bound
\be h_2 \leq \multiprod(\underbrace{\vec X_e,\ldots \vec X_e}_{M \text{ times}}) = 
f_t(D^{-M})- 1
\leq \frac{2 t^2}{D^M}
.\ee
\end{proof}

\subsubsection{Extended quasi-orthogonality of permutation operators with application to random circuits on $D$-dimensional lattices} 

In this section we prove lemmas \ref{lem:generalDproj}, \ref{lem:basicDmonomial}, \ref{lem:diamond2} and \ref{lem:diamondD}. Before getting to the proof we go over some notation and definitions.

Let $\text{Rows}(D,n) := \{r_1, \ldots, r_{n^{1-1/D}}\}$ be the set of
rows in the $D$-th direction and let $V_{\text{Rows}(D,n)}$ be the subspace $G_{\text{Rows}(D,n)}$ projects onto. 
Then $V_{\text{Rows}(D,n)} = V_{\Haar(r_1)} \tensor\ldots \tensor
V_{\Haar(r_{n^{1-1/D}})}$.  A spanning set for $V_{\text{Rows}(D,n)}$
is $H_{\text{Rows}(D,n)} := \{\ket{D_{\sigma_1, \ldots,\sigma_{n^{1-1/D}}} } : \sigma_1, \ldots,\sigma_{n^{1-1/D}} \in S_t\}$. Here $V_{\Haar(S)}$ is the Haar subspace (like $V_\Haar$) on a subset of qudits $S$. $\ket{D_{\sigma_1, \ldots,\sigma_{n^{1/D}}} }$ is the basis state representing maximally entangled states for each qudit such that the qudits in the first row are permuted by $\pi_1$, the qudits in the second row are permuted by $\pi_2$, and so on. In other words:
\be
\ket{D_{\sigma_1,\ldots, \sigma_{n^{1-1/D}}}} =  \bigotimes_{r_i \in \text{Rows}(D,n)}  \bigotimes_ {v \in r_i} \ket{\psi_{\sigma_i}}_{v}.
\ee

We view the $D$ dimensional lattice as $n^{1/D}$ $D-1$-dimensional
sub-lattices, each composed of $n^{1-1/D}$ qudits. More concretely,
the full lattice is the set $A = [n^{1/D}]^D$. For $1 \leq \beta \leq
n^{1/D}$, denote $p_\beta =\{(x(1), \ldots, x_{\text{Rows}(D,n)}) \in
A : x_{\text{Rows}(D,n)} = \beta\}$. We denote the set of these
lattices by $\text{Planes}(D) := \{p_1, \ldots,
p_{n^{1-1/D}}\}$. (This terminology is chosen to match the $D=3$ case
but the arguments here apply to any $D>2$.) These lattices are connected to each other by the rows in $\text{Rows}(D,n)$. $V_{\text{Planes}(D)} = V_{\Haar(p_1)} \tensor\ldots \tensor V_{\Haar(p_{n^{1-1/D}})}$ is the span of $H_{\text{Planes}(D)} := \{\ket{F_{\pi_1, \ldots,\pi_{n^{1-1/D}}} } : \pi_1, \ldots,\pi_{n^{1-1/D}} \in S_t\}$.  Here $\ket{F_{\pi_1, \ldots,\pi_{n^{1-1/D}}} }$ is the basis state of maximally entangled states for each qudit, such that the qudits in $p_1$ are permuted by $\pi_1$, qudits in $p_2$ are permuted by $\pi_2$ and so on. In other words:
 \be
\ket{F_{\pi_1,\ldots, \pi_{{n^{1/D}}}}} =  \bigotimes_{p_i \in \text{Planes}(D)}  \bigotimes_ {v \in p_i} \ket{\psi_{\pi_i}}_{v}.
\ee
Then $G_{\text{Planes}(D)}$ is the projector onto $V_{\text{Planes}(D)}$. 

Let $\tilde{V}_{\text{Planes}(D)} := V_{\text{Planes}(D)}\cap V_\Haar^\perp$ and $\tilde{V}_{\text{Rows}(D,n)} =: V_{\text{Rows}(D,n)}\cap V_\Haar^\perp$ be respectively the orthogonal complements of $V_{\text{Planes}(D)}$ and $V_{\text{Rows}(D,n)}$ with respect to $V_\Haar$. Also define $\tilde H_{\text{Rows}(D,n)}$ and $\tilde H_{\text{Planes}(D)}$ the same as $H_{\text{Rows}(D,n)}$ and $H_{\text{Planes}(D)}$, excluding basis marked with permutations that are all equal to each other. For example, $F_{\pi,\ldots, \pi} \notin \tilde H_{\text{Planes}(D)}$. Define the overlap matrix $[Q]_{gh} := \braket{g|h}$, for $g \in H_{\text{Planes}(D)}$ and $h \in H_{\text{Rows}(D,n)}$. Let $\tilde{J}_{\text{Planes}(D)}$ and $\tilde{J}_{\text{Rows}(D,n)}$ be the Gram matrices corresponding to $\tilde{H}_{\text{Planes}(D)}$ and $\tilde{H}_{\text{Rows}(D,n)}$, respectively. In other words, $[\tilde J_{D}]_{g,h} = \braket{g|h}$ for $g,h \in \tilde H_{\text{Rows}(D,n)}$ and $[\tilde {J_{\text{Planes}(D)}}]_{g,h} = \braket{g|h}$ for $g,h \in \tilde H_{\text{Planes}(D)}$.

We first prove Lemma \ref{lem:generalDproj}, which basically states that the composition of $G_{\text{Rows}(D,n)}$ and $F_{\text{Rows}(D,n)}$ is very close to $G^{(t)}_\Haar$, or equivalently, $\tilde V_{\text{Rows}(D,n)}$ and $\tilde V_{\text{Planes}(D)}$ are almost orthogonal:
\restatelemma{lem:generalDproj}

\begin{proof}
The proof is very similar to the proof of Lemma \ref{lem:2Dproj}. In
particular, we need generalized versions of propositions
\ref{prop:angle}, \ref{gapandQ} and \ref{Qissmall}. The generalization
of proposition \ref{prop:angle} states that
$\cos^2(\measuredangle
(\tilde{V}_{\text{Planes}(D)},\tilde{V}_{\text{Rows}(D,n)}))$ equals
the largest singular value of $F_{D} G_{\text{Rows}(D,n)}  - G_{\Haar}$.
Proposition \ref{gapandQ} generalizes to the statement that the cosine of this angle is equal to
\be
\frac{1}{\sqrt{ \lambda_{\min}(\tilde{J}_{\text{Planes}(D)})\lambda_{\min}( \tilde{J}_{\text{Rows}(D,n)})}} \|Q\|_\infty \leq c_{D,d,n,t}  \|Q\|_\infty.
\ee
Where $1/c_{D,d,n,t}$ is a lower bound on $\sqrt{ \lambda_{\min}(\tilde{J}_{\text{Planes}(D)})\lambda_{\min}( \tilde{J}_{\text{Rows}(D,n)})}$.

We first bound $\|Q\|_\infty$. Using Lemma \ref{sumofrows} 
\be
\|Q\|_\infty \leq \sqrt{\max_h  \sum_g Q_{gh} \max_g  \sum_h Q_{gh}}=: \omega.
\ee

Similar to the calculations in Section \ref{overlap}
\be
Q_{F_{\pi_1, \ldots,\pi_{n^{1-1/D}};D_{\sigma_1, \ldots,\sigma_{n^{1/D}}}}} = \frac{1}{d^{\sum_{i=1}^{n^{1-1/D}} \sum_{j=1}^{n^{1/D}} \dist (\pi_i, \sigma_j)}}.
\ee

Let $\alpha$ ($\beta$) be respectively the set of permutations
$\sigma_1, \ldots, \sigma_{n^{1/D}}$ ($\pi_1,\ldots, \pi_{n^{1-1/D}}$)
that are not all equal. We compute
\bea
\max_{\pi_1, \ldots, \pi_{n^{1-1/D}} \in \beta} \sum_{\sigma_1, \ldots , \sigma_{n^{1/D}}} \frac{1}{d^{\sum_{i=1}^{n^{1-1/D}} \sum_{j=1}^{n^{1/D}} \dist (\pi_i, \sigma_j)}} &=& \max_{\pi_1, \ldots, \pi_{n^{1-1/D}} \in \beta} (\sum_{\sigma} \frac{1}{d^{\sum_{i=1}^{n^{1-1/D}} \dist (\pi_i, \sigma)}})^{n^{1/D}}\\
&=& \max_{\pi_1, \ldots, \pi_{n^{1-1/D}} \in \beta} h(d^{{n^{1/D}}},t, \pi_1, \ldots, \pi_{n^{1-1/D}})\\
&\leq&  \left (\frac{1}{d}+\frac{1}{d^{n^{1-1/D}-1}}+\frac{2 t^2}{d^{n^{1-1/D}}}  \right )^{n^{1/D}}\\
 &=& \frac{1}{d^{\Omega(n^{1-1/D})}}.
\eea
and
\bea
 \max_{\sigma_1, \ldots, \pi_{n^{1/D}} \in \alpha} \sum_{\pi_1, \ldots , \pi_{n^{1-1/D}}} \frac{1}{d^{\sum_{i=1}^{n^{1-1/D}} \sum_{j=1}^{n^{1/D}} \dist (\pi_i, \sigma_j)}} &=&\max_{\sigma_1, \ldots, \sigma_{n^{1/D}} \in \alpha} (\sum_{\pi} \frac{1}{d^{\sum_{j=1}^{n^{1/D}} \dist (\pi, \sigma_j)}})^{n^{1/D}}\\
&=& \max_{\sigma_1, \ldots, \sigma_{n^{1/D}} \in \alpha} h(d^{n^{1-1/D}},t, \sigma_1, \ldots, \sigma_{n^{1/D}})\\
&\leq&  (\frac{1}{d}+\frac{1}{d^{{n^{1/D}}-1}}+\frac{2 t^2}{d^{{n^{1/D}}}}   )^{n^{1-1/D}}\\
 &=& \frac{1}{d^{\Omega(n^{1/D})}}.
\eea
Hence
\be
\omega = \frac{1}{d^{\Omega(n^{1-1/D})}}.
\ee

Next, we have to show that $c_{D,d,n,t}$ is not too large. Using exactly the same steps in the proof of Proposition \ref{Gramlowerbound} we can show that
\be
\lambda_{\min}(\tilde{J}_{\text{Planes}(D)}) \geq 1-\frac{n^{1/D} t(t-1)}{2 d^{n^{1-1/D}}},
\ee
and
\be
\lambda_{\min}(\tilde{J}_{\text{Rows}(D,n)}) \geq 1-\frac{n^{1-1/D} t(t-1)}{2 d^{n^{1/D}}}.
\ee
\end{proof}

Next, we use this result to prove Lemma \ref{lem:basicDmonomial}. Recall the expression $\tilde{G}_{n,D,c} $ from Definition \ref{def:recursive}
\be
\tilde{G}_{n,D,c} =  (\tilde{G}^{\tensor n^{1/D}}_{n^{1-1/D}, D-1 , c} G_{\text{Rows}(D,n)} \tilde{G}^{\tensor n^{1/D}}_{n^{1-1/D}, D-1 , c} )^c,
\ee
where $c$ is a constant depending on $D$ and $t$, but independent of $n$. Note that $\tilde{G}_{n,D,c} = \tilde{G}^\dagger_{n,D,c}$ if $\tilde{G}_{n^{1-1/D},D-1,c} = \tilde{G}^\dagger_{n^{1-1/D},D-1,c}$. Also let ${\hat {G}}_{n,D,c}:= G_{\text{Rows}(D,n)} \tilde{G}_{n^{1-1/D},D-1,c}^{\tensor n^{1/D}} G_{\text{Rows}(D,n)}$. Hence 
$\tilde{G}_{n,D,c} = \tilde{G}_{n^{1-1/D},D-1,c}^{\tensor n^{1/D}} ({\hat {G}}_{n,D,c})^{c-1}  \tilde{G}_{n^{1-1/D},D-1,c}^{\tensor n^{1/D}}$.

\restatelemma{lem:basicDmonomial}


\begin{proof}
The proof is by induction. The base case $D = 2$ is by Lemma \ref{lem:2Dproj}. We assume that for any large enough $m$, $\|{\hat {G}}_{m,D-1,c} - G_\Haar \|_\infty \leq \frac{1}{d^{O(m^{1/{D-1}})}}$, and we show that $\|{\hat {G}}_{n,D,c} - G_\Haar \|_\infty \leq \frac{1}{d^{\Omega(n^{1/D})}}$.
\bea
\left \|{\hat {G}}_{n,D,c} - G_\Haar \right\|_\infty &\leq& \left\|G_{\text{Rows}(D,n)} \tilde{G}_{n^{1-1/D},D-1,c}^{\tensor n^{1/D}} G_{\text{Rows}(D,n)} - G_\Haar\right\|_\infty\\
&\leq& \left\| \tilde{G}_{n^{1-1/D},D-1,c}^{\tensor n^{1/D}} G_{\text{Rows}(D,n)} - G_\Haar\right\|_\infty\\ 
&=& \Big\|(\tilde{G}_{n^{1-1/D},D-1,c}^{\tensor n^{1/D}} -G_{\text{Planes}(D)}) G_{\text{Rows}(D,n)}\\
 &&+ G_{\text{Planes}(D)} G_{\text{Rows}(D,n)}- G_\Haar\Big\|_\infty\\ 
&\leq& \Big\|(\tilde{G}_{n^{1-1/D},D-1,c}^{\tensor n^{1/D}} -G_{\text{Planes}(D)}) G_{\text{Rows}(D,n)}\Big\|_\infty\\
 &&+ \Big\|G_{\text{Planes}(D)} G_{\text{Rows}(D,n)}- G_\Haar\Big\|_\infty\\ 
&\leq& \Big\|\tilde{G}_{n^{1-1/D},D-1,c}^{\tensor n^{1/D}}  -G_{\text{Planes}(D)}\Big\|_\infty+ \Big\|G_{\text{Planes}(D)} G_{\text{Rows}(D,n)}- G_\Haar\Big\|_\infty\\ 
&\leq& n^{1/D} \Big\|\tilde{G}_{n^{1-1/D},D-1,c} - G_{\Haar(p_1)} \Big\|_\infty + \Big\|G_{\text{Planes}(D)} G_{\text{Rows}(D,n)}- G_\Haar\Big\|_\infty\\
&\leq& n^{1/D}  \frac{1}{d^{O( n^{(1-1/D)\cdot 1/{(D-1)}})}} + 1/d^{\Omega(n^{1-1/D})}\\
&\leq& \frac{n^{1/D}}{d^{\Omega(n^{1/D})}} + 1/d^{\Omega(n^{1-1/D})}\\
&\leq& \frac{1}{d^{\Omega(n^{1/D})}}.
\eea
\end{proof}

\restatelemma{lem:basicDmonomial}

\begin{proof}
The proof is by induction. Our induction hypothesis is $\max_x |\braket{x | (\tilde{G}_{n ,D,c}  -G_\Haar)^2 | x}| \leq \frac{\eps}{d^{nt}}$. First, we show this bound (for sub-lattices of dimension $D-1$) implies the statement of this theorem:
\bea
|\braket{x | \tilde{G}_{n ,D,c}  -G_\Haar | y}| &=& |\braket{x | \tilde{G}_{n^{1-1/D},D-1,c}^{\tensor n^{1/D}} ({\tilde{G'}}_{n,D,c}^{c-1} - G_\Haar)  \tilde{G}_{n^{1-1/D},D-1,c}^{\tensor n^{1/D}} | y}|\nonumber \\
&\leq& \Big\|{{\hat {G}}}_{n,D,c} - G_\Haar \Big\|^{c-1}_\infty \Big\|   \tilde{G}_{n^{1-1/D},D-1,c}^{\tensor n^{1/D}} \ket{y}\bra{x} \tilde{G}_{n^{1-1/D},D-1,c}^{\tensor n^{1/D}}\Big\|_1\nonumber \\
&\leq&  \Big\|{{\hat {G}}}_{n,D,c} - G_\Haar \Big\|^{c-1}_\infty \nonumber \\
&&\times\max_{x,y \in [d]^{2 t n^{1-1/D}}}\Big\|   \tilde{G}_{n^{1-1/D},D-1,c} \ket{y}\bra{x} \tilde{G}_{n^{1-1/D},D-1,c}\Big\|^{ n^{1/D}}_1\nonumber \\
&\leq& \frac{1}{d^{O(c \cdot n^{1/D})}} \max_{x,y \in [d]^{2 t n^{1-1/D}}}\Big\| \tilde{G}_{n^{1-1/D},D-1,c}\ket{y}\bra{x}  \tilde{G}_{n^{1-1/D},D-1,c}\Big\|^{ n^{1/D}}_1\nonumber \\
&\leq& \frac{1}{d^{O(c \cdot n^{1/D})}} \max_{x \in [d]^{2 t n^{1-1/D}}} |\bra{x}  \tilde{G}^2_{n^{1-1/D},D-1,c}\ket x|^{ n^{1/D}}\nonumber \\
&\leq& \frac{1}{d^{O(c \cdot n^{1/D})}} \nonumber \\
&&\times\max_{x \in  [d]^{2 t n^{1-1/D}}} (\braket{x|G_\Haar|x}  + |\bra{x} ( \tilde{G}_{n^{1-1/D},D-1,c}-G_H)^2\ket x|)^{ n^{1/D}}\nonumber \\
&\leq& \frac{1}{d^{O(c \cdot n^{1/D})}}   \left (\max_{x \in  [d]^{2 t n^{1-1/D}}}\frac{t! + 1/d^{n^{(1-1/D) \cdot \frac{1}{D-1}}}}{d^{n^{1-1/D}t}} \right)^{ n^{1/D}}\nonumber \\
&\leq& \frac{\eps}{d^{nt}}.
\label{eq:nd}
\eea

Next, assumming $\max_x |\braket{x | (\tilde{G}_{n^{1-1/D},D-1,c} -G_\Haar)^2 | x}| \leq \frac{\eps}{d^{n^{1-1/D} t}}$, we show $\max_x |\braket{x | (\tilde{G}_{n,D,c} -G_\Haar)^2 | x}| \leq \frac{\eps}{d^{nt}}$. The proof is very similar to the above calculation:
\bea
|\braket{x | (\tilde{G}_{n ,D,c} -G_\Haar)^2 | y}| &=& \bra x | \tilde{G}_{n^{1-1/D},D-1,c}^{\tensor n^{1/D}} ({{\hat {G}}}_{n,D,c}^{c-1} - G_\Haar) \tilde{G}_{n^{1-1/D},D-1,c}^{\tensor n^{1/D}} \nonumber \\
&&\times ({{\hat {G}}}_{n,D,c}^{c-1} - G_\Haar)  \tilde{G}_{n^{1-1/D},D-1,c}^{\tensor n^{1/D}} | \ket y\nonumber \\
&\leq& \Big\|({{\hat {G}}}_{n,D,c}^{c-1} - G_\Haar)  \tilde{G}_{n^{1-1/D},D-1,c}^{\tensor n^{1/D}} ({{\hat {G}}}_{n,D,c}^{c-1} - G_\Haar) \Big\|_\infty\nonumber \\
&&\times \Big\|\tilde{G}_{n^{1-1/D},D-1,c}^{\tensor n^{1/D}} \ket{y}\bra{x} \tilde{G}_{n^{1-1/D},D-1,c}^{\tensor n^{1/D}}\Big\|_1\nonumber \\
&\leq& \Big\|{{\hat {G}}}_{n,D,c}^{c-1} - G_\Haar \Big\|_\infty \nonumber \\
&&\times\max_{x,y \in [d]^{2 t n^{1-1/D}}}\Big\|   \tilde{G}_{n^{1-1/D},D-1,c} \ket{y}\bra{x} \tilde{G}_{n^{1-1/D},D-1,c}\Big\|^{ n^{1/D}}_1\nonumber \\
&\leq& \frac{\eps}{d^{nt}}.
\eea

In the third line we have used Lemma \ref{lem:infinitypalace}. We skip the calculations after the third line because it is similar to the calculations of \eqref{eq:nd}.

\end{proof}

Next, we prove Lemma \ref{lem:diamondD}. Lemma \ref{lem:diamond2} is a special case of this lemma and we skip its proof. 

\restatelemma{lem:diamondD}

\begin{proof}
As discussed in Section \ref{sec:elementarytools}, the superoperator $\Channel[G^{(t)}_\Haar]$ can be written in the following canonical form
\be
\Channel[G^{(t)}_\Haar] [X] = \sum_{\pi \in S_t} \Tr (V(\pi)X) \Wg(\pi).
\ee
Using the notation defined in Section \ref{sec:elementarytools}, $\cX [\Channel[G^{(t)}_\Haar]] = G^{(t)}_\Haar$ and
\be
\left(G_{\text{Rows}(D,n)} G_{\text{Planes}(D)} G_{\text{Rows}(D,n)}\right)^c - G^{(t)}_\Haar =: \sum_{a, b \in S_t^{\times n^{1-1/D}}} \ket {D_a} \Lambda_{a,b} \bra {D_b}.
\ee
Using the definition of $\Lambda$ we can write 
\bea
\Channel\left[(G_{\text{Rows}(D,n)}G_{\text{Planes}(D)}G_{\text{Rows}(D,n)})^c - G^{(t)}_\Haar\right] &=&\Channel\left[\sum_{a, b \in S_t^{\times n^{1-1/D}}} \ket {D_a} \Lambda_{a,b} \bra {D_b}\right]\nonumber \\
&=& \sum_{a, b\in S_t^{\times n^{1-1/D}}}  \frac{1}{d^{nt}}V(a) \Lambda_{a,b}  V^\ast(b).
\eea
Therefore 
\bea
\|\Channel[(G_{\text{Rows}(D,n)}G_{\text{Planes}(D)}G_{\text{Rows}(D,n)})^c - G^{(t)}_\Haar]\|_\diamond &\leq& \sum_{a, b\in S_t^{\times n^{1-1/D}}} |\Lambda_{a,b}|   \|\frac{1}{d^{nt}}V(a) V^\ast(b)\|_\diamond\nonumber \\
 &\leq&  \sum_{a, b\in S_t^{\times n^{1-1/D}}} |\Lambda_{a,b}| \leq t^{O(n^{1-1/D})} \|\Lambda\|^c_\infty.
\label{eq:272}
\eea
Here we have used $\|\frac{1}{d^{nt}}V(a) V^\ast(b)\|_\diamond \leq 1$. This is because $V(a) V^\ast(b)$ is a tensor product of $n^{1-1/D}$ superoperators, i.e., $ \otimes_i V(a_i) V^\ast(b_i)$, and hence $\|V(a) V^\ast(b)\|_\diamond = \prod_i \|V(a_i) V^\ast(b_i)\|_\diamond$. It is enough to show that each of $\|V(a_i) V^\ast(b_i)\|_\diamond$ is bounded by $1$.
\bea
\frac{1}{d^{nt}}\left\| V(a_1)V(b_1)^\ast\right\|_\diamond &=& \frac{1}{d^{nt}} \sup_{X: \left\|X\right\|_1 = 1} \left\|\Tr_A (V(a_1)_A \tensor \id_B X_{AB}) \tensor V_A(b_1)\right\|_1\nonumber \\
&=&  \sup_{X : \|X\|_1 = 1} \Big\|\Tr_A (V(a_1)_A \tensor \id_B X_{AB})\Big\|_1\cdot \frac{1}{d^{nt}} \Big\|V_A(b_1)\Big\|_1\nonumber \\
&\leq&  \sup_{X: \|X\|_1 = 1} \Big\|V_{(a_1)} \tensor \id X_{AB}\Big\|_1 \cdot 1\nonumber \\
&=&  \sup_{X: \|X\|_1 = 1} \|X_{AB}\|_1\nonumber \\
&\leq& 1.
\eea

It is enough to compute $\|\Lambda\|_\infty$. Let $\ket a$ be an orthonormal basis labeled according to the indices of $\Lambda$. Define
\be
T := \sum_{a,b} \sqrt \Lambda_{a b} \ket {D_a} \bra b.
\ee
$T T^\dagger = \sum_{a, b} \ket {D_a} \Lambda_{a,b} \bra {D_b}$ and $T^\dagger T = \sum_{a, b} \ket a (\sqrt \Lambda J \sqrt \Lambda)_{ab} \bra b$, where $[J]_{a,b} = \braket{D_a | D_b}$. First of all, using Lemma \ref{lem:spectrum}  $T T^\dagger$ and $T^\dagger T$ have the same spectra. Hence
\be
\|\sum_{a, b} \ket {D_a} \Lambda_{a,b} \bra {D_b}\|_\infty = \|T T^\dagger\|_\infty = \|T^\dagger T\|_\infty = \|\sqrt \Lambda J \sqrt \Lambda\|_\infty
\ee
Therefore
\bea
\|\Lambda\|_\infty &\leq& \Big\|\sum_{a, b} \ket {D_a} \Lambda_{a,b} \bra {D_b}\Big\|_\infty +  \|\sqrt \Lambda (J-\id) \sqrt \Lambda\|_\infty\nonumber \\
 &\leq&  \Big\|\left(G_{\text{Rows}(D,n)} G_{\text{Planes}(D)} G_{\text{Rows}(D,n)}\right)^c - G^{(t)}_\Haar\Big\|_\infty + \|\sqrt \Lambda\|_\infty  \|J-\id\|_\infty \|\sqrt \Lambda\|_\infty\nonumber \\
 &=&  \Big\|\left(G_{\text{Rows}(D,n)} G_{\text{Planes}(D)} G_{\text{Rows}(D,n)}\right)^c - G^{(t)}_\Haar\Big\|_\infty + \| \Lambda\|_\infty  \|J-\id\|_\infty
\eea
As a result 
\be
\|\Lambda\|_\infty \leq \frac{ \left\|(G_{\text{Rows}(D,n)} G_{\text{Planes}(D)} G_{\text{Rows}(D,n)})^c - G^{(t)}_\Haar\right\|_\infty}{1 -\|J-\id\|_\infty}.
\ee
In Lemma \ref{lem:generalDproj} we showed that $\|(G_{\text{Rows}(D,n)} G_{\text{Columns}(D)} G_{\text{Rows}(D,n)})^c - G^{(t)}_\Haar\|_\infty \leq \|G_{\text{Columns}(D)} G_{\text{Rows}(D,n)} - G^{(t)}_\Haar\|^c_\infty = 1/d^{O(c n^{1-1/D})}$. It is enough to show that $\| J-\id\|_\infty$ is small. But $J$ is tensor product of $n^{1-1/D}$ Gram matrices $J_1$ such that $\|J_1 - \id\|_\infty = \frac{O(t^2)}{d^{n^{1/D}}}$ (see Lemma \ref{SpectrumJ}), hence $\| J-\id\|_\infty = n^{1-1/D} \frac{O(t^2)}{d^{n^{1/D}}}$ which is bounded by $1/2$ for large enough $n$ and constant $t$ and $D$. As a result, $\|\Lambda\|_\infty = 1/d^{O(c n^{1-1/D})}$. Combining this with \eqref{eq:272} we find that
\be
\left \|\Channel\left [(G_{\text{Rows}(D,n)} G_{\text{Planes}(D)}G_{\text{Rows}(D,n)})^c - G^{(t)}_\Haar\right]\right\|_\diamond \leq t^{O(t n^{1-1/D})}1/d^{O(c n^{1-1/D})}.
\ee

\end{proof}

\newpage
\section{$O(n \ln^2 n)$-size random circuits with long-range gates output anti-concentrated distributions}
\label{sec:all-to-all-anti-conc}

Recall that for a circuit $C$, $\Coll(C)$ is the collision probability, 
\be\sum_{x \in \{0,1\}^n} |\braket{x |C|0}|^4,\label{eq:collision}\ee
 of $C$ in the computational basis. Also recall that $\mu^{(\text{CG})}_{t}$ is the distribution over random circuits obtained from application of $t$ random long-range gates. Unlike the previous section where we used $t$ to denote the degree of a monomial, here we use $t$ for time, i.e. the number of time-steps in a random circuit.

The goal of this section is to prove the following theorem:
\restatetheorem{thm:cg} 
\label{objective}
Our strategy is to relate the convergence of the expected collision probability to a classical Markov chain mixing problem. In Section \ref{section:cgbackground} we go over the notation and definitions we use in the proof of this theorem. In Section \ref{section:cgproof} we prove the theorem. This proof is based on several lemmas which we will prove in sections \ref{section:relatingtoMarkov} and \ref{chain}. 

\subsection{Background: random circuits with long-range gates and Markov chains}
\label{section:cgbackground}

Previous work \cite{ODP06, HL09, BF13, BF13-2} demonstrates that if we
only care about the second moment of $\mu^{(\text{CG})}_t$, then the
corresponding moment superoperator is related to a certain classical
Markov chain. In particular the application of the moment
superoperator on the basis $\text P^2_n := \left \{\sigma_p \tensor
  \sigma_p : p \in \{0,1,2,3\}^n\right \}$ is a classical Markov
chain.
 We now describe this connection.

We first start with some basic properties of moment superoperators.

\begin{fact} \OldNormalFont{}
Let $\mu$ and $\mu_1, \ldots, \mu_K$ be distributions over circuits.

\begin{enumerate}
\item If $\mu$ is a convex combination of $\mu_1, \ldots, \mu_K$ then $\Channel \left[G^{(2)}_\mu\right]$ is the same convex combination of
\\
 $\Channel \left[G^{(2)}_{\mu_1}\right], \ldots ,\Channel \left[G^{(2)}_{\mu_K}\right]$. 

\item  If $\mu$ is the composition of a circuit from $\mu_1$ with a circuit with $\mu_2$, then $\Channel \left[G^{(2)}_\mu\right] = \Channel \left[G^{(2)}_{\mu_2}\right] \circ \Channel \left[G^{(2)}_{\mu_1}\right]$.
\end{enumerate}
\end{fact}

Recall that $\Channel \left[G^{(2)}_{i,j}\right]$ denotes $\Channel
\left[G^{(2)}_{\text U(4)}\right]$ applied to qubits $i$ and
$j$. Since $\mu^{\text{CG}}_1$ is a convex combination of two-qubit
random $\text U(4)$ gates, the first point above implies that
\be
\Channel \left[G^{(2)}_{\mu^{\text{CG}}_1}\right] = \frac 2 {n(n-1)} \sum_{i < j}\Channel \left[G^{(2)}_{i,j}\right]
\ee
and since $\mu^{\text{CG}}_t$ is $t$ times compositions of $\mu^{\text{CG}}_1$ with itself, the second item implies that
\be
\Channel \left[G^{(2)}_{\mu^{\text{CG}}_t}\right] = \left(\frac 2 {n(n-1)} \sum_{i < j}\Channel \left[G^{(2)}_{i,j}\right]\right)^t.
\label{eq:186}
\ee

The moment superoperator $\Channel[G^{(2)}_{\text U(4)}]$ has the following simple action on the Pauli basis:
\be
\Channel[G^{(2)}_{\text U(4)}] (\sigma_p\tensor \sigma_q \tensor \sigma_a\tensor \sigma_b) = 
\begin{cases}
\sigma_0\tensor \sigma_0 \tensor \sigma_0\tensor \sigma_0 & p q =ab = 00 \\
\frac{1}{15} \sum_{\substack{c, d \in \{0,1,2,3\}^2 \\
cd \neq 00}} \sigma_c \tensor \sigma_d  \tensor \sigma_c \tensor \sigma_d & p q =ab \neq 00\\
0 & \text{otherwise}
\end{cases}
\label{eq:Pauli-action}
\ee

In particular the action of $\Channel[G^{(2)}_{\text U(4)}]$ on the
Pauli basis $\text P^2_2$ is a stochastic matrix, and for any pair $i\neq j$ the action of  $\Channel[G^{(2)}_{\text U(4)}]$ on qubits $i,j$ can be represented by a stochastic matrix acting on $\text P^2_n$.
Using
\eqref{eq:186} $\Channel \left[G^{(2)}_{\mu^{\text{CG}}_s}\right]$ on
$\text P^2_n$ is also a stochastic matrix.  We can describe this stochastic matrix as a Markov chain on state space $\cS = \{0,1,2,3\}^n$, with $S_t\in \cS$ describing the string at time $t$.

It turns out that the expected collision probability depends on the subset of qubits that have been hit by the random circuit. In case a subset of size $m$ of qubits (out of $n$ qubits) never have a gate applied to them, then the expected collision probability converges to a value like $\approx \frac{2^m}{2^n}$ and not $\frac 1 {2^n +1}$. So we need to separately track which qubits have ever been hit by a gate throughout this process.  Let $H_t\in 2^{[n]}$ denote the set of qubits that have been hit by at least one gate by time $t$, where $2^{[n]}$ denotes the power set of $[n]$.

%
%

Together $(S_t,H_t)$ can be modeled as the following Markov chain.

\begin{definition} \OldNormalFont{}
Let $(S_0,H_0), (S_1,H_1), (S_2,H_2), \ldots \in \cS \times 2^{[n]}$ be the following classical Markov chain.  
Initially $H_0 = \emptyset$ and $S_0$ is a random element of $\{0,3\}^n \backslash 0^n$. 
At each time step $t$ we choose a random pair $i,j \in [n]$ with $i\neq j$.  We let $H_{t+1} = H_t \cup \{i,j\}$ so that $H_t$ represents the set of all indices chosen up to time $t$.  We determine $S_{t+1}$ from $S_t$ using the transition rule of \eqref{eq:Pauli-action}.  Specifically if the $i,j$ positions of $S_t$ are both 0, then we leave them equal to 00, and otherwise we replace them with a uniformly random element of $\{01,02,\ldots,33\}$.
\label{def:classicalprocess}
\end{definition}

Suppose we condition on $H_t \supseteq H$ for some set $H$ with $|H|=n-m$.
Let 
\be P_t^{(n-m)}(k):= \Pr\left[| S_t(H)|=k | H_t \subseteq H\right]. \label{eq:Ptn-def}\ee
We can use this notation since the RHS of
\eqref{eq:Ptn-def} depends only on $|H|,t,n,k$ and not on $H$.

For a function $f : [n] \rightarrow \bbR$ we define $\|\cdot \|_\ast$ to be the following norm
\be
\|f\|_\ast := \sum_{k \in [n]} \frac {|f(k)|}{3^k}.
\label{eq:starnorm}
\ee
\subsubsection{Summary of the definitions}

 See below for a summary of the definitions:
\begin{center}
\begin{tabular} 
{  l l l}
\toprule
\textbf{Notation} & \textbf{Definition} & \textbf{Reference}\\
\midrule
$\Coll(C)$ & the collision probability of circuit $C$ &  Equation \eqref{eq:collision}\\
\hline
$G^{(t)}_\mu$ & average of $C^{\ot t,t}$ over $C\sim \mu$ &Definition \ref{def:moments}\\
\hline
$G^{(t)}_{i,j}$ & Haar projector of order $t$ on qudits $i$ and $j$&Definition \ref{def:moments}\\
\hline
$\mu^{\text{CG}}_{t}$ & the distribution over circuits with $t$ random
                        two-qubit gates &Definition \ref{def:cg}\\
\hline
$P^2_n$ & $\{\sigma_p \ot \sigma_p : p \in \{0,1,2,3\}^n\}$ &Section \ref{section:cgbackground}\\
\hline
$S_0, S_1, \ldots$ & Markov chain of Pauli strings & Definition \ref{def:classicalprocess}\\ 
\hline
$H_t$ & subset of $[n]$ that is covered according to the Markov chain of Pauli strings & Definition \ref{def:classicalprocess}\\
\hline
$S'_0, S'_1, \ldots$ & Accelerated Markov chain of binary strings with decoupled coordinates & Definition \ref{def:decoupled}\\ 
\hline
$X_t$ & $|S_t|$ & Section \ref{chain}\\
\hline
$Y_t$ & Steps of the accelerated Markov chain $Q$ & Section \ref{chain}\\
\hline
$P_t^{(n-m)}(k)$ & $\Pr\left[| S_t(H)|=k | H_t \subseteq H\right]$ for
                   any fixed $H$ with $|H|=n-m$&Equation \eqref{eq:Ptn-def}\\
\hline
$P_t(k)$ & $ \Pr\left[| S_t(H)|=k\right]$.  Also equal to $P_t^{(n)}(k)$ & Equation \eqref{eq:Ptn-def}\\
\hline
$\|f\|_\ast$ & $\sum_{x=1}^n \frac {|f(x)|}{3^x}$&Equation \eqref{eq:starnorm}\\
\hline
$\|f\|_\text{\Hsquare}$ & $\sum_{x=1}^n |f(x)|\frac {3n} {x 3^x}$ & Proposition \ref{prop:combining}\\
\hline
$P$ & transition matrix of the birth and death Markov chain & Equation \eqref{weightchain}\\
\hline 
$Q$ & transition matrix of the partially accelerated Markov chain & Equation \eqref{accelerated}\\
\hline
$T_{\text{left (right)}} (Y^s)$ & wait time for the steps $Y_0,\ldots, Y_s$ on the left (right) hand side of site $\frac 56 n$&Section \ref{sec:waitabc}\\
\hline
$\nu$ &$3/4 n$&Section \ref{sec:waitabc} \\
\hline
$\nu_\tau$ &$Y_0 \exp(- \frac{\tau}{\nu}) + \nu (1- \exp(- \frac{\tau}{\nu}))$&Section \ref{sec:waitabc} \\
\hline
$\beta$ &  $8(4+c) \ln n$ for constant $c$ fixed in advance &Section \ref{sec:waitabc}\\
\hline
$x(0)$ &  $\nu/\beta$ &Section \ref{sec:waitabc}\\
\hline
$\rho_x$ &  $\sum_{j = 1}^s I\{Y_j = x\}$&Section \ref{sec:waitabc}\\
\hline
$A$ &  $\cap_{1 \leq x\leq x(0)} \{N_x \leq \beta x\}$ &Section \ref{sec:waitabc}\\
\hline
$M_s$ &  $\min_{1 \leq j \leq s} \{Y_j\}$ &Section \ref{sec:waitabc}\\
\hline
$y^s$ &  short hand for $(y_0, \ldots, y_s)$ &Section \ref{sec:propcombining}\\
\hline
$\mathsf{Bin} (n,p)$ &  binomial distribution on $n$ elements each occurring with probability $p$  &\\
\hline
$\mathsf{Geo}(\alpha)$ & Geometric distribution with mean $\frac{1}{\alpha}$&\\
\hline
$\mathsf{Pois} (\tau)$ &  Poisson distribution with mean $\tau$  &\\
\hline
$\mathsf{Unif} [a,b]$ &  continuous uniform distribution on the interval $[a,b]$  &\\
\bottomrule\end{tabular}
\end{center}

\subsection{Proof of  \thmref{cg}: bound on the collision probability}
\label{section:cgproof}

Before giving the proof we state the following three main theorems. The first one relates the expected collision probability to the $\|\cdot\|_\ast$ norm of the probability vector on the state space of the Markov chain of weights. More concretely
\begin{theorem} \torestate{\
\be
\E_{C \sim \mu^{\text{CG}}_t} \text{Coll}(C)  \leq \frac 1 {2^n}  + \sum_{m=0}^n {n \choose m} e^{ -t m/n}  \|P_t^{(n-m)}\|_\ast
\ee
\label{thm:collisiontochain}}
\end{theorem}\OldNormalFont{}
\noindent This result is proved in Section \ref{section:relatingtoMarkov}. 

The second theorem shows that for $t \approx n \ln^2 n$, $\| P^{(n)}_t\|_\ast \approx \text{Constant} \times \frac{1}{2^n + 1}$, where $\frac{1}{2^n+1}$ is the value of this norm at the stationary state.
\begin{theorem}\torestate{
There exists a constant $c$ such that
if $t = c n \ln^2 n$ then $\|P^{(m)}_t\|_\ast \leq \frac{28}{2^m + 1}$. \label{thm:starmixing}}
\label{thm:pstartmixing}
\end{theorem}
\noindent This result is proved in Section \ref{chain}. 

The third theorem gives an exact expression for the collision probability in terms of the Markov chain $S_0, S_1, \ldots$. We use this to compute the lower bound.
\begin{theorem}\torestate{
$\Coll_{\mu^{\text {CG}}_t} = \frac{1}{2^n} \left(1 + \sum_{p,q \in \{0,3\}^n\backslash 0^n} \Pr [S_t = p | S_0 = q]\right)$\label{thm:lowerbound}}
\end{theorem}
The proof of this expression is the same as equation \eqref{eq:collexpression} and is derived in section \ref{section:relatingtoMarkov}.

\begin{proof}[\OldNormalFont{} Proof of \thmref{cg}]

We first prove the upper bound. There are two major steps. 

Combining Theorems \ref{thm:collisiontochain} and
\ref{thm:pstartmixing} and choosing $t=cn\ln^2(n)$ we obtain
\bea
\E_{C \sim \mu^{\text{CG}}_{t}} \text{Coll}(C) 
 &\leq& \frac 1 {2^n}  + \sum_{m=0}^n {n \choose m} e^{ -t m/n}  \frac{28}{2^{n-m}}\\
 &\leq& \frac 1 {2^n}  (1 +28 (1 + 2 e^{-t/n})^n)\\
 &\leq& \frac 1 {2^n}  (1 +28 (1 + \frac{2}{n^{c \ln n}})^n)\\
 &\leq& \frac{29}{2^n+1}.
\eea
Here we need to assume $n$ is larger than some universal
constant. This can be done by adjusting $c$ to cover the finite set of
cases where $n$ is too small.

For the lower bound we use the expression in Theorem \ref{thm:lowerbound} and bound it according to
\bea
\Coll_{\mu^{\text {CG}}_t} &\geq& \frac 1 {2^n}\sum_{p \in \{0,3\}^n} \Pr [S_t = p | S_0 = p],\\
&=& \frac 1 {2^n}\sum_{k=0}^n \sum_{p \in \{0,3\}^n : |p|_H = k} \Pr [S_t = p | S_0 = p],\\
&\geq& \frac 1 {2^n}\sum_{k=0}^n {n \choose k} r_k^t,
\eea
where
\be
r_k = \frac {14} {15} (1 - \frac k n) (1-\frac k {n-1}) + \frac 1 {15} \geq e^{- 3 \frac k n},
\ee
is the probability that a string of Hamming weight $k$ does not change after one step of the Markov chain. Assume $t \leq \frac 1 {3 c'}  n \ln n$ then 
\bea
\Coll_{\mu^{\text {CG}}_t} &\geq& \frac 1 {2^n}\sum_{k=0}^n {n \choose k} e^{- 3 \frac {kt} n},\\
&=& \frac 1 {2^n}(1 + e^{- 3 t/n})^n,\\
&\geq & \frac 1 {2^n}\exp(\frac{n^{1-1/c'}}{1 + n^{-1/c'}})\\
&\geq & \frac 1 {2^n} \cdot 1.6 ^{n^{1-1/c'}}
\eea
\end{proof}

\subsection{Proof of Theorem \ref{thm:collisiontochain}: relating collision probability to a Markov chain}
\label{section:relatingtoMarkov}

In this section we relate the expected collision probability of a random circuit with long-range gates to the $\|\cdot\|_\ast$ norm of the probability vector $P^{(m)}_t$ defined in Section \ref{section:cgbackground}.   We will prove several intermediate results along the way to \thmref{collisiontochain}.

\begin{theorem}[Section 3 of \cite{HL09}]\label{thm:HL09} \OldNormalFont{}
We can write
\be
\Channel \left[G^{(2)}_{\mu^{\text{CG}}_t}\right] (\sigma_q \ot \sigma_q) = \sum_{p \in \{0,1,2,3\}^n} \Pr[S_t = p|S_0 = q] \sigma_p \ot \sigma_p.
\ee
\end{theorem}

\begin{proof}[\OldNormalFont{} Proof of \thmref{lowerbound}]
We can write the expected collision probability in terms of the moment superoperator $\Channel \left[G^{(2)}_{\mu^{\text{CG}}_t}\right]$. We use the notation $\Coll_{\mu^{\text{CG}}_t} = \E_{C \sim \mu^{\text{CG}}_t} \Coll(C)$:
\bea
\Coll_{\mu^{\text{CG}}_t} &=& \sum_{z \in \{0,1\}^n} \E_{C \sim{\mu^{\text{CG}}_t}} |  \langle z  | C  | 0 \rangle |^4\nonumber\\
 &=& \sum_{z \in \{0,1\}^n} \bra z \tensor \bra z \E_{C \sim {\mu^{\text{CG}}_t}} \left(C \ket{0^n}\bra{0^n} C^\dagger \tensor C \ket{0^n}\bra{0^n} C^\dagger  \right) \ket z \tensor \ket z\nonumber\\
  &=& \Tr \sum_{z \in \{0,1\}^n} \proj  z \tensor \proj z \Channel \left[G^{(2)}_{\mu^{\text{CG}}_t}\right] \left( \ket{0^n}\bra{0^n} \tensor  \ket{0^n}\bra{0^n}\right)
  \label{eq:expressionforcoll}
\eea

It is useful to write $\ket{0^n}\bra{0^n} \tensor  \ket{0^n}\bra{0^n}$
and $\sum_{z \in \{0,1\}^n} \proj  z \tensor \proj z $ in the Pauli basis:
\ba
\ket{0^n}\bra{0^n} \tensor  \ket{0^n}\bra{0^n} &
= \frac 1 {4^n} \sum_{p,q \in \{0,3\}^n} \sigma_p \tensor \sigma_q.\\
\sum_{z \in \{0,1\}^n} \proj z \ot \proj z &
= \frac 1 {2^n} \sum_{p \in \{0,3\}^n} \sigma_p \tensor \sigma_p.
\label{eq:sumoverz}
\ea
Then the collision probability becomes:
\be
\Coll_\nu = \frac{1}{2^n} +(1- \frac{1}{2^n})\frac 1{2^n}\Tr \left(
  \sum_{z \in \{0,1\}^n} \proj z \tensor \proj z\right)
 \Channel\left[G^{(2)}_{\mu^{\text{CG}}_t}\right] \left( \frac 1 {2^n-1} \sum_{q \in \{0,3\}^n\backslash 0^n} \sigma_q \tensor \sigma_q \right) 
\ee
Using \thmref{HL09}
\be
\Channel\left[G^{(2)}_{\mu^{\text{CG}}_t}\right] \left( \frac 1 {2^n-1} \sum_{q \in \{0,3\}^n\backslash 0^n} \sigma_q \tensor \sigma_q\right) 
=  \frac 1 {2^n-1} \sum_{\substack{p \in \{0,1,2,3\}^n\backslash 0^n \\ q \in \{0,3\}^n\backslash 0^n}}
\Pr[S_t = p | S_0 = q]\sigma_p \tensor \sigma_p.
\ee

As a result 
\bea
\Coll_{\mu^{\text {CG}}_t} &=& \frac{1}{2^n}\left(1 +\sum_{p,q \in \{0,3\}^n\backslash 0^n} \Pr [S_t = p | S_0 = q]\right)
\label{eq:collexpression}
\eea
\end{proof}

For a string $a \in \{0,1,2,3\}^n$ and a subset $A \in 2^{[n]}$ we let $a(A)$ denote the substring of $a$ restricted to $A$.
\begin{lemma}\label{lem:uniform-123} \OldNormalFont{}
For $H \subseteq [n]$ and  $p,q \in \{0,1,2,3\}^n$ 
\be
 \Pr [S_t = p | S_0 = q, H_t=H]
=  \frac 1 {{|H| \choose |p(H)|} 3^{|p(H)|}}\Pr [|S_t(H)|=|p(H)| \big| S_0 = q, H_t=H]
\ee
if $q([n] \backslash H) = p([n] \backslash H)$ and $0$ otherwise. 
\end{lemma}
In other words, once we condition on $H_t=H$, the probability distribution of $S_t(H)$ depends only on its Hamming weight.

\begin{proof}
Conditioned on $H_t=H$ the sites of $[n]\backslash H$ are not hit, so the event that $q([n] \backslash H) \neq p([n] \backslash H)$ has zero probability. Now since the set $H$ is covered, $1,2$ or $3$ have equal probabilities of appearing at any position of the string $S_t (H)$.  As a result for each non-zero bit of $S_t (H)$ we get a factor of $1/3$.
\end{proof}

Using \thmref{lowerbound} and \lemref{uniform-123} we obtain

\begin{corollary}\OldNormalFont{}
\be
\Coll_{\mu^{\text{CG}}_t} =
\frac 1 {2^{n}} + (1-1/2^n) \sum_{H \subseteq [n]}  \Pr [H_t = H] \sum_{1 \leq k \leq |H|} \frac {\Pr [|S_t(H)|=k \big| H_t=H]} {3^{k}}.
\ee
\label{cor:anotherform}
\end{corollary}

\begin{proof}
Expanding \thmref{lowerbound} we have
\ba
\Coll_{\mu^{\text {CG}}_t} &= \frac{1}{2^n} \left(1 + \sum_{p,q \in \{0,3\}^n\backslash 0^n} \Pr [S_t = p | S_0 = q]\right)\\
&= \frac{1}{2^n} \left(1 + \sum_{H \subseteq [n]} \sum_{p,q \in \{0,3\}^n\backslash 0^n} \Pr [H_t = H | S_0 = q] \Pr [S_t = p | S_0 = q, H_t = H]\right)\\
&= \frac{1}{2^n} \left(1 + \sum_{H \subseteq [n]} \sum_{p,q \in
    \{0,3\}^n\backslash 0^n} \Pr [H_t = H] \Pr [S_t = p | S_0 = q, H_t
  = H]\right)
. 
\ea
Using \lemref{uniform-123} in the above we have
\ba
&= \frac{1}{2^n}  + \frac{1}{2^n}  \sum_{H \subseteq [n]}  \Pr [H_t = H] \sum_{p,q \in \{0,3\}^n\backslash 0^n} \Pr [S_t = p | S_0 = q, H_t = H],\\
&= \frac{1}{2^n}  + \frac{1}{2^n}  \sum_{H \subseteq [n]}  \Pr [H_t = H] \sum_{\substack{p,q \in \{0,3\}^n\backslash 0^n\\ p([n]\backslash H) = q([n]\backslash H)}} \frac 1 {{|H| \choose |p(H)|} 3^{|p(H)|}}\Pr [|S_t(H)|=|p(H)| \big| S_0 = q, H_t=H],\\
&= \frac{1}{2^n} +  \frac{1}{2^n}  \sum_{H \subseteq [n]}  \Pr [H_t = H] \sum_{\substack{q \in \{0,3\}^n}}\sum_{1 \leq k \leq |H|}\sum_{\substack{p \in \{0,3\}^n\backslash 0^n\\ p([n]\backslash H) = q([n]\backslash H)\\|p(H)| = k}} \frac 1 {{|H| \choose |p(H)|} 3^{|p(H)|}}\Pr [|S_t(H)|=|p(H)| \big| S_0 = q, H_t=H],\\
&= \frac{1}{2^n} + \frac{1}{2^n}  \sum_{H \subseteq [n]}  \Pr [H_t = H] \sum_{1 \leq k \leq |H|} \sum_{\substack{q \in \{0,3\}^n}}\sum_{\substack{p \in \{0,3\}^n\backslash 0^n\\ p([n]\backslash H) = q([n]\backslash H)\\|p(H)| = k}} \frac 1 {{|H| \choose k} 3^{k}}\Pr [|S_t(H)|=k \big| S_0 = q, H_t=H],\\
&= \frac{1}{2^n}  + \frac{1}{2^n}  \sum_{H \subseteq [n]}  \Pr [H_t = H] \sum_{1 \leq k \leq |H|} \sum_{\substack{q \in \{0,3\}^n}} \frac 1 {3^{k}}\Pr [|S_t(H)|=k \big| S_0 = q, H_t=H],\\
&=  \frac 1 {2^{n}} + (1-1/2^n) \sum_{H \subseteq [n]}  \Pr [H_t = H] \sum_{1 \leq k \leq |H|} \frac {\Pr [|S_t(H)|=k \big| H_t=H]} {3^{k}}.
\ea
\end{proof}

The standard coupon-collector bound is
\begin{lemma}[coupon collector]\OldNormalFont{}
 Let $H \subseteq [n]$. Then $\Pr [H_t \subseteq H] \leq e^{-(n-|H|)t/n }$.
\label{lem:coupon}
\end{lemma}

\begin{proof}
Let $E^{(i)}_H$ be the event that at step $i$ of the circuit a random gate lands completely inside the set $H$. Then $\Pr[E^{(i)}_H] = \frac{|H| (|H|-1)}{n(n-1)}$. Now $\Pr [H_t \subseteq H] = \Pr[E^{(i)}_H]^t = (\frac{|H| (|H|-1)}{n(n-1)})^t \leq (\frac{|H|}{n})^t  \leq e^{-(n-|H|)t/n }$.
\end{proof}

We now have all the pieces to prove \thmref{collisiontochain}.

\begin{proof}[\OldNormalFont{} Proof of \thmref{collisiontochain}]
Using corollary \ref{cor:anotherform} the total collision
probability is
\ba \Coll_{\mu^{\text{CG}}_t}
&=
\frac {1} {2^n} + (1-1/2^n)\sum_{H\subseteq [n]} \Pr[H_t=H] \sum_{k=1}^n
\frac{\Pr[|S_t(H)|=k|H_t=H]}{ 3^{k}} \nn \\ 
&=
\frac {1} {2^n} + (1-1/2^n)\sum_{H\subseteq [n]}\sum_{k=1}^n
\frac{\Pr[|S_t(H)|=k, H_t=H]}{ 3^{k}} \nn \\ 
&\leq
\frac {1} {2^n} + (1-1/2^n)\sum_{H\subseteq [n]}\sum_{k=1}^n
\frac{\Pr[|S_t(H)|=k, H_t \subseteq H]}{ 3^{k}} \nn \\ 
&\leq
\frac {1} {2^n} + \sum_{H\subseteq [n]} \Pr [H_t \subseteq H] \sum_{k=1}^n
\frac{\Pr[|S_t(H)|=k | H_t \subseteq H]}{ 3^{k}} \nn \\ 
& \leq \frac{1}{2^n}  +
\sum_{H\subseteq [n]} \Pr[H_t\subseteq H] \sum_{k=1}^n
\frac{P_t^{(|H|)}(k)}{ 3^{k}} \nn \\
&\leq \frac 1 {2^n}  +
\sum_{H\subseteq [n]} e^{-(n-|H|)t/n } \sum_kP_t^{(|H|)}(k) /
 3^{k}
& \text{\lemref{coupon}}
 \nonumber\\
        &= \frac 1 {2^n}  +\frac 1 {(1-\eps)}\sum_{m = 0}^n {n \choose m}
        e^{-mt/n }\|P^{(n-m)}_t\|_\ast.
& \text{setting }m=n-|H|
\ea

\end{proof}
\subsection{Proof of Proposition \ref{prop:nonincreasing}: collision probability is non-increasing in time}
\label{sec:nonincreasing}

When we try to recover the original chain from the accelerated chain
we find that $s$ steps of the accelerated chain typically correspond
to $t=O(s)$ steps of the original chain, but with a significant
variance.   This means that our bounds on the collision probability of
the accelerated chain translate only into bounds for a distribution of
times of running the original chain.  This issue can be addressed
using the following fact.

\begin{proposition}\OldNormalFont{}
$\E_{C \sim \mu^{\text{CG}}_t} \text{Coll}(C)$ is a non-increasing function of $t$.
\label{prop:nonincreasing}
\end{proposition}

\begin{proof}

$\Channel [G_{\mu^{(\text{CG})}_1}]$ corresponds to an average of
$n(n-1)/2$ projectors (using the Hilbert-Schmidt inner product). Hence it is a psd matrix with maximum eigenvalue $\leq 1$. Let $\alpha = \sum_{p \in \{0,3\}^n} \sigma_p \tensor \sigma_p$. \eqref{eq:expressionforcoll} may be written as
\begin{align}
\sum_{z \in \{0,1\}^n} & \Tr\left(\ket z \bra z \tensor \ket z \bra z
  \Channel \left[G^{(2)}_{\mu^{\text{CG}}_t}\right] \left(
    \ket{0^n}\bra{0^n} \tensor  \ket{0^n}\bra{0^n}\right)\right)
\nonumber = \\
&\Tr\left(\frac{\alpha}{2^n}  \Channel \left[G^{(2)}_{\mu^{\text{CG}}_t}\right] \left(
    \ket{0^n}\bra{0^n} \tensor  \ket{0^n}\bra{0^n}\right)\right)
\label{eq:equivalentcoll}
\end{align}
Using \eqref{eq:sumoverz}, terms of the form $\sigma_p \ot \sigma_q$
for $p \neq q$ in the decomposition of $\ket {0^n}\bra{0^n} \tensor
\ket {0^n}\bra{0^n}$ do not contribute to the collision probability.
Therefore using this observation and \eqref{eq:equivalentcoll}, the
collision probability after $t$ steps is proportional to $\Tr ( \alpha
\Channel [G_{\mu^{(\text{CG})}_1}]^t \alpha)$. Since $\Channel
[G_{\mu^{(\text{CG})}_1}]$ has all eigenvalues between 0 and 1, we
conclude that the collision probability after $t$ steps cannot
increase in $t$. 
\end{proof}

This argument relied on the starting state being
$\ket{0^n}$.  There exist starting states, such as
$\ket{+}^{\otimes n}$, for which the collision probability increases
when random gates are applied.

\subsection{Proof of \thmref{starmixing}: the Markov chain analysis}
\label{chain}

Consider the following birth-and-death Markov chain on the state space $\{0, 1,2, \ldots, n\}$.
\begin{equation}
P(k,l) : = 
\begin{cases}
\vspace{4mm}
1-\dfrac{2k (3n-2k-1)}{5n(n-1)} & l=k\\
\vspace{4mm}
\dfrac{2k(k-1)}{5n(n-1)} & l=k-1\\
\vspace{4mm}
\dfrac{6k(n-k)}{5n(n-1)} & l=k+1\\
0 & \text{otherwise}
\end{cases}
\label{weightchain}
\end{equation}

This Markov chain is reducible in general, however if we restrict the state space to $\{0\}$ or $\{1,2,\ldots ,n\}$ it is irreducible. 
Consider the following initial distribution over the state space $\{1,2,\ldots ,n\}$:
\be P^{(n)}_0 (k) = \frac{{n \choose k}}{2^n-1}  \qquad k \in
\{1,2,\ldots,n\} \ee
We claim that
\begin{lemma}\OldNormalFont{}
\be P_t^{(n)} = P^t P^{(n)}_0\ee
\end{lemma}

\begin{proof}
The proof follows from the fact that 
$\Pr [|S_t| = l \big| |S_0| = k]
= P^t(k,l)$
which was shown in 
Lemma 5.2 of \cite{HL09}.
\end{proof}

We now prove \thmref{starmixing} which gives a sharp upper bound on $\|P^{(n)}_t\|_\ast$. Throughout this section we drop the superscript $(n)$. Moreover we use the notation $X_t := |S_t|$.

{Proof overview:}
The philosophy of our analysis is to consider an acceleration of the
chain $P$: a chain with transition matrix $Q$ which is the same as $P$
but moves faster. As mentioned in the introduction, previous work \cite{HL09, BF13-2}
considered a ``fully accelerated'' chain, but we will instead carefully
choose the amount of acceleration so that the transition probabilities
are affine functions of $x$.    This will allow an exact solution of
the dynamics of this partially accelerated chain using a method of
Kac~\cite{K47}, as we describe in Section~\ref{sec:exact-OU}.  
We then analyze how much time should $P$
wait in each step of its walk in order to simulate steps of $Q$. In
order to do this we need to prove bounds on how many times each site
of the Markov chain has been visited during the accelerated walk and
based on that we count how many steps the original chain should
wait. This analysis is demonstrated in Section \ref{sec:waitabc}. Along the way during the wait-time analysis we will further
modify the partially accelerated chain to run in continuous time, so
that in time $t$ we sample $t'$ from $\text{Pois}(t)$ (the Poisson
distribution) and move $t'$ steps.  This resulting chain is also
exactly solvable, and the solution turns out to be extremely simple, and exemplifies the connection of the accelerated walk with the well-known Ornstein-Uhlenbeck process (see Proposition \ref{prop:Ybound}). We need to analyze the error from moving to continuous time, which
turns out to be a straightforward analysis of the Poisson
distribution.

Now suppose that the accelerated chain goes through a sequence of transitions $Y_0, Y_1, \ldots, Y_s$.

Let $p(x) = P(x,x+1)$ and $q(x) = P(x, x-1)$. We first consider the chain $P$ conditioned on moving at every single step. This chain at site $x$ has probability of moving forward and backwards $\dfrac{p(x)}{p(x)+q(x)}$ and $\dfrac{q(x)}{p(x)+q(x)}$, respectively. We can compute these probabilities
\begin{eqnarray}
Q^a(x,x)  &=& 0,\nonumber \\
Q^a(x, x+1)  &=& \dfrac{3 (n-x)}{3n-2x-1},\nonumber \\
Q^a(x,x-1)  &=& \dfrac{x-1}{3n-2x-1}.
\end{eqnarray}
Such a chain is called accelerated.  The chain $Q^a$ was used in 
\cite{HL09,DJ10,BF13} but we will not use it in this paper.

Instead of an accelerated chain we now define a partially accelerated chain as:
\begin{eqnarray}
Q^w(x,x) &=& w(x),\nonumber \\
Q^{w}(x,x+1) &=&  (1-w(x)) \dfrac{3 (n-x)}{3n-2x-1},\nonumber \\
Q^w(x,x-1) &=& (1-w(x)) \dfrac{x-1}{3n-2x-1}.
\end{eqnarray}
for arbitrary probability value $w(x)$. Setting $w(x) = \frac{2x}{3n-1}$ the partially accelerated chain becomes affine:
\begin{eqnarray}
Q(x,x) &=& \frac{2x}{3n-1},\nonumber \\
Q(x,x+1) &=&  \dfrac{3 (n-x)}{3n-1},\nonumber \\
Q(x,x-1) &=&  \dfrac{x-1}{3n-1}.
\label{accelerated}
\end{eqnarray}
By ``affine'' we mean that the transition probabilities are degree-1 polynomials in $x$.
%
%
%
Let $X_0, X_1, \ldots$ be the steps of the Markov chain evolving according to the transition matrix $P$ and $Y_0, Y_1, \ldots$ be the Markov chain according to $Q$. We now describe a coupling between these two.

\subsubsection{Coupling between $X$ and $Y$ chains}
\label{sec:XYcoupling}

For $x < \frac 56 n$ let $\alpha(x) = 1- \frac{p(x) + q(x)}{1- w(x)} = 1 -
\frac{2x(3n-1)}{5n(n-1)}$. If $0 < x < \frac 56 n$, $0 < \alpha (x) < 1$. So for this range we can view $\alpha(x)$ as a probability.

For $x \geq \frac 56 n$, let $\beta (x)$ be the solution to the following equation
\be
p(x) + q(x) = 1 - w(x) + \beta (x) w(x) ( p(x) + q(x)).
\ee
We can solve for $\beta(x)$ to find
\be
\beta(x) = \frac 1 {w(x)} \frac {-\alpha(x)}{1 - \alpha(x)}
= \frac {2x(3n-1)-5n(n-1)}{4x^2}.
\ee
For $x\geq \frac 56 n$ we have $\alpha(x) < 0$, so from the first
expression for $\beta(x)$ we see that $\beta(x) > 0$.  From the
second expression for $\beta(x)$ we can calculate the upper bound $\beta(x) \leq 1/4 + \frac6 {5 n}$.

%

\begin{mdframed}

\begin{coupling}
\OldNormalFont{}
The following describes a coupling between $X_0, X_1, \ldots$ and
$Y_0, Y_1,\ldots$.  It takes as input an arbitrary $x\in [n]$.  We
write $A \leftarrow a$ to mean that we assign value $a$ to variable $A$.

\bit
\item Set $X_0\leftarrow x$ and $Y_0\leftarrow x$.
\item Sample $Y_1,Y_2,\ldots,$ according to the Markov chain $Q$.
\item Set $s\leftarrow 0$ and $t\leftarrow 0$.
\item Repeat the following steps.
\bit
\item If $\alpha(X_t)>0$ then\\
{\em In this case, the $X$ chain may move more slowly than the $Y$
  chain, so one step of the $Y$ chain corresponds to one or more steps
  of the $X$ chain.}
\begin{enumerate}
\item With probability $1-\alpha(X_t)$, set $s\leftarrow s+1$.
\item Set $t\leftarrow t+1$.
\item Set $X_t \leftarrow Y_s$.
\end{enumerate}
\item Else\\
{\em Otherwise, there is the possibility of advancing the $X$ chain
  while the $Y$ chain waits.  This is only possible if $Y_s=Y_{s+1}$.}
\begin{enumerate}
\item If $Y_s \neq Y_{s+1}$ then 
\begin{enumerate}
\item[a.] Set $s \leftarrow s+1$.
\item[b.] Set $t \leftarrow t+1$.
\item [c.] Set $X_t \leftarrow Y_s$.
\end{enumerate}
\item Else
\begin{enumerate}
\item [a.] With probability $\beta(X_t)$, set $s \leftarrow s+1$
\item [b.] Otherwise (with probability $1-\beta(X_t)$)
\begin{enumerate}
\item $t \leftarrow t+1$
\item $X_t \leftarrow Y_t$
\end{enumerate}
\end{enumerate}
\end{enumerate}
\eit
\eit

\label{xy-coupling}
\end{coupling}
\end{mdframed}

\begin{definition}\OldNormalFont{}
For a tuple $L$ and a number $x$ let $L_{\text{left}(x)}$ be the same as $L$ except that we remove elements which are $> x$. Similarly define $L_{\text{right}(x)}$ to be the tuple resulted from removing the elements that are $< x$.
\end{definition}

\begin{theorem}\OldNormalFont{}
Assume $X_0 = Y_0$ and fix $s>0$, and let ${\cal Y} := (Y_0, Y_1, \ldots, Y_s)$. Define 
\be
S := \{ i : {\cal Y}_{\text{right}(\frac 56n)} [i] ={\cal Y}_{\text{right}(\frac 56n)} [i+1]\}.
\ee
Let
\be
T_{\text{left}} (\mathcal{Y}) = \sum_{y \in {\cal Y}_{ \text{left} (\frac 56 n)}} \mathsf{Geo}({\alpha_{y}})
\ee
and
\be
T_{\text{right}} (\mathcal{Y}) = \sum_{y \in S} \mathsf{Bern} (\beta(y))
\label{eq:betabern}
\ee
then the process in Coupling \ref{xy-coupling} satisfies
\be
Y_s = X_{s + T_{\text{left}}(\mathcal{Y}) - T_{\text{right}}(\mathcal{Y})}
\ee
\end{theorem}

\begin{proof}
We prove this by induction on Coupling~\ref{xy-coupling}. For the
base case we have $X_0 = Y_0$. Now suppose for $s > 0$, $Y_s = X_{s +
  T_{\text{left}} - T_{\text{right}}}$. Let $Y_{s+1}$ the
$s+1$-th step. There are two possibilities: if $Y_s <
\frac 56 n$, then $\alpha(Y_s) > 0$. In this case, $s$ will be
incremented once while $t$ may be incremented many times.  The number
of times $t$ will advance is distributed according to $\mathsf{Geo}(\alpha(Y_s))$.
Let $X'=Y_{s+1}$, i.e.~ the location on the chain after one step of $Q$. We show that $X'$ is distributed according to $X_{s + T_{\text{left}} - T_{\text{right}} + \mathsf{Geo}(\alpha(Y_s))}$. To see this note:
\bea
\Pr[X' = x | X_{s + T_{\text{left}} - T_{\text{right}}} = x] &=& \alpha(x) + (1-\alpha(x)) w(x) = 1 - p(x) - q(x) = P(x,x)\nonumber \\
\Pr[X' = x+1 | X_{s + T_{\text{left}} - T_{\text{right}}} = x] &=& (1-\alpha(x)) (1-w(x)) \frac{p(x)}{p(x) + q(x)} =p(x) = P(x,x+1)\nonumber \\
\Pr[X' = x-1 | X_{s + T_{\text{left}} - T_{\text{right}}} = x] &=& (1-\alpha(x)) (1-w(x)) \frac{q(x)}{p(x) + q(x)} =q(x) = P(x,x-1).
\eea
Now if $Y_s \geq \frac 56 n$ then if $Y_{s+1} \neq Y_s$, $X_{s +
  T_{\text{left}} - T_{\text{right}} + 1} = Y_{s+1}$. But if $Y_{s+1}
= Y_s$, then with probability $\beta(Y_s)$ the $X$ process skips this,
ie, $Y_s = X_{s + 1 +  T_{\text{left}} - T_{\text{right}} - 1}$. Let
$x \geq \frac 56 n$. Let $E_+$ be the event that $X_{s + T_{\text{left}} - T_{\text{right}} + 1} = x+1$ conditioned on $X_{s + T_{\text{left}} - T_{\text{right}}} = x$. Then
\be
\Pr (E_+) = (1-w(x)) \frac {p(x)}{p(x) + q(x)} + \beta(x) w(x) \Pr (E_+)
\ee
which implies that $\Pr (E_+) = P(x, x+1)$. Similarly if we define $E_-$ to be the event that $X_{s + T_{\text{left}} - T_{\text{right}} + 1} = x-1$ conditioned on $X_{s + T_{\text{left}} - T_{\text{right}}} = x$, then
\be
\Pr (E_-) = (1-w(x)) \frac {q(x)}{p(x) + q(x)} + \beta(x) w(x) \Pr (E_-)
\ee
which implies that $\Pr (E_-) = P(x,x-1)$. Using this, if $E_0$ is the event that $X_{s + T_{\text{left}} - T_{\text{right}} + 1} = x$ conditioned on $X_{s + T_{\text{left}} - T_{\text{right}}} = x$ then $\Pr[E_0] = \Pr [(E_+ \cup E_-)^{\mathsf{c}}] = P(x,x)$.

\end{proof}



We need the following two theorems which basically assert that 1) the wait time during the accelerated process is not too long and 2) the accelerated chain mixes after $O(n \ln^2 n)$ steps in the $\|\cdot \|_\ast$ norm.
\begin{theorem}[Wait-time bound]\label{thm:main-wait} \OldNormalFont{}
Let $Y_0, Y_1, \ldots, Y_s$ be $s$ steps of the accelerated Markov
chain defined in \eqref{accelerated} for $Y_0 \sim \Bin
(n,1/2)$, and $W_s$ be the number of steps Markov chain $X_0, X_1,
\ldots$ has waited after $s$ steps of the accelerated chain. Then for
$s = O(n \ln n)$, and for any constant $\alpha > 0$ there exists a
constant $c$ such that 
\be\Pr[T_{\text{left}}(s) \geq c n \ln^2 n] \leq 2^{-n}\cdot n^{-\alpha}.\ee
\end{theorem}

\begin{theorem}[Accelerated-chain mixing] \OldNormalFont{}
If $s \geq 3 n \ln n$ then
\be 
\|Q_s\|_\ast \leq \frac{27}{2^n + 1}(1+ \frac 1 {\poly(n)}).
 \ee
\label{thm:acc-mix}
\end{theorem}

Also, the following theorem combines Theorems \ref{thm:main-wait} and \ref{thm:acc-mix} to argue
that the original Markov chain mixes rapidly in the $\|\cdot\|_{*}$ norm.
\begin{proposition}
\torestate{\OldNormalFont{}
Let 
\be
\|f\|_\text{\Hsquare} := \sum_{k=1}^n |f(k)|\frac {3n} {k 3^k}
\ee
For any $t_0 \leq t_1 \leq t_2$:
\be
\E_{\tau} \|P_\tau\|_* \leq  \frac{t_0}{2} \cdot
\Pr[T_{\text{left}}(t_0) \geq t_1-t_0] + 
\frac 1T \sum_{t_0 \leq s \leq  4 t_2}\|Q_s\|_{\text{\Hsquare}} + 6 t_2  \frac{1}{1.4^{t_2}}
\ee
where, $T= t_2 - t_1+1$.
\label{prop:combining}}
\end{proposition}

\begin{proof}[\OldNormalFont{} Proof of \thmref{starmixing}]
We need to find suitable values for $t_0, t_1, t_2$. Let $t_0 = 3 n \ln n$ so that $\max_{s \in [t_0, t_2]}\|Q_s\|_{\text{\Hsquare}} \leq \frac{27}{2^n + 1}(1+ \frac 1 {\poly(n)})$ in Proposition \ref{prop:combining}. Next, choose $c$ to be large enough so that (using Theorem \ref{thm:main-wait}) if $t_1 = c  n \ln^2 n$
\be
\Pr[T_{\text{left}}(t_0) \geq t_1-t_0] \leq \frac{1}{2^n +1} \frac1 {n^3}.
\ee
Finally, let $c' > c$ be any constant and choose $t_2 =c' t_1$. Using Theorem \ref{thm:acc-mix} we conclude that:
\be
\frac 1 T \sum_{\tau = t_1}^{t_2} \|P_\tau\|_\ast \leq \frac{28}{2^n+1}.
\label{eq:averageP} 
\ee
This implies that there exists a value $t_1 \leq t^\ast \leq t_2$ such that 
\be
\|P_{t^\ast}\|_\ast \leq  \frac 1 T \sum_{\tau = t_1}^{t_2} \|P_\tau\|_\ast \leq \frac{28}{2^n+1}.
\ee
Since $t^\ast$ is related to $n \ln^2 n$ by a constant, this implies the proof.


\end{proof}

It remains to prove Theorems~\ref{thm:main-wait} and \ref{thm:acc-mix}
and \propref{combining}.  We prove \thmref{main-wait} in Sections
\ref{sec:waitabc} and \ref{sec:Ornsteininaction}, \thmref{acc-mix} in
Sections \ref{sec:exact-OU} and \ref{sec:exact-Q}, and \propref{combining} in \secref{propcombining}.


\subsubsection{Wait-time analysis}
\label{sec:waitabc}

In this section we prove Theorem \ref{thm:main-wait}.
Before getting to the proof we need some preliminaries. Sites with low
Hamming weight have the largest wait times. Hence, intuitively, we
want to say that during the accelerated walk, these sites are not hit
so often. More formally, let $N_x = \sum_{\tau=1}^s I   \{ Y_\tau \leq x
\}$ and let $\beta>1$.  
\begin{proposition}\OldNormalFont{}
Let $\nu=3/4 n$. For $x \leq \nu/ \beta$, $\Pr [N_x \geq \beta x ] \leq s^{3/2} e \cdot
e^{- \frac \beta 8 x}$. 
\label{prop:hitlow}
\end{proposition}
If we set  $\beta = 8(4+c)\ln n$ then \propref{hitlow} implies that
\be 
\Pr [N_x \geq \beta x ] \leq \frac 1
{{n \choose x} n^c}.\ee
Let $x(0)$ denote the corresponding $\nu/\beta$, i.e.
\be x(0) := \frac{\nu}{8(4+c) \ln n} .\ee

\begin{proof} 
We observe that $N_x$ conditioned on $Y_0 =z \geq 1$ is stochastically
dominated by the same variable conditioned on  $Y_0 =1$.  The proof is
by just taking the natural coupling that makes sure the latter walk
is always $\leq$ the former. Hence we can assume that the walk starts
out from $Y_0 =1$ and we will obtain a valid upper bound.

In \cite{HL09} (see the proof of lemma A.5) the authors show that
\begin{eqnarray}
\Pr [N_x \geq \beta x ] \leq \sum_{\tau=\beta x}^s \Pr  [Y_\tau \leq x ].
\label{Nx}
\end{eqnarray}
To understand these probabilities we will develop an exactly solvable
analogue for $Y_\tau$.   Although
$Y_\tau$ is a random walk in discrete time and space, we can
approximate it by a process that takes place in continuous time and
space.  If $Y_\tau$ were an unbiased random walk then we could
approximate it with Brownian motion.  However, it is biased to always
drift towards the point $\frac 34 n$.  The continuous-time-and-space
random process which diffuses like Brownian motion but is biased to
drift towards a fixed point is called the Ornstein-Uhlenbeck process.
We will not prove a formal connection between $Y_\tau$ and the
Ornstein-Uhlenbeck process, but instead will prove bounds on $Y_\tau$
that are inspired by the analogous facts about Ornstein-Uhlenbeck.

\begin{proposition} [Connection with the Ornstein-Uhlenbeck process] \OldNormalFont{}
Define 
\be 
\nu_\tau := z e^{-\frac{4\tau}{3n}} + \frac{3}{4}n\left(1-e^{-\frac{4\tau}{3n}}\right).
\label{eq:mu-tau-def} \ee
Then we can bound
\be
\Pr [Y_\tau \leq x ] \leq \sqrt \tau e \cdot e^{-\frac{(\nu_\tau-x)^2}{2 \nu_\tau}}
\ee
\label{prop:Ybound}
\end{proposition}
The proof is in Section~\ref{sec:Ornsteininaction}. 

This proposition is inspired by the fact that the exact solution to
the Ornstein-Uhlenbeck process is a Gaussian with mean and variance
both equal to $\nu_\tau$.  We can see that once $\tau \gtrsim n\ln n$,
this is close to a Gaussian centered at $\frac 34n$, i.e. the
stationary distribution.

Note that $\nu_\tau$ is an increasing function of $\tau$, and
furthermore, for $\nu_\tau \geq x$, $e^{-\frac{(\nu_\tau-x)^2}{2
    \nu_\tau}}$ is decreasing in $\nu_\tau$, and therefore $\tau$.  Hence the sum in \eqref{Nx} can be bounded by
\be
\Pr [N_x \geq \beta x ] \leq s^2 e \cdot \exp\left({-\frac{(\nu(1-e^{-\frac{\beta x}\nu})-x)^2}{2 \nu (1- e^{-\frac{\beta x}\nu})}}\right).
\ee 

Using the following inequalities
\be
\frac u{1+u} \leq 1-e^{-u} \leq u, \text{ for } u \leq 1.
\ee
we find that
\be
\Pr [N_x \geq \beta x ] \leq s^{3/2} e \cdot \exp\left({-\frac{\left(\frac{\beta x}{1 + \frac{\beta x}\nu}-x\right)^2}{2 \beta x}}\right).
\ee
Since $\frac{\beta x}\nu < 1$ then  
\be
\Pr [N_x \geq \beta x ] \leq s^{3/2} e \cdot e^{- \frac \beta 8 x}.
\label{eq:Nx-tail-bound}
\ee
\end{proof}

Now following \cite{HL09,DJ10,BF13},
define the good event $A :=  \cap_{1 \leq x \leq x(0)} \{ N_x \leq
\beta \cdot x \}$. Recall that $\beta = 8(4+c)\ln n$ and $x(0)=\nu/\beta$.

\begin{proposition}\label{prop:Hc-bound} \OldNormalFont{}
$\Pr [A^{\mathsf{c}}|Y_0] \leq \frac{2}{{n \choose Y_0} n^{c-1}}$.
\end{proposition}

To prove \propref{Hc-bound}, we will need a bound on the minimum site
visited during the accelerated walk. Let $M_s := \min_{1 \leq i\leq s}
\{Y_i   \}$. Then 
\begin{proposition}\label{prop:Ms-bound} \OldNormalFont{}
$\Pr   [ M_s \leq a   | Y_0 = z  ] \leq  s\frac{{n \choose a} 3^a}{{n \choose z} 3^z}$
\end{proposition}

We need the following lemma which is a standard fact about Markov chains.

\begin{lemma}\OldNormalFont{}
Let $Y_0, \ldots$ be a Markov chain with stationary distribution $\pi$ then for any $x, y$ in the state space and integer $s > 0$
\be
\Pr [ Y_s = y   | Y_0 = x ] \leq \frac{\pi_y}{\pi_x}.
\ee
\label{lem:ret}
\end{lemma}
\begin{proof} 
\bea
\Pr [ Y_s = y   | Y_0 = x ] &=&\frac 1 {\pi_x} \pi_x \Pr [ Y_s = y   | Y_0 = x ]\\
&\leq&\frac 1 {\pi_x}\sum_z \pi_z \Pr [ Y_s = y   | Y_0 = z ]\\
&\leq&\frac {\pi_y} {\pi_x}
\eea
\end{proof}

\begin{proof}[\OldNormalFont{} Proof of \propref{Ms-bound}]
\begin{align}
\Pr   [ M_s \leq a   | Y_0 = z  ] &\leq \Pr   [  \cup_{1 \leq i \leq s} \{ Y_i = a  \}  | Y_0 = z ]\nonumber \\
&\leq  \sum_{j = 1}^s \Pr  [Y_j = a  | Y_0=z]\nonumber \\
&\leq  s \cdot \frac{\pi_a}{\pi_z} & \text{using \lemref{ret}}\nonumber \\
&=  s \cdot \frac{{n \choose a} 3^a}{{n \choose z} 3^z}.
\end{align}
\end{proof}

Now we show that the event $A= \cap_{1 \leq x \leq x(0)} \{ N_x \leq
\beta \cdot x \}$ is very likely. 
\begin{proof}[\OldNormalFont{} Proof of \propref{Hc-bound}]
The proof is very similar to the proof of lemma 4.5 in Brown and Fawzi \cite{BF13}.
\begin{eqnarray}
\Pr[A^{\mathsf{c}}] &=& \Pr  [  \cup_x N_x > \beta \cdot x ]\nonumber \\
&\leq& \sum_x \Pr  [ N_x > \beta \cdot x ]\nonumber \\
&\leq& \sum_{ x<M_s}  \Pr  [ N_x > \beta \cdot x ] +
\sum_{M_s\leq x<Y_0}  \Pr  [ N_x > \beta \cdot x ] +
 \sum_{x(0) \geq x\geq Y_0} \Pr  [ N_x > \beta \cdot x ]\label{eq:three-terms}
\end{eqnarray}
In the last line we have used the fact that $M_s\leq Y_0$.
Now we will handle each term in \eq{three-terms} separately.  
When $x<M_s$, $N_x=0$, so $\sum_{ x<M_s}  \Pr  [ N_x > \beta \cdot x
]=0$.
Next when $x\geq Y_0$, we can use \propref{hitlow} to bound
$\Pr[N_x>\beta x] \leq \binom{n}{Y_0}n^{-c}$.
Finally, when
 $M_s\leq x<Y_0$, we have
\begin{align}
\Pr  [ N_x > \beta \cdot x ]&=  \Pr  [ N_x > \beta \cdot x  | M_s \leq x ] \Pr  [M_s \leq x  ]\nonumber \\
&\leq \Pr  [ N_x > \beta \cdot x  |Y_0 = 1  ] \Pr  [M_s \leq x  ]\nonumber \\
&\leq \Pr [M_s \leq x  ] \cdot \frac{1}{{n\choose x} n^{c}} &
                                                               \text{using \propref{hitlow}} \\
&\leq\frac{{n\choose x}}{{n\choose Y_0} 3^{Y_0-x}} \cdot
  \frac{1}{{n\choose x} n^{c}}
& \text{using \propref{Ms-bound}}
\end{align}

We now combine these contributions and sum over $x$ to obtain
\ba
\Pr[A^{\mathsf{c}}] 
&\leq s\frac{1}{{n\choose Y_0}n^c} \sum_{x<Y_0}
3^{Y_0-x}+ \sum_{x(0)\leq x \leq Y_0} \frac{1}{{n\choose Y_0} n^{c}}\\
&\leq s\frac{1}{2 {n\choose Y_0}n^c} + \frac{1}{{n\choose Y_0} n^{c-1}}\\
&\leq \frac{2}{{n\choose Y_0} n^{c-2}}.
\ea
\end{proof}


\begin{proof}[\OldNormalFont{} Proof of Theorem \ref{thm:main-wait}]
Recall that the initial position on the chain $Y_0$ is distributed according to a binomial around $n/2$. Hence it is enough to show that starting from position $Y_0$ on the chain the probability that the wait time is larger than the bound stated in the theorem is bounded by
\be\frac{1}{{n \choose Y_0} \poly(n)}.\ee
 If such bound holds than the probability of waiting too long is bounded by 
\be\frac 1{2^n-1} \sum_{Y_0 =1}^n \frac{{n\choose Y_0}}{{n \choose Y_0} \poly(n)} = \frac 1 {2^n+1}\cdot  \frac 1{\poly(n)}.\ee
 We achieve this in the following.

Let $a$ be a constant. Consider the following bound on the wait-time random variable $T_{\text{left}}(s) = T_{\text{left}}({Y_0}) + \ldots + T_{\text{left}}({Y_s})$:
\begin{eqnarray}
\Pr  [ T_{\text{left}}(s) \geq a  ] &\leq& \Pr [ T_{\text{left}}(s) \geq a  | A ] + \Pr  [ A^{\mathsf{c}} ]\nonumber \\
&\leq& \sum_{m=1}^{Y_0}\Pr [ T_{\text{left}}(s) \geq a  | A, M_s=m ] \Pr [M_s=m ] + \Pr  [ A^{\mathsf{c}} ]\nonumber \\
&\leq& \frac{1}{{n\choose Y_0} 3^{Y_0}}\sum_{m=1}^{Y_0}\Pr[ T_{\text{left}}(s) \geq a | A, M_s=m] {n \choose m} 3^m + \frac{2}{{n\choose Y_0} n^{c-2}},
\label{unbelievable}
\end{eqnarray}
using Propositions \ref{prop:Ms-bound} and \ref{prop:Hc-bound}.

Let $\rho_x = N_x - N_{x-1}$ be the number of times site $x$ has been visited during $s$ rounds of the accelerated walk. Recall from Section \ref{chain} that
\be
T_{\text{left}}(s) \preceq \sum _{x=1}^{5n/6} \rho_x \cdot \mathsf{Geo}({\frac{6 x}{5 n}}).
\ee

Hence we need a concentration bound for sums of geometric random variables. Fortunately we know the following Chernoff-type tail bounds on the sum of geometric random variables.
\begin{theorem}[Janson~\cite{J14}]\OldNormalFont{}
Let $G = \sum_{i=1}^n \mathsf{Geo}({p_i})$ be the sum of independent geometric random variables with parameters $p_1, \ldots, p_n$, and let $p^\ast  = \min_i p_i$ and $\phi := \sum_{i=1}^n \frac{1}{p_i} = \E G$, then for any $\lambda \geq 1$
\be
\Pr [G \geq \lambda \phi] \leq \frac{1}{\lambda} (1-p^\ast)^{(\lambda-1-\ln \lambda) \phi},
\ee
\label{Janson}
\end{theorem}

The bound we need for our results is:
\begin{corollary}\OldNormalFont{}
Let $G$ be sum of $s$ geometric random variables with
parameters $p^\ast = p_1 \leq \ldots \leq p_s$, and $\E G = \phi$.
If $T > 3 c \ln( c) \phi$, then 
\be 
\Pr [G > T] \leq \frac{1}{3 c \ln c} (1-p^\ast)^{T(1-1/c)}.
\ee
In particular, we can say that if $T > \E W$, then for any constant $c$ there exists a constant $c'$ such that 
\be 
\Pr[G > c'  T] \leq (1-p^\ast)^{c T}.
\ee
\label{corjanson}
\end{corollary}

\begin{proof} 
It is enough to show that if $\lambda > 3 c \ln c$ then $\lambda -1-\ln \lambda > \lambda (1-1/c)$. Let $f(\lambda) :=\frac{ \lambda}{c} - \ln (e\lambda)$ for $c>1$, we observe that $f$ is an increasing function for $\lambda >c$. We need to find a point $\lambda^\ast$ such that $f(\lambda^\ast) >0$, and one can check that $\lambda^\ast = 3 c \ln c$ works.
\end{proof}

In order to employ Corollary \ref{corjanson} in the context of wait
time (specifically \eqref{unbelievable}) we just need to find an upper
bound on the expected wait time. Now we condition on $A$. Hence for $x
\leq x(0)$, $N_x \leq \beta x$. Among all possibilities given by event
$A$, the wait time is maximized when the minimum visited site ($M_s$)
is  visited as often as possible  (see Brown-Fawzi~\cite {BF13} for a
discussion). So it will suffice to bound the wait time for the
situation when $x \leq x(0)$, $\rho_x = \beta$ and for $x = x(0)$,
$\rho_x = s  - \beta x(0)$.  In this case, the expected wait time
(conditioned on any starting point) is bounded by  
\be
\E[T_{\text{left}}(s) | A] \leq \beta \sum_{1 \leq x \leq x(0)} \frac{5 n}{2 x} + (s - \beta x(0)) \frac{5 n}{ 2 x(0)}
\ee
Assuming the parameters in Proposition \ref{prop:hitlow} we find that $\E[T_{\text{left}}(s) | A]  = O(n \ln^2 n + s \ln n)$, and in particular if $s = O(n \ln n)$ then $\E[T_{\text{left}}(s) | A] = O(n \ln^2 n)$.

Therefore using Lemma \ref{Janson} for any $C >0$ there exists a large enough constant $C'$ such that
\be
\Pr[ T_{\text{left}}(s) \geq C' n \ln^2 n | H, M_s=m] \leq e^{- C \cdot \frac{m}{n} \cdot n \ln^2 n}.
\ee
Combining this with \eqref{unbelievable} and choosing $C$ large enough yields
\begin{eqnarray}
\Pr  [ T_{\text{left}}(s) \geq C' n \ln^2 n ] &\leq& \frac{1}{{n\choose Y_0} 3^{Y_0}}\sum_{m=1}^{Y_0}e^{- C \cdot \frac{m}{n} \cdot n \ln^2 n} {n \choose m} 3^m +  \frac{2}{{n\choose Y_0} n^{c-2}}\nonumber \\
&\leq& \frac{3}{{n \choose Y_0} \cdot n^{c-2}}.
\label{unbelievable2}
\end{eqnarray}
and this completes the proof.
\end{proof}

\subsubsection{Proof of Proposition \ref{prop:Ybound}: Connection with the Ornstein-Uhlenbeck process}
\label{sec:Ornsteininaction}
We first define a new Markov chain $S'_0, S'_1, S'_2, \ldots$ which is
easier to analyze and gives us useful bounds for the Markov chain $S_0, S_1, S_2, \ldots$.
\begin{definition} \OldNormalFont{} $S'_0, S'_1, S'_2, \ldots$ is the following Markov
  chain. The state space is $\{0,1\}^{n}$. 
The initial string $S'_0$ is sampled uniformly at random from
$\{0,1\}^n \backslash 0^n$. At each step $t$, $S'_{t+1}$ results from
$S'_t$ by picking a random position of $S'_t$. If it was a zero we
flip it, otherwise if it was a $1$ with probability $1/3$ we flip it
and with probability $2/3$ it doesn't change. 
\label{def:decoupled}
\end{definition}

The Hamming weight of these strings corresponds to the position on a birth-and-death chain on the state space $\{0,1,2,\ldots, n\}$. Given a string $S' \in \{0,1\}^{n}$ the probability that the Hamming weight of $S'$ increases by $1$ is $1- x/n$ and the probability that it decreases is $\frac{x}{3 n}$. Let $Q'$ be the transition matrix describing the Hamming weight.
  
We now claim that:
 \begin{proposition} \OldNormalFont{}
 Starting from a string of Hamming weight $\geq 1$, at any time $t$,
 $Y_t$ stochastically dominates $Y'_t$, meaning that
 \be\Pr [Y'_t \geq k] \leq \Pr [Y_t \geq k]\ee
 \label{prop:stochasticdomination}
 \end{proposition}
 
 \begin{proof} 
 It is enough to observe that for $0 \leq x \leq n$, the probability of moving forward for $Q$ is larger than the probability of moving forward for $Q'$, and also the probability of moving backwards for $Q$ is smaller than the probability of moving backwards for $Q'$. 
 \end{proof}

Now suppose that we simulate $Q'$ for $T$  steps. First,
instead of considering $T$ steps we consider this number to be a
Poisson random variable $T \sim \mathsf{Pois}(\tau)$, where $\tau$ is
some positive real number.
Let $f_l$ be the number of times that site $l$ is hit after $T$
steps. Then $(f_1, \ldots, f_n) \sim \mathsf{Multi} (T, \frac{1}{n},
\ldots, \frac{1}{n})$ is the number of times each position in $[n]$ is
hit after $T$  steps. Here
$\mathsf{Multi} (T, \frac{1}{n}, \ldots, \frac{1}{n})$ is the
multinomial distribution over $n$ items summing up to $T$, each
happening with probability $1/n$.   

We can then consider $T$ in turn to be a random variable
distributed according to  $T \sim \mathsf{Pois}(\tau)$.  
It turns out that defining $T$
in this way will make $f_1,\ldots,f_n$ independent.  
Moreover, for any $l \in \{1,\ldots, n\}$,
\be 
f_l \sim \mathsf{Pois} (\tau/n).
\ee
In other words, the number of times each site is hit is independently
distributed according to a Poisson distribution. This technique is
sometimes called
Poissonization. 


Now suppose the $l$'th bit of $S'_0$ starts out from $0$ and that $f_l = k$. 
We find that the probability of ending up with a $1$ in this case is
\be 
p_k = \frac{3}{4}  \left(1 - \left(\frac{-1}{3}\right)^k \right),
 \ee
and the probability of reaching a $0$ is
\be 
1-p_k = \frac{1}{4} + \frac{3}{4} \left(\frac{-1}{3}\right)^k.
 \ee
 Using these two probabilities and taking the expectation over the Poisson measure we can compute
\begin{eqnarray}
\Pr[S'_T [l]=1 | S'_0[l] = 0] &=& 
\sum_{k=0}^\infty \frac{e^{-\tau/n}}{k!} (\tau/n)^k \left(\frac{3}{4}-\frac{3}{4} \left(-1/3\right)^k\right)\nonumber \\
&=& \frac{3}{4} \left(1-e^{-\frac{4\tau}{3n}}\right)\nonumber \\
&=:& \alpha_\tau.
\end{eqnarray}
Note that the $T$ on the LHS is still a random variable distributed
according to $\Pois(\tau)$.

For the case when the $l$'th bit starts out equal to $1$ and $f_l = k $, we find the probabilities in a similar way. The probability of ending up in bit $1$ is 
\be 
p_k = \frac{3}{4} + \frac{1}{4} \left(\frac{-1}{3}\right)^k,
 \ee
and the probability of ending up in $0$ is
\be 
1-p_k =  \frac{1}{4} - \frac{1}{4} \left(\frac{-1}{3}\right)^k.
 \ee
We then compute
\begin{eqnarray}
\Pr[S'_T [l] =1 | S'_0 [l] \neq 0] 
&=& \sum_{k=0}^\infty \frac{e^{-\tau/n}}{k!} (\tau/n)^k \left(\frac{3}{4}+\frac{1}{4} \left(-1/3\right)^k\right)\nonumber \\
&=& \frac{3}{4} + \frac{1}{4}e^{-\frac{4\tau}{3n}}\nonumber \\
&=:& \beta_\tau.
\end{eqnarray}

As a result conditioned on $|S'_0|=z$,
\be 
Y'_T \sim Y'_{\Pois(\tau)} \sim \mathsf{Bin} (n-z, \alpha_\tau) + \Bin (z, \beta_\tau).
 \ee
This has expectation equal to
\be 
\E  \left[Y'_T  | Y'_0 = z  \right] =  z e^{-\frac{4\tau}{3n}} +
\frac{3}{4}n\left(1-e^{-\frac{4\tau}{3n}}\right),
 \ee
which is simply equal to $\nu_\tau$, which was first introduced in \eqref{eq:mu-tau-def}.
Next, using a simple Chernoff bound for sum of binomial random
variables we can show that for all $x < \nu_j$
\begin{equation}
\Pr  [ Y'_{\Pois(\tau)} \leq x  ] \leq  e^{-\nu_\tau \frac{(1-x/\nu_\tau)^2}{2}} =  e^{-\frac{(\nu_\tau-x)^2}{2 \nu_\tau}}.
\label{dynamics}
\end{equation}
 This bound is exactly the one that we expect from an
Ornstein-Uhlenbeck process. 

Fix a number $x \in [n]$. Let $B$ (the bad event) be $\{|S'_T| \leq k\}$. Then
\ba
\Pr[B] &= \sum_{s=0}^\infty \Pr [T = s] \Pr [B | T = s ]\\
&\geq \Pr [T = \tau] \Pr [B | T = \tau]\\
&\geq \Pr [T = \tau] \Pr [|S'_\tau| \leq x ]\label{eq:badevent}
\ea 
We can evaluate $\Pr [T = \tau] = \frac{\tau^\tau}{\tau !} e^{-\tau} \geq \frac{1}{\sqrt \tau e}$, where we have use 
Stirling's formula (from wikipedia) which states that $\frac{\tau !}{(\tau /e)^\tau}
\leq e\sqrt{\tau}$. Together with the bound in \eqref{eq:badevent} we find that
\be
\Pr [Y'_\tau \leq x ] \leq \sqrt \tau e \cdot \Pr[B]
\ee
Combining this inequality with \eqref{dynamics} we conclude that
\be
\Pr [Y'_\tau \leq x ] \leq \sqrt \tau e \cdot e^{-\frac{(\nu_\tau-x)^2}{2 \nu_\tau}}
\ee
Using Proposition \ref{prop:stochasticdomination}
\be
\Pr [Y_\tau \leq x ] \leq \sqrt \tau e \cdot e^{-\frac{(\nu_\tau-x)^2}{2 \nu_\tau}}
\ee

If $\tau \geq \frac 34 n \ln n$ then $\frac 34 n \geq \nu_\tau \geq \frac 34 n-1$. Therefore
\be
\Pr [Y_\tau \leq x ] \leq \sqrt {\frac 34 n \ln n} e \cdot e^{-\frac{2 (\frac 34 n -x -1)^2}{3 n}}
\ee

\subsubsection{Proof of Theorem \ref{thm:acc-mix}:  exact solution to the Markov chain $Q$}
\label{sec:exact-OU}

In this section we give an exact solution to the Markov chain $Q$
defined in Section \ref{chain}. Here, by giving an exact solution we
mean we can find the eigenvalues and eigenvectors of the transition
matrix explicitly and evaluate the norm $\|Q_t\|_\ast$. The
construction follows nearly directly from a result of Kac~\cite{K47}.

Recall the transition probabilities of Markov chain $Q$ according to
Equation \eqref{accelerated}. In \eqref{accelerated}, $Q$ is defined over the state space $[n]$. Without loss of generality and for convenience we can relabel the state space to $\{0,1,2,\ldots, n-1\}$ and redefine the transition matrix according to:
\begin{eqnarray}
p_i &:=& Q(i,i+1) = \dfrac{3 (n-i-1)}{3n-1},\\ \nonumber
q_i &:=& Q(i,i-1) = \dfrac{i}{3n-1},\\ \nonumber
r_i &:=& Q(i,i) = \dfrac{ 2 (i+1)}{3n-1}.
\label{coeff}
\end{eqnarray}
for $i \in \{0,1,2,3, \ldots, n-1\}$.

Now we consider the eigenvalue problem
\be 
x^ {(\lambda)} Q = \lambda x^{ (\lambda)},
 \ee
where $x^{(\lambda)}$ is a row vector with entries $x^{(\lambda)}(i)$, is the left eigenvector corresponding to the eigenvalue $\lambda$. For now we drop the superscript $\lambda$ in $x^{(\lambda)}$. Expanding this equation we have
\be 
p_{i-1} x({i-1}) + r_i x(i) + q_{i+1} x({i+1})= \lambda x(i).
\label{eq:3-term-recur}
 \ee

Notice that $q_0 = p_{n-1} = 0$. Define the generating function
\be 
g_\lambda(z) = \sum_{i=0}^\infty x(i) z^i,
 \ee
where for $i\geq n$, we set $x(i)=0$.  It suffices to solve
\eqref{eq:3-term-recur} subject to the boundary conditions $x_{-1} =
x_{n} = 0$.
For $i>0$ we can write

\be 
p_{i-1} x(i-1) z^i + r_i x(i) z^i + q_{i+1} x(i+1) z^i= \lambda x(i) z^i,
 \ee
assuming $x_{-1}=0$. Using the coefficients of \eqref{coeff} we get
\be 
\dfrac{3(n-i)}{3n-1} x(i-1) z^i +  \dfrac{2 (i+1)}{3n-1} x(i) z^i +  \dfrac{i+1}{3n-1} x(i+1) z^i= \lambda x(i) z^i.
\label{eq:eigenvalue}
 \ee

For $i=0$ the equation is
\be 
x_{1}= ((3 n-1) \lambda - 2)x(0).
 \ee
Summing ($\sum_{i=0}^\infty$) over the first term in the left-hand side of \eqref{eq:eigenvalue} we obtain
\be
\frac{3(n-1)}{3n-1} z \cdot g_\lambda(z) - (\frac{3}{3n-1})z^2 \frac{d}{dz} g_\lambda(z).
\ee
Similarly for the second term we get
\be
\frac{2}{3n-1} g_\lambda(z) + (\frac{2}{3n-1})z\frac{d}{dz} g_\lambda(z),
\ee
and for the third term
\be
(\frac{1}{3n-1})\frac{d}{dz} g_\lambda(z),
\ee
and for the term on the right-hand side
\be
\lambda g_\lambda(z)
\ee

Let $\lambda' = \lambda \frac{3n-1}{3(n-1)} - \frac{2}{3(n-1)}$. Putting all of these together we obtain the following first order differential equation
\be 
\dfrac{1}{g_\lambda(z)}\dfrac{d}{dz} g_\lambda(z)=  (n-1) \dfrac{3 \lambda' -3z }{- 3 z^2 +2 z + 1},
 \ee
with the boundary conditions
\begin{eqnarray}
g_\lambda(0) &=& x(0),\\
\dfrac{d^{n}}{d z^{n}} g (0) &=& 0.
\end{eqnarray}

Assume $n-1$ is divisible by $4$. Solving this differential equation and applying the first boundary condition ($g_\lambda(0) = x(0)$) we get
\be 
g_\lambda(z) = x^{(\lambda)}(0) (1+3z)^{\frac{n-1}4 (1+3\lambda')} (1-z)^{\frac{n-1}4 (3-3\lambda')}.
 \ee
The second boundary condition basically says that $g_\lambda(z)$ should be a
polynomial of degree at most $n-1$. This implies that $3 \lambda'
(n-1)/4$ should be an integer.   Since the exponents of both the
$(1+3z)$ and the $(1-z)$ terms should be nonnegative, we can further
constrain $3\lambda'(n-1)/4$ to lie in the interval
$[-\frac{n-1}{4},3\frac{n-1}4]$.
These constraints are enough to determine the $n$
eigenvalues $\lambda_0,\ldots,\lambda_{n-1}$.  They must (up to an
irrelevant choice of ordering) satisfy
\be 3\lambda'_m\frac{n-1}4 = 3\frac{n-1}{4} - m.\ee
Rearranging and solving for $\lambda_m$ we have 
\be \lambda_m = 1-\frac{4m}{3n-1}.\ee
The eigenvalue gap is exactly $\dfrac{4}{3 n-1}$. Note for $m=0$ we get $\lambda_0 =1$ and
\be 
g_1(z) = x^{(1)}(0) (1+3 z)^{n-1} = x^{(1)}(0) \sum_{i=0}^{n-1} {n -1
  \choose i} 3^i z^i
= \sum_i \pi(i) z^i.
 \ee
In the last equation we have introduced $\pi(i)$, which is the
stationary distribution. This is a binomial
centered around $\frac 34(n-1)$ and shifted by 1.  Its mean
$\frac 34n+\frac 14$ differs from that of the non-accelerated
chain by an offset of $\approx \frac 14$.  We might expect a shift
like this because the accelerated chain spends less time on lower
values of $x$.

Since the stationary distribution has unit 1-norm we can evaluate 
\be
x^{(1)} = \frac 1 {4^{n-1}}
\label{eq:x1}
\ee

The eigenvectors for each eigenvalue $\lambda$ can be indirectly read from the generating function $g_\lambda(z)$. We use the notation $x^{(\lambda)}$ for the eigenvector corresponding to eigenvalue $\lambda$. Also we denote the $i$-th component of these vectors by $x^{(\lambda)}(i)$, for $i \in \{0,1,2,3, \ldots, n-1\}$.

\subsubsection{Exact solution to the Markov chain $Q$ implies a good upper bound on $\|Q_t\|_{\small\Box}$}\label{sec:exact-Q}

We want to use the above exact solution to derive a bound on
$\|Q_t\|_{\small\Box}$.  We begin by stating some facts.
\begin{enumerate}
\item $\lambda_m = 1-\frac{4m}{3n-1}\leq e^{-  \frac {4m}{3n-1}}$ for $m \in [0, n-1]$.
\item $g_m(z) = x^{(m)}(0) (1+3z)^{n-m-1} (1-z)^{m} = \sum_{i=0}^{n-1}x^{(m)}(i) z^i$ for $m \in [0, n-1]$.
\item $x^{(m)} Q = \lambda_m x^{(m)} =  (1-\dfrac{4m}{3n-1} ) x^{(m)}$
  for  $m \in \left [0, n-1\right]$.
\item $Q$ is a reversible Markov chain on $\{0,\ldots,n-1\}$ with
  stationary distribution $\pi(i) = \binom{n-1}{i} 3^i / 4^{n-1}$.
\end{enumerate}

Since $x^{(m)}$'s are the left eigenvectors of $Q$, they can be used
to find the right  eigenvectors $y^{(n)}$:
\be 
y^{(m)}(i) = \dfrac{x^{(m)}(i)}{\pi(i)}.
 \ee
Left and right eigenvectors are orthonormal with respect to each other, i.e., for any $l,m\in [n-1]$ 
\be 
\sum_{i=0}^{n-1} x^{(m)}(i) y^{(l)}(i) = \sum_{i=0}^{n-1} \dfrac{x^{(m)}(i)x^{(l)}(i)}{\pi(i)} = \delta_{m,l}.
 \ee
We define the following inner product between functions
\begin{equation}
(f,g) := \sum_{i} \dfrac{1}{\pi(i)} f(i) g(i),
\label{()norm}
\end{equation}
according to which $\{x^{(m)} : m \in [n-1]\}$ forms an orthonormal basis, i.e.,
\begin{equation}
(x^{(i)},x^{(j)}) := \delta_{i,j}.
\end{equation}

We denote the initial distribution by $Q_0(i) =
\dfrac{1}{2^n-1}{n\choose i+1}$. Also we denote the eigenvector corresponding to eigenvalue $1$ with $x^{(1)} = \pi$, which is the same as the stationary distribution. We write this initial vector as a combination of eigenvectors of the chain
\be 
Q_0 = \sum_{i=0}^{n-1} \alpha_i x^{(i)}
\qquad \text{with}\qquad
\alpha_i = (x^{(i)}, Q_0).
 \ee
Therefore after $t$ steps
\begin{eqnarray}
Q_t &=& \sum_{m=0}^{n-1} \alpha_m \lambda_m^t x^{(m)} = \sum_{m=0}^{n-1} (x^{(m)}, Q_0) \lambda_m^t x^{(m)},\nonumber \\
 &=& \sum_{m=0}^{n-1} (x^{(m)}, Q_0) \lambda_m^t x^{(m)}.
\label{eq:Qtexpansion}
\end{eqnarray}
We are interested in
\be 
\|Q_t\|_{\text{\Hsquare}}:= (1 - 1/2^n) \frac {12} {2^n} (Q_0, Q_t) \leq \frac {12} {2^n} (Q_0, Q_t)
 \ee
 
Using Equation \eqref{eq:Qtexpansion} this can be evaluated as
\begin{eqnarray}
\|Q_t\|_{\text{\Hsquare}} &\leq& \frac {12} {2^n}  ( \sum_{m=0}^{n-1} (x^{(m)}, Q_0) \lambda_m^t x^{(m)} , Q_0)\\
&=& \frac {12} {2^n} \sum_{m=0}^{n-1} (x^{(m)}, Q_0)^2 \lambda_m^t\\
&=& \frac {12} {2^n} \sum_{m=0}^{n-1} \alpha_m^2 \lambda_m^t
\label{eq:errr}
\end{eqnarray}

As a result the problem reduces to evaluating the overlaps $\alpha_m = (x^{(m)}, Q_0)$.
\begin{eqnarray}
\alpha_m &=&(x^{(m)}, Q_0)\nonumber \\
&=&\sum_{i =0}^{n-1} x^{(m)} (i) \frac{\frac{{n \choose i+1}}{2^n-1}}{\frac{{n-1 \choose i}}{4^{n-1}} 3^i}\\
&=&3 \cdot \frac{4^{n-1}}{2^n-1 }\sum_{i =0}^{n-1} x^{(m)} (i) \frac{n}{(i+1) \cdot  3^{i+1}}\\
&=&3n \cdot \frac{4^{n-1}}{2^n-1 }\int_{z = 0}^{1/3}g_{m} (z)dz \\
\end{eqnarray}
Now we evaluate the integral $\int_{z = 0}^{1/3}g_{m} (z) dz$. We consider two cases, one for $m = 0$ and one for $m > 0$:
\begin{enumerate}
\item $m = 0$:

In this case $g_0(z) = (1 + 3z)^{n-1}$. Therefore
\begin{align}
\int_{z = 0}^{1/3}g_{0} (z)dz =& x^{(0)}(0) \int_{z = 0}^{1/3} (1 + 3z)^{n-1}dz
\\ &= x^{(0)}(0) \frac{2^{n}}{3 \cdot n}
\\ &= \frac{4 }{2^{n} \cdot 3 n} \text{using Equation \eqref{eq:x1}}
\end{align}

\item $m > 0$:

In this case we give an upper bound on the integral
\bea
\int_{z = 0}^{1/3}g_{m} (z)dz &=& x^{(m)}(0)  \int_{z = 0}^{1/3}(1+3z)^{n-m-1} (1-z)^{m}dz\\
&\leq & x^{(m)}(0)  2^{n-1} \int_{z = 0}^{1/3} (\frac{1-z}{1+3z})^{m}dz\\
&\leq& x^{(m)}(0)  2^{n-1} \int_{z = 0}^{1/3} (1-z)^{m}dz\\
&\leq& x^{(m)}(0)  \frac{2^{n-1}}{m+1}
\eea 
\end{enumerate}

As a result we conclude that
\be
\alpha_m \leq 
\begin{cases}
1+ \frac 1 {2^{n-1}} & m=0\\
x^{(m)}(0)  4^n \frac{3n}{4(m+1)} & m > 0\\
\end{cases}
\label{eq:alpham}
\ee

The last step is to evaluate $x^{(m)}(0)$. In order to do this we need some insight from a well studied class of polynomials known as the Krawtchouk polynomials.
It turns out the Krawtchouk polynomial naturally appears in the expansion of
$(1+3z)^{n-m-1}(1-z)^m$ as the coefficients of $z$ monomials. 
The degree-$t$ Krawtchouk polynomial is defined as:
\be 
K^{(t)} (x) := \sum_{i=0}^{t}  {x \choose i} {n-x-1 \choose t-i}   3^{t-i} (-1)^i.
\ee
(Elsewhere in the literature the Krawtchouk polynomials have been defined with the 3 above
replaced by either 1 or some other number.)
Now we evaluate the coordinates in each $x^{(m)}$ vector. 
\begin{eqnarray}
(1+3z)^{n-m-1}(1-z)^m &=& \sum_{i=0}^{n-m-1} {n-m-1 \choose i} 3^i z^i \sum_{j=0}^m {m \choose j} (-1)^j z^j,\nonumber \\
&=& \sum_{i=0}^{n-m-1} \sum_{j=0}^m {n-m-1 \choose i}   {m \choose j} 3^i (-1)^j z^{i+j}\nonumber \\
&=& \sum_{t=0}^{n-1} z^t \sum_{i=0}^{t}  {m \choose i} {n-m-1 \choose t-i}   3^{t-i} (-1)^i,\nonumber \\
&=:& \sum_{t=0}^{n-1} z^t K^{(t)} (m).
\end{eqnarray}

Hence these Krawtchouk polynomials define the eigenstates, up to overall
normalization, according to
\be 
x^{(m)}(i) = x^{(m)}(0) K^{(i)} (m).
 \ee
Moreover using the orthogonality of the $x^{(m)}$'s, we have
\be 
(4^n-1) {x^{(m)}(0)}^2 \sum_{t=0}^{n-1} \dfrac{{K^{(t)}(m)}^2}{{n \choose t} 3^t} =1.
 \ee
 

 In order to compute $x^{(m)}(0)$ we prove the following proposition.
\begin{proposition}
$\sum_{t=0}^{n-1} \dfrac{{K^{(t)}(m)}^2}{{n-1 \choose t} 3^t} = \dfrac{4^n}{{n-1 \choose m} 3^m}$.
\label{prop:mainortho}
\end{proposition}

Proving this will require two lemmas that establish symmetry and orthogonality properties of Krawtchouk polynomials.

\begin{lemma}
[Orthogonality] \torestate{If we define
\be 
k^{(t)}(x) := \sum_{i=0}^t {x \choose i}{N-x \choose t-i} p^{t-i} (-q)^i,
 \ee
for $p,q \in [0,1]$ and $p+q=1$. Then these Krawtchouk polynomials satisfy the following orthogonality relationship
\be 
\sum_{x=0}^n {N\choose x} p^x q^{N-x} k^{(t)} (x) k^{(s)} (x) = {N \choose t} (pq)^t \delta_{t,s}.
 \ee
\label{lem:orthogonality}}
\end{lemma}

\begin{lemma}[Symmetry]\torestate{ The Krawtchouk polynomials obey the following symmetry relation.
\be 
\dfrac{{n-1\choose x}}{3^t}K^{(t)}(x) = \dfrac{{n-1\choose t}}{3^x}K^{(x)}(t).
 \ee
\label{lem:symmetry}}
\end{lemma}
These two lemma are proved in appendix \ref{section:krawtchuk}.

\begin{proof}[Proof of \propref{mainortho}]
Using Lemma \ref{lem:orthogonality}, setting $p=3/4$ and $q=1/4$, and $N=n-1$ and $t=s$, we have
\be 
4^t k^{(t)} (x) =\sum_{i=0}^t {x \choose i}{n-x-1 \choose t-i} 3^{t-i} (-1)^i = K^{(t)}(x),
 \ee
Therefore we obtain the relation
\be 
\sum_{x=0}^n {n-1 \choose x} 3^x {K^{(t)}(x)}^2 = {n-1 \choose t}3^t 4^{n-1}.
 \ee

We now use the symmetry from Lemma \ref{lem:symmetry} to obtain
\be 
K^{(t)}(x) = \dfrac{3^t}{3^x} \frac{{n-1\choose t}}{{n-1\choose x}}K^{(x)}(t).
 \ee
As a result
\be 
\sum_{x=0}^n  \dfrac{ {K^{(x)}(t)}^2}{{3^x}{n-1\choose x}} = \frac{4^{n-1}}{{n-1 \choose t}3^t}. 
 \ee
This concludes the proof.
\end{proof}

A corollary of Proposition \ref{prop:mainortho} is that
\be 
{x^{(m)}(0)}  = \frac{1}{(4^n-1)} \sqrt{{n-1 \choose m}3^m}.
\label{x0}
\ee
Plugging this into Equation \eqref{eq:alpham} we get
\be
\alpha_m \leq
\begin{cases}
2 & m=0\\
\sqrt{{n-1 \choose m}3^m} \frac{3n}{2(m+1)} & m > 0\\
\end{cases}
\label{eq:alpha}
\ee

%

Now we are ready to prove Theorem \ref{thm:acc-mix}.

\begin{proof}[Proof of Theorem \ref{thm:acc-mix}]
Using Equations \eqref{eq:errr} and \eqref{eq:alpha}
\bea
\|Q_t\|_{\text{\Hsquare}} &\leq& \frac {12} {2^n} \sum_{m=0}^{n-1} \alpha_m^2 \lambda_m^t\\
&\leq& \frac {24} {2^n} +  \frac {12} {2^n} \sum_{m=1}^{n-1} \alpha_m^2 \lambda_m^t\\
&\leq& \frac {24} {2^n} +  \frac {12} {2^n} \sum_{m=1}^{n-1} \left (\sqrt{{n-1 \choose m}3^m} \frac{3n}{2(m+1)}\right)^2 \lambda_m^t\\
&\leq& \frac {24} {2^n} +  \frac {27n^2 } {2^n} \sum_{m=1}^{n-1} {n-1 \choose m}3^m \lambda_m^t\\
&\leq& \frac {24} {2^n} +  \frac {27n^2 } {2^n} \sum_{m=1}^{n-1} {n-1 \choose m}\left(3 e^{- \dfrac{ 4 t}{3n-1}}\right)^m\\
&\leq& \frac {24} {2^n} +  \frac {27n^2 } {2^n} \sum_{m=1}^{n-1} {n-1 \choose m}\left(3 e^{- 4 n \ln n} \right)^m\\
&\leq& \frac {24} {2^n} (1 + O(\frac 1 n))
\eea
\end{proof}

\subsubsection{Proof of Proposition \ref{prop:combining}: Combining wait-time analysis with the analysis of the accelerated chain}
\label{sec:propcombining}

\restateprop{prop:combining}

\begin{proof}[\OldNormalFont{} Proof of Proposition \ref{prop:combining}]
Let $\tau \sim \mathsf{Unif} (t_1, t_2)$. Then 
\be
\frac 1 T \sum_{s = t_1}^{t_2} \|P_s\|_\ast = \E_\tau \|P_\tau\|_\ast
\ee
We use the notation $y^s = (y_1, \ldots, y_s)$, for $y_j$ running over
$[n]$. Consider the event $\{X_\tau = k\}$. This event is equivalent
to the disjoint union $\cup_{s \geq 0} \cup_{y^s \in [n]^s : y_s = k}
\{Y^s=y^s\}\cap \{W_{s-1} < \tau \leq W_s\}$. Here $y_0 \sim \Bin (n,1/2)$,
conditioned on $y_0\neq 0$. 
Therefore
\bea
\nonumber
\Pr [X_\tau = k] &=& \sum_{s \geq 0} \sum_{y^s: y_s = k} \Pr [Y^s =
y^s] \Pr [W_{s-1}<  \tau \leq W_s]\\ 
\nonumber
&=& \sum_{0 \leq s < t_0} \sum_{y^s: y_s = k} \Pr [Y^s = y^s]  \Pr [W_{s-1}<  \tau \leq W_s]\\ 
&& +\sum_{t_0 \leq s } \sum_{y^s: y_s = k} \Pr [Y^s = y^s]  \Pr [W_{s-1}<  \tau \leq W_s].
\label{eq:breakingintoterms}
\eea

We first argue about the time average of the first term.
\begin{align}
\nonumber
 \E_\tau \sum_{0 \leq s <t_0} \sum_{y^s: y_s = k}\Pr [Y^s =
y^s] &\Pr [W_{s-1}<  \tau \leq W_s] \leq \E_\tau \sum_{0 \leq s <t_0} \sum_{y^s: y_s = k}\Pr [Y^s =
y^s] \Pr [W_s \geq \tau] \\ 
\nonumber
  &= \E_\tau\sum_{0 \leq s <t_0} \sum_{y^s: y_s = k} \Pr [Y^s =
  y^s]\Pr [s + T_{\text{left}(y^s)} - T_{\text{right}(y^s)} \geq
  \tau]\\ 
\nonumber
    &\leq \E_\tau \sum_{0 \leq s <t_0} \sum_{y^s: y_s = k} \Pr [Y^s =
    y^s] \Pr [T_{\text{left}(y^s)} \geq \tau-s]\\ 
\nonumber
       &\leq  \sum_{0 \leq s <t_0} \sum_{y^s: y_s = k} \Pr [Y^s =
        y^s]  \Pr [T_{\text{left}(y^s)} \geq t_1 - t_0]\\ 
  &\leq t_0 \cdot \Pr[T_{\text{left}}(t_0) \geq t_1-t_0].
  \label{eq:waitthing}
\end{align}
In the last step we have used the fact that $T_{\text{left}(y^s)}$ is
a nondecreasing function of $s$.  To bound the contribution to the $\|
\cdot \|_\ast$ norm, observe that $\| (1,1,\ldots,1) \|_\ast = 1/3
+ 1/3^2 + \ldots \leq 1/2$.  Thus the contribution from the first term
is $\leq\frac{t_0}{2} \cdot \Pr [T_{\text{left}(y^{t_0})} \geq t_1 - t_0]$.

 Next we argue about the time average of the second term in \eqref{eq:breakingintoterms}.
\ba
\label{eq:firstterm}
\sum_{\substack{s \geq t_0\\ y^s: y_s = k}}
 \Pr [Y^s = y^s] \E_\tau  \Pr [W_{s-1}<  \tau \leq W_s] &\leq
  \sum_{\substack{0 \leq s \leq  4 t_2\\ y^s: y_s = k}} \Pr [Y^s = y^s] \E_\tau  \Pr [W_{s-1}<  \tau \leq W_s]. &\text{(part i)}\\
&+  \sum_{s > 4 t_2} \max_{y^s: y_s = k} \E_\tau  \Pr [W_{s-1}<  \tau \leq W_s]. &\text{(part ii)}
\label{eq:secondterm}
\ea

We now analyze each part independently
\begin{enumerate}
\item[(part i)] Write
\be
\E_\tau  \Pr [W_{s-1}<  \tau \leq W_s]  = \E_W \E_\tau  I [W_{s-1}<  \tau \leq W_s]. 
\ee
Here $\E_W$ is the expectation value over wait times $W_{y_1}, \ldots, W_{y_s}$, and $I[W_{s-1}<  \tau \leq W_s]$ is the indicator of the event $W_{s-1}<  \tau \leq W_s$. 

We first bound $\E_W \E_\tau  I [W_{s-1}<  \tau \leq W_s]$.  Fix $y^s$ such that
$y_s=k$, and for $a\leq b$ integers, let $[a,b]$ denote the set
$\{a,a+1,\ldots, b\}$.  Then
\ba\nonumber
\E_W \E_\tau  I [W_{s-1}<  \tau \leq W_s],
&=
\E_W \frac{|[t_1,t_2]\cap [W_{s-1},W_s]|}{T},
\\ 
& \leq  \E_W \frac{|[W_{s-1},W_s]|}{T}.
\label{eq:interval}
\ea
There are two possibilities for the random variable
$|[W_{s-1},W_s]| = W_s - W_{s-1}$; one for $k < \frac 56n$ and one for $k \geq
\frac 56n$:
\be
W_s - W_{s-1} \sim
\begin{cases}
\mathsf{Geo}({1-\alpha(k)}) & k < \frac 56n\\
\mathsf{Bern} (\beta_k) & k \geq \frac 56n
\end{cases}
\ee
Therefore
\ba
\nonumber
\E_W [W_s - W_{s-1}]  &\leq \mathsf{Geo}({1-\alpha(k)}) + \mathsf{Bern} (\beta_k)\\ 
\nonumber
& \leq \frac {5 n (n-1)}{2k (3n-1)} + 1/2\\
&\leq \frac{3n} k.
\ea
Using this in \eqref{eq:interval} and \eqref{eq:firstterm} we find the bound
\ba
 \sum_{\substack{0 \leq s \leq  4 t_2\\ y^s: y_s = k}} \Pr [Y^s = y^s] \E_\tau  \Pr [W_{s-1}<  \tau \leq W_s] &\leq  \sum_{\substack{0 \leq s \leq  4 t_2\\ y^s: y_s = k}} \Pr [Y^s = y^s] \frac {3n}{kT}
 &\leq  \sum_{t_0 \leq s \leq  4 t_2} \Pr [Y_s = k] \frac {3n}{kT}
 \label{eq:somebound}
\ea

\item [(part ii)] For the second part we use
\ba
\nonumber
\sum_{s > 4 t_2} \max_{y^s: y_s = k} \E_\tau  \Pr [W_{s-1}<  \tau \leq W_s] &\leq \sum_{s > 4 t_2} \max_{y^s: y_s = k} \E_\tau  \Pr [W_{s-1}<  \tau]\\
\nonumber
&\leq \sum_{s > 4 t_2} \max_{y^s: y_s = k} \max_{t_1 \leq \tau \leq t_2}  \Pr [W_{s-1}<  \tau]\\
\nonumber
&\leq \sum_{s > 4 t_2} \max_{y^s: y_s = k} \max_{t_1 \leq \tau \leq t_2}  \Pr [s-1 + T_{\text{left}(y^s)} < \tau + T_{\text{right}(y^s)}]\\
\nonumber
&\leq \sum_{s > 4 t_2} \max_{y^s: y_s = k} \max_{t_1 \leq \tau \leq t_2}  \Pr [s-1 -  \tau <T_{\text{right}(y^s)}]\\
&\leq \sum_{s > 4 t_2} \max_{y^s: y_s = k} \Pr [s-1 -  t_2 <T_{\text{right}(y^s)}]
\label{eq:stail}
\ea

Now recall from Equation \eqref{eq:betabern} we know that $T_{\text{right}(y^s)}$ is statistically dominated by $\mathsf{Bin}(s,1/2)$. So the RHS of \eqref{eq:stail} gets bounded by:
\ba
\nonumber
&\leq  \sum_{s > 4 t_2} \Pr [s -  t_2 \leq \mathsf{Bin} (s,1/2)]\\
\nonumber
&\leq  \sum_{s > 4 t_2} \sum_{k = s -t_2}^s \frac{{s \choose k}}{2^s}\\
\nonumber
&\leq  \sum_{s > 4 t_2} t_2 \frac{{s \choose t_2}}{2^s} &\hspace{-2cm}\text{(using }s > 4 t_2)\\
\nonumber
&\leq  \sum_{s > 4 t_2} t_2 \frac{{s \choose s/4}}{2^s} &\hspace{-2cm}\text{(using }s > 4 t_2)\\
\nonumber
&\leq  \sum_{s > 4 t_2} t_2 \frac{(4 e)^{s/4}}{2^s}\\
\nonumber
&\leq  t_2 \cdot \sum_{s > 4 t_2}  \frac{1}{1.09^s}\\
&\leq  12 t_2  \frac{1}{1.4^{t_2}}
\label{eq:nottoolong}
\ea
\end{enumerate}

Using \eqref{eq:somebound}, \eqref{eq:nottoolong}, \eqref{eq:waitthing} and \eqref{eq:breakingintoterms}
\be
\E_{\tau} \Pr [X_\tau = k] \leq  t_0 \cdot \Pr[T_{\text{left}}(t_0) \geq t_1-t_0] + \sum_{t_0 \leq s \leq  4 t_2} \Pr [Y_s = k] \frac {3n}{kT} + 12 t_2  \frac{1}{1.4^{t_2}}
\ee
Therefore
\be
\E_{\tau} \|P_\tau\|_* \leq  \frac{t_0}{2} \cdot
\Pr[T_{\text{left}}(t_0) \geq t_1-t_0] + 
\frac 1T \sum_{t_0 \leq s \leq  4 t_2}\|Q_s\|_{\text{\Hsquare}} + 6 t_2  \frac{1}{1.4^{t_2}}.
\ee

\end{proof}

\subsection{Towards exact constants}
\label{sec:constants}
Here we discuss what constant factors we may expect from the bound in Theorem \ref{thm:cg}. We do not consider the case of $D$-dimensional graphs here.

What is the right time scale in order to get anti-concentration?  Since Pauli strings of
weight $k$ have contribution $1/3^k$ as well as expected wait-time of $\approx n/k$, it
seems reasonable to guess that lower values of $k$ contribute more to the
anti-concentration probability.  On the other hand, the initial distribution of $k$ is
centered around $n/2$.  Still, enough probability mass survives at low values of $k$ to
yield a non-trivial lower bound in \thmref{cg}.

Thus, let us focus initially on walks starting with weight $k=1$.
Here the expected ``escape time'' from the low-$k$ sector (say to
$k=n/2$) is $\approx \frac{5}{6} n \ln n$, and, simultaneously, it
takes $\approx \frac{5}{6} n \ln n$ time to hit $\frac 34 n -o(n)$.
This is the basis for the following conjecture. A special case of this conjecture for anti-concentration was recently resolved in \cite{DHJB20}.

\begin{conjecture}
\OldNormalFont{} 
If $t = \frac 56 n \ln n + o(n \ln n)$ then $\Pr_{C \sim \mu^{(\text{CG})}_t}   [|\bra{x}C\ket{0}|^2 \geq \frac{\alpha}{2^n} ] = \Omega(1)$.
\label{conj:exact-constant}
\end{conjecture}

Here is the reasoning behind this conjecture.  Recall that the transition matrix $P$ is a
birth-death chain, with probability of moving forward, backwards, and staying put being
$p_l$, $q_l$ and $r_l$, respectively. Let $\pi$ be the stationary distribution. Let $T_l=
\min \{t: X_t \geq l\}$ be the time of hitting the chain site $l$ starting from the first site.
 For any birth-death chain, starting at site $l-1$ \cite{APL09}, the expected time of moving one step forward is
\be 
\bbE_{l-1} ( T_l) = \frac{1}{q_l} \sum_{i=1}^{l-1} \dfrac{\pi(i)}{\pi(l)}.
 \ee
In our chain
\be 
\bbE_{l-1} ( T_l)=\frac{5}{2} \sum_{i=1}^{l-1} \dfrac{{n \choose i}}{{n-2 \choose l-2} 3^{l-i}}.
 \ee

In order to bound this we use the inequalities (proven in \cite{BP17})
\be 
{n-2 \choose l-2} \leq {n-2 \choose i-1} \left (\dfrac{n-i-1}{i} \right)^{l-i-1},
 \ee
\noindent and
\be 
{n \choose i} \leq {n \choose l-1} \left (\dfrac{l-1}{n-l+2} \right)^{l-i-1}.
 \ee
Therefore
\begin{eqnarray}
\label{eq:e_l}
\bbE_{l-1} ( T_l) &\leq& \frac{5}{6} \sum_{i=1}^{l-1} \dfrac{{n \choose l-1}}{{n-2
                               \choose l-2} }
\left (\dfrac{l-1}{3(n-l+2)}\right )^{l-i-1}, \\
\label{eq:e_l-simplified}
&\leq& \frac{5}{6}  n  \left(\dfrac{1}{l-1}+\frac{1}{3n/4-l+7/4} \right) \left(1 + O(1/n) \right).
\label{eq:bound-hit-time}\end{eqnarray}
The last line holds for $l < \frac 34 n$. The transition from \eqref{eq:e_l} to
\eqref{eq:e_l-simplified} is directly inspired by Equation (2) of [the arXiv version of] \cite{BP17}.

To bound the time of reaching $\frac 34 n-\delta$ for some $\delta\geq 0$ we sum
\eq{bound-hit-time} over $1 \leq l \leq \frac 34 n-\delta$ and neglect the $1+O(1/n)$ corrections.
\ba
\bbE_1[T_{\frac 34n-\delta}]
& \leq \frac 56 n \left(\sum_{l=1}^{\frac 34n-\delta}  \frac 1\ell +\frac{1}{3n/4-l+11/4} \right)\\
\label{eq:hitting-time-terms}
&\approx \frac 56 n \left(\ln\left(\frac{\frac 34n-\delta}{1}\right) + 
\ln\left(\frac{3n/4+7/4}{\delta + 11/4}\right)\right)\\
& = \frac 56 n\left( \ln\frac{n^2}{\delta+1} + O(1)\right).\ea
Using this bound, for $a < b$ we can also compute $\bbE_{a} [T_b]$ as $\bbE_{1} [T_b] -
\bbE_{1} [T_a]$.  We wish to estimate $\bbE_a[T_b]$ in two main regimes.  Recall that our
starting distribution is $\Bin(n,1/2)$ and the stationary distribution is $\Bin(n,3/4)$.
Thus we need to know the time for most of the probability mass to reach $\approx 3/4n$,
and for the left tail of the initial distribution to reach the left tail of the final
distribution. (The right tail is less demanding and less important, because it does not
have the long wait times and it is suppressed by the $1/3^k$ factors.)  For the bulk of
the probability distribution we use the estimate $\bbE_{n/2} ( T_{3/4 n - O(1)}) \lesssim
\frac 56 n \ln n$.  For the left tail, we use the bound $\bbE_1[T_{0.74n}] \lesssim \frac
56n\ln n$.  In each case the time required is $\frac 56 n\log n + O(n)$.

%
%

\newpage
\section{Alternative proof for anti-concentration of the outputs of random circuits with nearest-neighbor gates on $D$-dimensional lattices}
\label{sec:2Dmain}

\subsection{The $D=2$ case}
In this section we consider a simplified version of $\mu^{\lattice,n}_{2,c,s}$, where $c=1$ and that $K^{(t)}_{\mu^{\lattice,n}_{2,1,s}} = k^s_R k^s_C$. We prove the following:
\begin{theorem} \OldNormalFont{}
If $s = O ( \sqrt n + \ln (1/\eps) )$ then $\mu^{\lattice,n}_{2,1,s}$ satisfies
\be
\E_{C\sim \mu^{\lattice,n}_{2,s}}\Coll (C) \leq \frac{2}{2^n+1} (1+\eps).
\ee
\label{thm:designs2restatement}
\end{theorem}
This result is already established in Theorem \ref{thm:grid}, we give an alternative proof based on a reduction to a classical probabilistic process. This alternative approach may help with the analysis of random circuits on arbitrary graphs.

We use the following two statements
\begin{lemma}[Brand\~ao-Harrow-Horodecki`13 \cite{BHH-designs}] \OldNormalFont{}
Let $t = O(\sqrt{n} + \ln \frac{1}{\eps})$ then
\be
\Channel[g^t_{\text{Rows}}] \preceq  \bigotimes_{i \in \text{Rows}} \Channel[G_i] \cdot (1+\eps ).
\ee
the same holds for $\Channel[g^t_{\text{Columns}}]$.
\label{l:7}
\end{lemma}

\begin{proof}
This result is proved by Brand\~ao-Harrow-Horodecki in \cite{BHH-designs}.
\end{proof}

\begin{proposition}
\OldNormalFont{} Let $K_i$ an $\eps$ approximate $2$-designs on row or column $i \in \{R,C\}$, in the sense that
\be
K_i \preceq (1 + \eps) \Channel [G_i]
\ee
then for any sequence of rows or columns $i_1,\ldots, i_t$
\be
\Coll( K_{i_t} \ldots K_{i_1}) \leq (1+\eps)^t \Coll( \Channel [G_{i_t}] \ldots \Channel[G_{i_1}]).
\ee
\label{prop:10}
\end{proposition}

\begin{proof}
This proposition is proved in Section \ref{sec:comparison-proofs}.
\end{proof}

Putting these together
\be
\Coll(\mu^{\lattice,n}_{2, 1, s}) \leq  ( 1+\eps )^2 \Coll(\Channel[G_R G_C]).
\ee

Therefore our objective is to show that
\begin{proposition} \OldNormalFont{}
\be
\Coll(\Channel[G_R G_C]) \leq \frac{2}{2^n+1} \left(1+\frac 1 {\poly(n)}\right).
\ee
\label{prop:collpalace}
\end{proposition}

%
%
%

\begin{proof}
Using the Markov chain interpretation discussed in Section \ref{objective}, the initial distribution on the chain is
\be
V_0 := \frac{1}{2^n}  (\sigma_0\tensor \sigma_0 + \sum_{\substack{p \in \{0,3\}^n\\
p\neq 0}} \sigma_p\tensor \sigma_p ),
\ee
and after the application a large enough random quantum circuit the distribution converges to 
\be
V^\ast := \frac{1}{2^n} \sigma_0\tensor \sigma_0 +  (1-\frac{1}{2^n} ) \cdot \frac{1}{4^n-1}\sum_{\substack{p \in \{0,1,2,3\}^n\\
p\neq 0}} \sigma_p\tensor \sigma_p,
\ee
and we want to see how fast this convergence happens. 

For clarity, throughout this proof we represent distributions along the full lattice by capital letters (such as $V$) and for individual rows or columns with small letters (such as $v^i$ for distribution $v$ on row or column $i$). Also, for simplicity we write $0$ instead of $\sigma_0\tensor \sigma_0$, and $\sigma_0^i$ for all zeros across row or column $i$.

$V_0$ is separable across any subset of nodes. So the initial distribution along each row or column is exactly
\be
\frac{1}{2^{\sqrt{n}}} (\sigma_0 + \sum_{\substack{p \in \{0,3\}^{\sqrt{n}}\\
p\neq 0}} \sigma_p\tensor \sigma_p ) =: v_0.
\ee

 After one application of $\Channel[G_R]$ each such distributions become
\be
v^\ast := \frac{1}{2^{\sqrt{n}}} \sigma_0\tensor \sigma_0 + (1-\frac{1}{2^{\sqrt{n}}}  ) \frac{1}{4^{\sqrt{n}}-1}\sum_{\substack{p \in \{0,1,2,3\}^{\sqrt{n}}\\
p\neq 0}} \sigma_p\tensor \sigma_p =: \frac{1}{2^{\sqrt{n}}} \sigma_0 +  (1-\frac{1}{2^{\sqrt{n}}}  ) v.
\label{eq:4.34}
\ee
Here we have defined 
\be
v := \frac{1}{4^{\sqrt{n}}-1}\sum_{\substack{p \in \{0,1,2,3\}^{\sqrt{n}}\\
p\neq 0}} \sigma_p\tensor \sigma_p.
\ee
therefore the distribution along the full chain is $V_1 :=  ( \frac{1}{2^{\sqrt{n}}}\sigma_0 + (1-\frac{1}{2^{\sqrt{n}}})v )^{\tensor \sqrt{n}}$. We also use the notation $v^y \sigma_0^{\backslash y} :=  \otimes_{i : y_i = 1} v \otimes \bigotimes_{i : y_i = 0} \sigma_0$, for $y \in \{0,1\}^{\sqrt{n}}$.

Before getting to the analysis, we should first understand the main reason why $\Coll (\Channel[G_R])$ is large.

After we apply $\Channel[G_R]$ the collision probability across each row is exactly $\frac{2}{2^{\sqrt{n}}+1}$. So the collision probability across the whole lattice is $\approx \frac{2^{\sqrt{n}}}{2^n}$; which is much larger (by a factor of $2^{\sqrt{n}}$) than what we want. The crucial observation here is that if in \eqref{eq:4.34} we project out all the $\sigma_0$ terms across each row, then the bound becomes $\approx \frac{1}{2^n}$. So what really slows this process are the zero $\sigma_0$ terms. The issue is that, after an application of $\Channel[G_R]$, all zeros states get projected to themselves. However, if one applies $\Channel[G_C]$ they get partially mix with other rows. So the objective is to show that after application of $\Channel[G_C] \Channel[G_R] $ for \emph{constant} number of times, these zeros disappear with large enough probability.

Let $V_s$ be the distribution along the full chain after we apply $ (\Channel[G_C] \Channel[G_R] )^s$. Eventually we want to compute
\be
\Coll(\Channel[G_s])=\frac{1}{2^n} \Tr \left(v_0^{\tensor \sqrt{n}}V_s \right) =: \kappa (V_s). 
\ee
Here we have defined the map
\be
\kappa : A \mapsto \frac{1}{2^n} \Tr  ( V_0 A ).
\ee

As a result
\bea
V_1 =  \otimes_{r \in \text{Rows}} \frac{1}{2^{\sqrt{n}}} \sigma_0^r +  (1-\frac{1}{2^{\sqrt{n}}}  ) v^r &=& \sum_{y \in \{0,1\}^{\sqrt{n}}} \frac{1}{2^{\sqrt{n}(\sqrt{n} - |y|)}}  (1-\frac{1}{2^{\sqrt{n}}}  )^{|y|} v^{y} \sigma_0^{\backslash y}.
\label{V_1}
\eea
An important observation here is that
\be
\kappa_i \left(\frac{1}{2^{\sqrt{n}}} \sigma_0^i \right) = \frac{1}{2^{\sqrt{n}}}, \hspace{1.5cm} \kappa  \left((1- \frac{1}{2^{\sqrt{n}}}) v^{ i} \right) =\frac{(1- \frac{1}{2^{\sqrt{n}}})}{2^{\sqrt{n}}+1} < \frac{1}{2^{\sqrt{n}}}.
\ee
the relevant information here is that when $\kappa$ is applied to the summation in \eqref{V_1}, it amounts to 
\be
\kappa(V_1) < \frac{1}{2^n} \sum_{y \in \{0,1\}^{\sqrt{n}}} 1 = \frac{2^{\sqrt{n}}}{2^n}.
\ee
In other words, each $\sigma_0$ term contributes to the number $1$ in the above summation. That means if we had started with the distribution
\be
V' =  \bigotimes_{r \in \text{Rows}}o(1/\sqrt{n}) \frac{1}{2^{\sqrt{n}}} \sigma_0^r +  \left(1-o(1/\sqrt{n}) \frac{1}{2^{\sqrt{n}}} \right) v^r,
\ee
then we would have obtained
\be
\kappa(V') = \frac{2}{2^n+1} \left(1+\frac 1{\poly(n)}\right),
\ee
which is exactly what we want. The last relevant piece of information is that if $v''_j$ is a distribution over row $j$ that with probability $1$ contains a nonzero item, then when $\Channel[G_j]$ is applied to it, it will instantly get mapped to $v_j$. This phenomenon is related to strong stationarity in Markov chain theory.

We claim that after the first application of $\Channel[G_C]$, the expected collision probability is according to the bound claimed in this theorem. In order to see this, we consider the distribution $V_1$ (\eqref{V_1}), this time along each column. Note that the distribution along columns. For any set of columns $j_1, \ldots, j_k$ let $E_{j_1, \ldots, j_k}$ be the event that these columns are all zeros, and the rest of the columns have at least one non-zero element in them. Here we use the notation $E_{j_1, \ldots , j_k} \equiv E_y$ for $y \in \{0,1\}^{\sqrt{n}}$ such that the $j_1, \ldots, j_k$ locations of $y$ are ones and the rest of its bits are zeros.

Therefore
\bea
\Coll (\Channel [G_C] V_1 ) &=& \sum_{y \in \{0,1\}^{\sqrt{n}}} \Pr [E_y ] \kappa \left(\sigma_0^y V^{\backslash y} \right)\nonumber \\
&=& \frac{1}{2^n} + \sum_{y \in \{0,1\}^{\sqrt{n}}\backslash} \Pr [E_y ]  \left(\frac{1}{2^{\sqrt{n}+1}} \right)^{\sqrt{n} - |y|}.\nonumber
\eea
Let $p_0 := \frac{1}{2^{\sqrt{n}}}+ \frac{1}{4}(1- \frac{1}{2^{\sqrt{n}}})$. The main observation is that for each such $y$,
\be
\Pr [E_y ] \leq p_0^{\sqrt{n} |y|} \left(1-p^{\sqrt{n}}_0 \right)^{\sqrt{n}-|y|}.
\ee
Therefore
\bea
\Coll (\Channel[G_C] V_1 ) &\leq& \frac{1}{2^{n}}+\sum_{y \in \{0,1\}^{\sqrt{n}}\backslash 0}p_0^{\sqrt{n} |y|} \left(1-p^{\sqrt{n}}_0 \right)^{\sqrt{n}-|y|}  \left(\frac{1}{2^{\sqrt{n}+1}} \right)^{\sqrt{n} - |y|}\nonumber \\
&=& \frac{1}{2^n}+ \left(p_0^{\sqrt{n}}+ \left(1-p_0^{\sqrt{n}}\right)\frac{1}{2^{\sqrt{n}+1}} \right)^{\sqrt{n}}\nonumber \\
&=& \frac{1}{2^n} + \frac{1-\frac{1}{2^{n}}}{2^n+1} \left(1+\frac 1 {\poly(n)}\right)\nonumber \\
&=&  \frac{2}{2^n+1} \left(1+\frac 1{\poly(n)}\right),
\eea
and this completes the proof.

\end{proof}

\subsection{Generalization to arbitrary $D$-dimensional case}
See Section \ref{sec:definitions} for definitions in this section. In particular, we need definitions for $\Channel[g_i]$, $K_i$ and $\Channel[G_i]$ for each coordinate $i$ of the lattice, and $K_t = (\prod_i k_i)^t$.

In this section we prove that 
\begin{theorem} \OldNormalFont{}
$D$-dimensional $O(D n^{1/D} + D \ln (\frac{D}{\eps}))$-depth random circuits on $n$ qubits have expected collision probability $\frac{2}{2^n+1} \left(1+\frac 1 {\poly(n)}\right)$.
\label{thm:52}
\end{theorem}

\begin{proof} The proof is basically a generalization of the proof for Theorem \ref{thm:designs2restatement}. Here we sketch an outline  and avoid repeating details. In particular, we need generalizations of
Lemma \ref{l:7} and Proposition \ref{prop:10}


The generalization of Lemma \ref{l:7} is simply that $k_i^t$ for $t = O(n^{1/D} + \ln \frac{D}{\eps})$ is an $\frac{\eps}{d}$-approximate $2$-design.
Proposition \ref{prop:10} naturally generalizes to: if for each coordinate $K_i$ is an $\frac{\eps}{D}$-approximate $2$-design then
\be
\Coll  \left(\prod_i K_i \right) \leq \left(1+\frac{\eps}{D}\right)^{D}  \cdot \Coll  \left(\prod_i \Channel [G_i]  \right).
\ee

Our objective is then to show
\be
\Coll  \left(\prod_i \Channel[G_i]  \right) = \frac{2}{2^{n}+1} \left(1+\frac 1 {\poly(n)}\right).
\ee
This last step may be the most non-trivial part in this proof.

Here we just outline the proof. For detailed discussions see the proof of Proposition \ref{prop:collpalace}. We first separate the all zeros state of the chain which contributes as $1/2^{n}$ to the expected collision probability. After the application of $G_1$ on the first coordinate, each row in this coordinate, will be all zeros vector with probability $1/2^{n^{1/D}}$ and $V$ with probability $1-1/2^{n^{1/D}}$. After the application of $G_2$ each plane in the direction $1,2$ will be all zeros with probability $\approx 1/2^{2 n^{1/D}}$ and $V$ with probability $\approx 1- 1/2^{2 n^{1/D}}$. After the application of $G_3$ each plane in $1,2,3$ direction is all zeros with probability $\approx 1/2^{3 n^{1/D}}$ and $V$ otherwise, and so on. Eventually after the application of $G_d$ the distribution along the chain is all zeros with probability $\approx 1/2^{D n^{1/D}}$ and $V$ otherwise. At this point the distribution along each individual row in each coordinate is $\approx 1/2^{D n^{1/D}} 0 + (1-1/2^{D n^{1/D}}) V$. So the collision probability across each such row is 
\be
\approx \frac{1}{2^{D n^{1/D}}} +\frac{1}{2^{n^{1/D}}}.
\ee
Therefore the collision probability across the full chain is
\be
\approx \frac{1}{2^n} +  \left(\frac{1}{2^{D n^{1/D}}} +\frac{1}{2^{n^{1/D}}} \right)^{n^{1-1/D}} \approx \frac{1}{2^n} + \frac{1}{2^{n}} \exp \left(\frac{1}{2^{d n^{1/D}}} n^{1-1/D}\right).
\ee

\end{proof}

%
%
%

\begin{corollary} \OldNormalFont{}
$O(\ln n \ln \ln n)$-depth random circuits with long-range gates have expected collision probability $\frac{2}{2^n+1} \left(1+\frac 1 {\poly(n)}\right)$.
\end{corollary}
\begin{proof}
Set $D = \ln n$ in Theorem \ref{thm:52}.
\end{proof}

\newpage

\section{Scrambling and decoupling with random quantum circuits}

\label{sec:scr}

In this section we reconstruct some of the results of Brown and Fawzi \cite{BF13,BF13-2}. The paper \cite{BF13-2} proves random circuit depth bounds required for scrambling and some weak notions of decoupling. We are able to use our proof technique to reconstruct and improve on the results of this paper. \cite{BF13} on the other hand introduces a stronger notion of decoupling with random circuits. Unfortunately our method does not seem to yield any results about this model.

We first define an approximate scrambler based on \cite{BF13-2}.
\begin{definition}
[Scramblers]\OldNormalFont{} $\mu$ is an $\eps$-approximate scrambler if for any density matrix $\rho$ and subset $S$ of qubits with $|S| \leq n/3 $
\be
\E_{C \sim \mu} \|\rho_S(C) - \frac{I}{2^{|S|}} \|^2_1 \leq \eps.
\ee
where $\rho_S(C) = \Tr_{\backslash S}C \rho C^\dagger$ and $\Tr_{\backslash S}$ is trace over the subset of qubits that is complimentary to $S$.
\end{definition}

We show that small depth circuits from $\mu^{\lattice,n}_{D,c,s}$ are good scramblers.

\begin{theorem} \OldNormalFont{} If $s = O(D \cdot n^{1/D} + \ln D)$ and $c= 1$ then $\mu^{\lattice,n}_{D,c,s}$ is a $\frac 1{\poly(n)}$-approximate scrambler. In particular, for $D = O(\ln n)$ this corresponds to an ensemble of $O(\ln n \ln \ln n)$ depth circuits that are $\frac 1 {\poly(n)}$-approximate scramblers.
\label{thm:54}
\end{theorem}

Brown and Fawzi show a circuit depth bound of $O(\ln^2 n)$ for random circuits with long-range interactions. Our result improves this to $O(\ln n \ln \ln n)$ depth. We believe that the right bound should be $O(\ln n)$. Moreover, no bound for the case of $D$-dimensional lattices was mentioned in their result. 

\begin{proof}
We first rewrite $\E_{C \sim \mu} \|\rho_S(C) - \frac{I}{2^{|S|}} \|^2_1 \leq 2^{|S|} \E_{C \sim \mu} \Tr (\rho^2_S(C))-1$ (to see why this is true see \cite{BF13-2}). Next, consider an arbitrary density matrix
\be
\rho= \sum_{i,j} \rho_{i,j} \ket{i}\bra{j}.
\ee
We first find an expression for $\Tr _{\backslash S} (C \rho C^{\dagger})$
\bea
\Tr _{\backslash S} (C \rho C^{\dagger}) &=& \sum_{i,j} \rho_{i,j} \Tr _{\backslash S} (C \ket{i}\bra{j} C^{\dagger})\nonumber \\
&=& \sum_{i,j} \rho_{i,j} \Tr _{\backslash S} \sum_{g, h} C_{ig} C^\ast_{jh} \ket{g}\bra{h}\nonumber \\
&=& \sum_{i,j} \rho_{i,j} \sum_{\tilde{g}, \tilde{h}} \sum_p C_{i,\tilde{g};p} C^\ast_{j,\tilde{h};p} \ket{\tilde{g}}\bra{\tilde{h}}.
\eea

Therefore
\bea
\E_{C\sim \mu^{\lattice,n}_{D,c,s}}\Tr_S \left(\Tr _{\backslash S} \left (C \rho C^{\dagger}\right)\right)^2 &=& \E_{C\sim \mu^{\lattice,n}_{D,c,s}} \Tr_S\left (\sum_{i,j} \rho_{i,j}  \sum_{\tilde{g}, \tilde{h}} \sum_p C_{i,\tilde{g};p} C^\ast_{j,\tilde{h};p} \ket{\tilde{g}}\bra{\tilde{h}} \right)^2\nonumber \\
&=& \E_{C\sim \mu^{\lattice,n}_{D,c,s}}\sum_{i,j} \sum_{k,l} \sum_{\tilde{g}_1, \tilde{h}_1} \sum_{\tilde{g}_2, \tilde{h}_2} \sum_{p,q}\nonumber \\
&&\rho_{i,j}\rho_{kl}  C_{i,\tilde{g}_1;p} C^\ast_{j,\tilde{h}_1;p}C_{i,\tilde{g}_2;q} C^\ast_{j,\tilde{h}_2;q} \delta_{\tilde{h}_1 = \tilde{g}_2}\delta_{\tilde{h}_2 = \tilde{g}_1}\nonumber \\
&=&\E_{C\sim \mu^{\lattice,n}_{D,c,s}} \sum_{ij,k,l} \sum_{a,b,c,d}\rho_{i,j}\rho_{kl}  C_{i, a ;b} C^\ast_{j,c;b}C_{i,c;d} C^\ast_{j,a;d}\nonumber \\
&=& \Tr  \left (\rho \tensor \rho \Channel \left [G^{(2)}_{\mu^{\lattice,n}_{D,c,s}}\right ] \left (\sum_{a,b,c,d}  \ket{ab}\bra{cb} \tensor \ket{cd}\bra{ad}\right ) \right)\nonumber \\
&=& \Tr  \left(\rho \tensor \rho \Channel \left [ G^{(2)}_{\mu^{\lattice,n}_{D,c,s}}\right] (A ) \right).
\eea
both $\rho \tensor \rho$ and $A$ are psd therefore using Lemma \ref{lem:comparison}
\be
\Tr  \left(\rho \tensor \rho \Channel[G_\mu^{(2)}] (A ) \right) \leq (1+\eps)^{D} \cdot \Tr \left (\rho \tensor \rho \prod_{1\leq i \leq D} \Channel [G_i]  (A ) \right).
\ee

Next, using Equation 3 of \cite{BF13-2} we reduce computation of $\Tr  (\rho \tensor \rho \prod_{1\leq i \leq D} \Channel[G_i]  (A ) )$ to the following probabilistic process: starting from a uniform distribution over $\{0,3\}^{n} \backslash I^n$ show that the probability that after the application $\prod_{1\leq i \leq D} \Channel[G_i]$ the string on Markov chain $K$ defined in Section \ref{objective} has weight $\leq n/3$ is $\poly(n)/2^{n}$ and this reconstructs theorem A.1 of \cite{BF13-2}. 

The initial state on the chain is $\frac{1}{2^n}\sum_{p \in \{0,3\}^n\backslash 00} \sigma_p \tensor \sigma_p$ we add the term $\frac{1}{2^n} \sigma_0\tensor \sigma_0$ this can only slower the process. With this modification each site is initially independently $Z \tensor Z$ or $I \tensor I$, each with probability $1/2$.

From using the proof of Theorem \ref{thm:52} after the application $\prod_i \Channel[G_i]$ the distribution along the each row  is $\approx 1/2^{D n^{1/D}} \sigma_0 \tensor \sigma_0 + (1-1/2^{D n^{1/D}}) V$. Therefore the probability that each site is zero is at most $1/4 + 1/2^{D n^{1/D}} =: 1/4 + \delta =: p_0$. Hence the probability of having at most $n/3$ is at most
\be
\sum_{k=1}^{n/3} \frac{{n \choose k}}{4^n-1} p_0^{n-k} \left (1-p_0\right)^{k} = \sum_{k=1}^{n/3} \frac{{n \choose k}}{4^n-1} \left (1/4 + \delta\right)^{n-k} \left (3/4-\delta\right)^{k} \leq  e^{4 \cdot 2/3 n \cdot \delta} \sum_{k=1}^{n/3} \frac{{n \choose k}}{4^n-1} 1/4^{n-k} (3/4)^{k}
\ee
which is within $1 + O\left (n/2^{D n^{1/D}}\right)$ of what we would expect from the Haar measure. Also when $D = O(\ln n)$ with a proper constant, this value is $1 + 1/ \poly(n)$.
\end{proof}

Next, we consider the following notion of decoupling defined in \cite{BF13-2}. Consider a maximally entangled state $\Phi_{MM'}$ along equally sized systems $M$ and $M'$ each with $m$ qubits, and a pair of equally sized systems $A$ and $A'$. Similar to \cite{BF13-2} we consider two models for $AA'$: 1) a pure state $\ket{0}_A\bra{0}$ along system $A$ with $n-m$ qubits and 2) a maximally entangled state $\phi_{AA'}$. We then apply a random circuit to systems $M' A$ and we want that for a small subsystem $S$ of $M$ the final state $\rho_{MS} (t)$ be decoupled in the sense that $\rho_{MS} (t) \approx I/2^{m+s}$.

\begin{definition}[Weak decouplers] \OldNormalFont{} a distribution $\mu$ over $\text U(2^n)$ is an $\eps$-approximate weak decoupler if $\|\rho_{MS} (t) - \frac{I_M}{2^{|M|}} \tensor \frac{I_S}{2^{|S|}}\|_1 \leq \eps$.
\end{definition}

\begin{theorem} \OldNormalFont{} Let $D$ be a constant integer. If $s = O(D \cdot n^{1/D})$ and $c= 1$ then there exists a constant $c'<1$ such that if $m < c' n^{1/D}$ then $\mu^{\lattice,n}_{D,c,s}$ is a $\frac 1 {\poly(n)}$-approximate weak decoupler. 
\end{theorem}

The depth bound Brown and Fawzi find in \cite{BF13-2} for this problem is $n^{1/D} \cdot O(\ln n)$ depth for $m = \poly(n)$. 

\begin{proof} We first show that the bound we want to calculate for the 1-norm in this theorem can be written as $\Tr \left (E \Channel\left [G_{\mu^{\lattice,n}_{D,c,s}}\right] F\right)$ where $E$ and $F$ are psd matrices. Hence using Lemma \ref{lem:comparison} we can use the overlapping projectors $\prod_i \Channel[G_i]$ instead $\Channel[G_{\mu^{\lattice,n}_{D,c,s}}]$ as the second-moment operator.

We first start with the case when $\psi_A$ is the pure state $\ket{0}_A\bra{0}$. The initial state is the (pure) density matrix
\be
\rho_{\mathrm{init}} = \frac{1}{2^m} \sum_{i,j} \ket{i}\bra{j} \tensor \ket{i0}\bra{j0}
\ee
where $\ket{i}$ runs through the computational basis of $M$ and $0$ is the initial state of $A$. After the application of a circuit $C$
\be
\rho_{\mathrm{init}} \mapsto \rho_C = \frac{1}{2^m} \sum_{i,j,k,l} \ket{i}\bra{j} \tensor \ket{k}\bra{l} C_{i0;k}C^\ast_{j0;l}.
\ee
where $C_{a;b}$ is the $ab$ entry of $C$. The density matrix corresponding to subsystem $MS$ becomes
\be
\frac{1}{2^m} \sum_{i,j, k',l', q'} \ket{i}\bra{j} \tensor \ket{k'}\bra{l'} C_{i0;k'q'}C^\ast_{j0;l'q'}.
\ee
We use the bound (also used in \cite{BF13-2})
\be
\|\rho_{MS} (C) - \frac{I_M}{2^{|M|}} \tensor \frac{I_S}{2^{|S|}}\|_1 \leq 2^{m+s} \Tr (\rho^2_{MS}(C))-1.
\label{eq:249}
\ee
Next, using the proof of Theorem \ref{thm:54} $\E_{C\sim \mu}\Tr (\rho^2_{MS}(C))$ can be written as $\Tr (C \Channel[G^{(2)}_\mu] D )$ where $C$ and $D$ are psd, hence $\Tr (C \Channel[G_{\mu^{\lattice,n}_{D,c,s}}^{(2)}] D ) \leq \Tr (C \prod _i \Channel[G_i] D) (1+\eps)$. Hence we can just use $\prod_i \Channel[G_i]$ to bound the expectation $\E_{C \sim \mu^{\lattice,n}_{D,c,s}}\|\rho_{MS} (C) - \frac{I_M}{2^{|M|}} \tensor \frac{I_S}{2^{|S|}}\|_1$. 

Next, we do the same calculation for the case when $\psi_{AA'}$ is the maximally entangled state $\frac{1}{2^{n-m}}\sum_{i,j} \ket{i}\bra{j} \tensor \ket{i}\bra{j}$. Therefore the initial density matrix is
\be
\rho_{\mathrm{init}} = \frac{1}{2^n} \sum_{i,j,k,l} \ket{i}_M\bra{j} \tensor  \ket{i}_{M'}\bra{j} \tensor \ket{k}_A\bra{l} \tensor  \ket{k}_{A'}\bra{l} = \frac{1}{2^n} \sum_{i,j,k,l} \ket{i}_M\bra{j} \tensor  \ket{i k }_{M'A'}\bra{j l} \tensor  \ket{k}_{A'}.\bra{l}
\ee
After the application of the random circuit this gets mapped to
\be
\rho_{\mathrm{init}} \mapsto \rho(C) = \frac{1}{2^n} \sum_{i,j,k,l} \ket{i}_M\bra{j} \tensor  \ket{z }_{M'A'}\bra{w} \tensor  \ket{k}_{A'}\bra{l} C_{ik, z} C^\ast_{jl, w}.
\ee
Again we can use a bound similar to \eqref{eq:249} and similar to the proof of Theorem \ref{thm:54} we can show that tracing out a subsystem, the trace of the resulting density matrix squared can be written as $\Tr \left(C \Channel[G^{(2)}_\mu] D \right)$ for $C$ and $D$ psd.

As proved in theorem 3.5 of \cite{BF13-2}, the task is to show that starting with uniform distribution over all strings with weight $\leq m = O(n^{1/D})$, prove that the probability that after the application of the random circuit the weight of the string on the chain is $\geq n/2$ is at least $1-1/4^m$. It is enough to show that this is true for the initial state with Hamming weight $1$. Without loss of generality assume the nonzero digit in this string is in the first row of the first direction. After the application of $G_1$ the first row in this direction becomes $V$. Using Chernoff bound for independent Bernoulli trials, with probability at least $1- e^{-O(n^{1/D})}$ there are at most $1/4 \cdot n^{1/D} \cdot 2^{1/D}$ zeros on this row. After the application of $G_2$ with probability at least $1- e^{-O(n^{1/D})}$ there are $1/4 \cdot n^{2/D} \cdot 2^{2/D}$, and so on. Hence after the completion of $\prod_i G_i$ with probability at least $(1- e^{-O(n^{1/D})})^D$ there are at most $1/4 \cdot n^{D/D} \cdot 2^{D/D} = n/2$ zeros on the chain. For constant $D$ the failure probability is at most $e^{-O(n^{1/D})}$ and we can choose the constant $c'$ small enough so that if $m < c' n^{1/D}$ the probability of failure is at most $1/4^m$.

\end{proof}

\appendix

\section{Proof of Theorem \ref{thm:supremacy}}
\label{sec:supremacy}
In this section we prove Theorem \ref{thm:supremacy}. The proof is directly inspired by the work of Bremner, Montanaro and Shepherd (see Theorem 6 and 7 of \cite{BMS16}). {A similar theorem was also proved in \cite{HBSE18}.}

\begin{definition} \OldNormalFont{}
Let $\mu$ be a $\frac 1 {\poly(n)}$-approximate $2$-design over the $n$-qubit unitary group. $\mathcal{C}_x$ is the family of unitaries constructed by first applying a circuit $C\sim \mu$ and then sampling an $n$-bit string $x$ uniformly at random, and then applying an $X$ gate to qubit $j$ whenever $x_j =1$.
\label{def:Cx}
\end{definition}

\begin{proof}[\OldNormalFont{} Proof of Theorem \ref{thm:supremacy}] \OldNormalFont{} Let $C$ be a random quantum circuit $\sim \mu$, and define $p_x = |\braket{x|C|0}|^2$. Denote this output distribution with $p_C$. Suppose there exists a $\BPP$ algorithm that samples from a distribution $q_x$ that is within total variation distance $\epsilon$ of $p_x$. Therefore
\be
\sum_x |p_x -q_x| \leq \epsilon.
\ee 
Stockmeyer showed that given a $\BPP$ machine, there exists an 
$\FBPP^{\NP}$ 
algorithm that computes its output probabilities within (inverse polynomial) $\frac 1 {\poly(n)}$ multiplicative error. As a result, there is an $\FBPP^{\NP}$ algorithm that for each string $x$ computes a number $\hat{q}_x$ that satisfies
\be
|{q}_x -\hat q_x| = q_x \cdot \frac 1 {\poly(n)}.
\ee
Therefore using triangle inequality
\be
\sum_x |p_x -\hat{q}_x| \leq \sum_x |p_x -q_x| + \sum_x |q_x -\hat{q}_x|  \leq \epsilon + 1 / \poly(n).
\ee
Let $0 < \delta < 1$. Using Markov's inequality, for at least $1-\delta$ fraction of $n$-bit strings (such as $y$),
\be
|p_y -\hat{q}_y| \leq \frac{\epsilon + 1/\poly(n)}{2^n \delta}.
\label{eq:delta}
\ee


Using the definition of $\mathcal{C}_x$ with probability at least $1-\delta$ over a circuit $C'$ from the family $\mathcal{C}_x$, $p'_0 = |\braket{0|C'|0}|^2$ satisfies \eqref{eq:delta}. Furthermore, we will show below that given a $\frac 1{\poly(n)}$-approximate $2$-design $\mu$, for any output string $y \in \{0,1\}^n$, there exist a constant fraction $\geq 1/8 -\frac 1{\poly(n)}$ of unitaries $C\sim \mu$, such that $p_y \geq 1/2^{n+1}$. 
Therefore w.p. at least $1-\delta$ the $\FBPP^{\NP}$ algorithm computes $\hat {q'}_0$ that satisfies
\be
|\hat{q'}_0-p'_0| \leq \frac{\epsilon + \frac 1{\poly(n)}}{2^n \delta}\leq \frac{2(\epsilon + \frac 1 {\poly(n)})}{\delta} p'_0.
\ee
for $1/8 - \frac 1{\poly(n)}$ fraction of random unitaries $C'$ from the ensemble $\mathcal{C}_x$.
In the last line, we have used \eqref{eq:zyg} (which we are going to prove next).

Now we show that for any output string $y \in \{0,1\}^n$, there exist a constant fraction $\geq 1/8 -\frac{1}{\poly(n)}$ of unitaries $C\sim \mu$, such that $p_y \geq 1/2^{n+1}$. To see this first recall the following known moments of the Haar measure
\be
\E_{C \sim \Haar} |\braket {x|C|0}|^2 = \frac{1}{2^n}, \hspace{1cm} \E_{C \sim \Haar} |\braket {x|C|0}|^4 = \frac{2}{2^n (2^n+1)}.
\ee
Since $\mu$ is a $\frac 1 {\poly(n)}$-approximate $2$-design
\be
\E_{C \sim \mu} |\braket {x|C|0}|^2 = \frac{1+\frac 1 {\poly(n)}}{2^n}, \hspace{1cm} \E_{C \sim \mu} |\braket {x|C|0}|^4 = \frac{2 }{2^n (2^n+1)}\left(1+\frac 1 {\poly(n)}\right).
\ee
Using the Paley-Zygmund inequality and the moments of a $2$-design above
\be
\Pr_{C\sim \mu} \left [ |\braket{x|C|0} |^2 \geq \frac{1}{2^{n+1}} \right] \geq 1/4 \frac{ \left (\E_{C \sim \mu}  |\braket{x|C|0} |^2 \right )^2}{ (\E_{C\sim \mu}  |\braket{x|C|0} |^4 )} =  1/4 \frac{\frac{1+\frac 1 {\poly(n)}}{4^n}}{\frac{2\left(1+\frac1 {\poly(n)}\right)}{2^n \cdot (2^n+1)}} = 1/8 - \frac 1 {\poly(n)}.
\label{eq:zyg}
\ee

\end{proof}

\section {Basic properties of the Krawtchouk polynomials}
\label{section:krawtchuk}

\restatetheorem{lem:symmetry}

\begin{proof} 
This is implied by the observation that for all $i \in [t]$
\be 
{n-1 \choose x} {x \choose i}{n -x -1 \choose t-i} = \dfrac{(n-1)!}{ (x-i)! (t-i)! i! (n-x-t+i-1)!}
 \ee
is symmetric in $x$ and $t$. As a result
\be 
\dfrac{{n-1\choose x}}{3^t}K^{(t)}(x) = \sum_{i=0}^{t}  {n-1\choose x}{x \choose i} {n-x-1 \choose t-i}   3^{-i} (-1)^i.
 \ee
is also symmetric in $x$ and $t$.
\end{proof}

The second lemma we use here is the orthogonality of the Krawtchouk polynomials

\restatelemma{lem:orthogonality}

\begin{proof}
Consider the generating function
\begin{eqnarray}
g_{p,x}(z) &=& (1+ p z)^{N-x} (1- q z)^{x}\nonumber \\
&=& \sum_{i=0}^{n-x} {N -x \choose i} p^i z^i\sum_{j=0}^x {x \choose i} (-q)^j z^j\nonumber \\
&=& \sum_{t=0}^N z^t \sum_{i=0}^{t} {N -x \choose t-i} {x \choose i} p^{t-i} (-q)^i\nonumber \\
&=& \sum_{t=0}^N z^t k^{(t)}(x).
\end{eqnarray}

Define the binomial norm $( \cdot , \cdot ) : \mathcal{F} \times \mathcal{F}  \rightarrow \R$, where $\mathcal{F}$ is the set of functions $: [N]  \rightarrow \R$.
\be 
f,g \mapsto (f,g) := \E_{X \sim \Bin (N,P)}  [ f(X)g(X)  ]= \sum_{x=0}^N {N \choose x} p^x q^{N-x} f(x) g(x).
 \ee

Now for all real values $y$ and $z$ consider the overlap $(g_{p}(y),g_{p}(z))$. On the one hand
\begin{eqnarray}
 (g_{p}(y),g_{p}(z) ) &=& \sum_{x=0}^N {N \choose x} p^x q^{N-x} g_{p,x} (y) g_{p,x} (z),\nonumber \\
&=& \sum_{t,s=0}^N z^{t+s} \sum_{x=0}^N {N \choose x} p^x q^{N-x} k^{(t)}(x) k^{(s)}(x),\nonumber \\
&=& \sum_{t,s=0}^N z^{t+s}  (k^{(t)}, k^{(s)} ).
\end{eqnarray}
On the other hand
\begin{eqnarray}
(g_{p}(y),g_{p}(z)) &=& \sum_{x=0}^N {N \choose x} p^x q^{N-x} g_{p,x} (y) g_{p,x} (z),\nonumber \\
&=& \sum_{x=0}^N {N \choose x} p^x q^{N-x} (1+ p y)^{N-x} (1- q y)^{x} (1+ p z)^{N-x} (1- q z)^{x},\nonumber \\
&=& \sum_{x=0}^N {N \choose x}    (q (1+ p z) (1+ p y) )^{N-x}  (p (1- q y)  (1- q z) )^{x},\nonumber \\
&=&   (q (1+ p z) (1+ p y) + p (1- q y)  (1- q z)  )^N,\nonumber \\
&=&   (1  +qp yz  )^N,\nonumber \\
&=& \sum_{t=0}^N {N \choose t}(pq)^t y^t z^t.
\end{eqnarray}
Equating these two for all $y$ and $z$ we obtain
\be 
 (k^{(t)}, k^{(s)} ) =  \sum_{x=0}^n {N\choose x} p^x q^{N-x} k^{(t)} (x) k^{(s)} (x) = \delta_{t,s} {N \choose t}(pq)^t.
 \ee
\end{proof}

\section*{Acknowledgments}

We are grateful to Scott Aaronson and Yury
Polyanskiy for detailed comments and to Yuval Peres for telling us about \cite{BP17}.
We thank Sami Boulebnane for pointing us to some minor errors in the first version on arXiv. We also thank Matthew Khoury for numerically studying the collision probability and validating several theoretical time-scales proved in this paper. The first draft of this work was completed when SM was affiliated with CSAIL MIT.

\section*{Declarations}

The authors contributed equally to this work. AWH was funded by NSF
grants CCF-1452616, CCF-1729369, PHY-1818914, the NSF QLCI program
through grant number OMA-2016245 and ARO contract W911NF-17-1-0433.
SM was funded by NSF grant CCF-1729369. The authors have no relevant
financial or non-financial interests to disclose.

\bibliographystyle{hyperabbrv}

\end{document}